\documentclass[12pt,a4paper,oneside]{extarticle}

\usepackage{cmap}	
\usepackage[T2A]{fontenc}
\usepackage[utf8]{inputenc}
\usepackage[english]{babel}
\usepackage{setspace}

\usepackage{amsfonts}
\usepackage{amsthm}
\usepackage{amstext}
\usepackage{amsmath}
\usepackage{amssymb}
\usepackage{booktabs}
\usepackage{indentfirst} 
\usepackage{enumerate}
\usepackage[font=small,skip=0pt]{caption}

\usepackage[numbers,square,sort]{natbib} 
\bibpunct{(}{)}{,}{a}{,}{,} 

\usepackage{graphics}
\usepackage[pdftex]{graphicx}
\usepackage{color}
\usepackage{lscape} 
\usepackage{bm} 
\usepackage{tikz}
\usepackage{titlesec} 
\usepackage{float} 
\usepackage{longtable}
\usepackage{booktabs}
\usepackage[center]{subfigure}
\usepackage{captcont}
\usepackage{tabularx}
\usepackage{diagbox} 
\usepackage{rotating}
\usepackage{geometry}

\usepackage[super]{nth}
\usepackage{comment} 
\usepackage{textcomp}

\newtheorem{ass}{Assumption}[section]
\newtheorem{lem}{Lemma}[section]
\newtheorem{thm}[lem]{Theorem}
\newtheorem{cor}[lem]{Corollary}
\newtheorem{prop}[lem]{Proposition}

\theoremstyle{definition}

\newtheorem{example}{Example}[section]
\newtheorem{remark}{Remark}[section]
\renewcommand{\baselinestretch}{1.2}

\renewcommand\appendix{\par
  \setcounter{section}{0}
  \setcounter{subsection}{0}
  \setcounter{figure}{0}
  \setcounter{table}{0}
  \renewcommand\thesection{Appendix \Alph{section}}
  \renewcommand\thefigure{\Alph{section}\arabic{figure}}
  \renewcommand\thetable{\Alph{section}\arabic{table}}
  
}

\newtheorem{proposition}{Proposition}

\newtheorem{lemma}{Lemma}

\begin{document}

\renewcommand{\baselinestretch}{1.24}
\title{New robust inference for predictive regressions\footnote{\footnotesize{We thank the Editor Peter C.B. Phillips, the Co-Editor Anna Mikusheva and three anonymous referees for many helpful comments and suggestions. We are also grateful to Walter Distaso, Jean-Marie Dufour, Siyun He, Nour Meddahi, Mikkel Plagborg-M\o{}ller, Aleksey Min, Ulrich K. M{\" u}ller, Rasmus S. Pedersen, Artem Prokhorov, Rogier Quaedvlieg, Robert Taylor, Alex Maynard, Neil Shephard, Yang Zu, and the participants at the Center for Econometrics and Business Analytics (CEBA, St. Petersburg University) and University of Nottingham seminar series, the session on Econometrics of Time Series at the 12th World Congress of the Econometric Society and iCEBA-2021, 2022 conferences for helpful discussions and comments. Rustam Ibragimov and Anton Skrobotov' research for this paper was supported in part by a grant from the Russian Science Foundation (Project No. 22-18-00588). 
Jihyun Kim is grateful to the French Government and the ANR for support under the Investissements d'Avenir program; Grant ANR-17-EURE-0010. Address correspondence to Jihyun Kim, School of Economics, Sungkyunkwan University, Seoul 03063, South Korea; e-mail: kim.jihyun@skku.edu.}}}
\author{
Rustam Ibragimov$^{a,d}$, Jihyun Kim$^b$, Anton Skrobotov$^{c,d}$ \\
{\small {$^{a}$ Imperial College Business School, Imperial College London}}\\
{\small {$^{b}$ Sungkyunkwan University and Toulouse School of Economics}}\\
{\small {$^{c}$ Russian Presidential Academy of National Economy and Public Administration}}\\
{\small {$^{d}$ Center for Econometrics and Business Analytics, St. Petersburg University}}
}
\date{}
\date{March 23, 2023}
\maketitle


\begin{abstract}
We propose a robust inference method for predictive regression models under heterogeneously persistent volatility as well as endogeneity, persistence, or heavy-tailedness of regressors. 
This approach relies on two methodologies, nonlinear instrumental variable estimation and volatility correction, which are used to deal with the aforementioned characteristics of regressors and volatility, respectively. Our method is simple to implement and is applicable both in the case of continuous and discrete time models. According to our simulation study, the proposed method performs well compared with widely used alternative inference procedures in terms of its finite sample properties in various dependence and persistence settings observed in real-world financial and economic markets. 
\end{abstract}

\medskip



\newpage
\renewcommand{\baselinestretch}{1.4}

\section{Introduction}
Many papers in the literature have focused on econometric analysis of predictive regressions for stock returns (see \citeauthor{Phillips2015}, \citeyear{Phillips2015}, for an up-to-date review). Predictive regression data is known to have several problematic characteristics, especially in statistical inference of stock return predictability. First, it is widely believed that the popular regressors, such as dividend-price and earnings-price ratios, used in the predictive regressions have near unit roots and their innovations are correlated with stock returns in the long run. The characteristics of the regressors, which are persistence and endogeneity jointly cause standard hypothesis tests to become substantially biased (see \cite{stambaugh1999predictive}). Second, there is some evidence that supporting volatility of stock returns is stochastic and highly persistent (see, e.g., \cite{jacquier2004bayesian} and \cite{hansen2014estimating}). \cite{cavaliere2004testing} shows that persistent stochastic volatility may cause substantial size distortions on standard tests developed mostly under the assumption that a volatility process is stationary with a constant unconditional mean, such as stationary GARCH-type models. Lastly, there are several other characteristics of predictive regression data that include heavy-tailedness of regressors as well as jumps, structural breaks and regime switching in volatility. These characteristics may also yield jointly or individually a significant distortion of standard hypothesis tests for predictive regressions.

In this paper, we propose a new method for robust inference on parameters of predictive regression models under the aforementioned characteristics of predictive regression data. Our approach relies on a simple nonlinear instrumental variable (IV) estimation and a nonparametric volatility correction. The nonlinear IV estimator in our approach is an IV estimator with the instrument being the sign transformation of the regressor. This particular IV estimator was first proposed by \cite{cauchy1836lxxviii}, and is called the Cauchy estimator. As is known in the literature, the use of the instrument can effectively eliminate the problems caused by the persistent endogeneity, heavy-tailedness and other problematic characteristics of the regressors (see \cite{SoShin}, \cite{CJP2016} and \cite{kim-meddahi-2019}). On the other hand, volatility correction is used to deal with the problems caused by the presence of heterogeneity and persistence in stock return's volatility. As for the volatility correction, we consider a standard kernel-based nonparametric estimator of volatility.

Many authors have studied the issue of persistent endogeneity of regressors in predictive regressions. Among many of them, \cite{CampbellYogo2006}, \cite{ChenDeo2009} and \cite{PhillipsMagdalinos2009} have proposed tests of return predictability, which are aimed at dealing with persistent and endogenous regressors. Though their tests perform well under the presence of persistent endogeneity of regressors, they are not expected to deal with other problematic characteristics of predictive regression data effectively. Our simulation study shows they have serious size distortions under the null of no predictability when volatility is persistent or incorporates structural breaks or regime switching. In contrast, the robustness of our approach is quite evident. Our approach always yields almost exact sizes in a variety of designs considered in our simulation study. Moreover, the robustness of our approach is obtained with no significant loss of power. The discriminatory powers of our test are comparable to the tests by \cite{CampbellYogo2006} and \cite{ChenDeo2009}, which are optimal for the basic Gaussian model.

Our work is closely related to \cite{CJP2016}, who propose an inference approach for predictive regressions. Similar to our method, their approach relies on the Cauchy estimator to eliminate the problems caused by the problematic characteristics of the regressors. They also use a nonparametric volatility correction. However, their approach for volatility correction is quite different from ours, and its applicability is limited to a predictive regression equipped with appropriate high frequency data since their method and theory are developed in a continuous time framework. More precisely, their volatility correction, called the time change, requires uniformly consistent estimation of a quadratic variation of a stock price for which high frequency observations of the stock price are necessary. Consequently, their approach requires the assumption that the sampling interval decreases to zero, and applications of their method on relatively low frequency data, i.e. monthly or quarterly data, are largely restricted. However, predictive regressions are often estimated using monthly or quarterly data. In contrast, our method can be applied to a discrete time model and a discrete sample collected from an underlying continuous time model as in \cite{CJP2016}. Our simulation study shows that both our method and the method by \cite{CJP2016} perform well and have good size and power performances under continuous time settings considered in this paper. However, unlike our method, the \cite{CJP2016} method is not applicable under discrete time settings. Therefore, we may say that our method is more flexible and widely applicable since it can be applied to both high and low frequency data.

The rest of the paper is organized as follows. Section 2 introduces the predictive regression models, persistent volatility, and the Cauchy estimator. Section 3 proposes the robust inference method and presents its asymptotic properties. Section 4 generalizes the baseline predictive regression models, which have one persistent volatility factor, to have a two-factor volatility, where one factor is persistent and the other is transient such as a stationary GARCH process. Section 5 provides numerical results on finite sample properties of the proposed robust inference approach. Section 6 makes some concluding remarks. 

The online supplementary appendix  provides a discussion of the Cauchy estimator and general nonlinear IV estimators with the relevant asymptotic results that, in particular, point to the importance and usefulness of the Cauchy estimator (Appendix A); useful auxiliary results (Appendix B); the proofs of the main results in the paper (Appendix C); and some additional simulation results on finite sample performance of inference approaches dealt with (Appendix D).

\section{Predictive Regressions}

\subsection{Research Problems and Models}

Throughout the paper, we consider $(\mathcal{F}_t)$-adapted processes defined on a filtered probability space $(\Omega,\mathcal{F}, (\mathcal{F}_t)_{t\geq0}, P)$ equipped with an increasing filtration $(\mathcal{F}_t)$ of sub-$\sigma$-fields of $\mathcal{F}$. We consider a test for no predictability of the process $(y_t)$ (e.g., the time series of excess stock returns) based on some covariate process $(x_t)$ (e.g., the time series of price-to-dividend ratios). We consider the linear predictive regression model
\begin{align}
y_t &= \alpha + \beta x_{t-1} + u_t, \;\; t=1, ..., T,\label{PredRegr1}
\end{align}
where $(u_t)$ is a martingale difference sequence (MDS) with respect to $(\mathcal{F}_t)$. In particular, $(u_t)$ is conditionally heteroskedastistic. Following the usual specification for a volatility model, we assume that
\begin{align*}
u_t = v_t \varepsilon_t,
\end{align*}
where $(v_t)$ is a volatility process and $(\varepsilon_t)$ is an MDS with respect to $(\mathcal{F}_t)$.

\begin{ass}\label{assumption-mds}
(a) $(v_t)$ is $(\mathcal{F}_{t-1})$-adapted and is defined on $[\underline{v},\bar{v}]$ for some $0<\underline{v}<\bar{v}<\infty$, (b) $E(\varepsilon_t^2|\mathcal{F}_{t-1})=1$, and (c) $\sup_{t\geq1}E(\varepsilon_t^4|\mathcal{F}_{t-1})<\infty$.
\end{ass}

The conditions (a)-(b) in Assumption \ref{assumption-mds} are not stringent, and are required for the identification of the conditional variance of $u_t$. In particular, the conditional variance of $u_t$ given $\mathcal{F}_{t-1}$ is well identified and we have $E(u_t^2 | \mathcal{F}_{t-1}) = v_t^2$. Our test relies on uniform convergence results for a nonparametric estimator of the volatility process $(v_t)$. The condition (c) is used to obtain a uniform convergence rate of the nonparametric estimator of the volatility process $(v_t)$. Note that Assumption \ref{assumption-mds} implies $\sup_{t\geq1}E(u_t^4)<\infty$, and hence, it rules out a predictive regression model having a heavy-tailed  regression error $(u_t)$. 

As for a nontrivial example, we let $v_t = f(z_{t-1})$ and $\varepsilon_t \sim iid \mathbb{N}(0,1)$, where $f$ is a positive function and $z_t$ is an $\mathcal{F}_t$-adapted process. Then $(v_t,\varepsilon_t)$ satisfies Assumption 2.1. If we assume that $f$ is bounded above, then $u_t = v_t \varepsilon_t$ satisfies $\sup_{t\leq 1} E(|u_t|^4 | \mathcal{F}_{t-1})<\infty$ a.s. for any $\mathcal{F}_t$-adapted process $(z_t)$ since $\varepsilon_t \sim iid\mathbb{N}(0,1)$. Moreover, $u_t$ is not uniformly bounded, i.e., there does not exist $M$ such that $|u_t|\leq M < \infty$ with probability one even if $f$ is bounded above, since the standard normal random variable $\varepsilon_t$ is not uniformly bounded. Further examples of martingales with bounded conditional moments of MDS summands are provided by more general martingale transforms and randomly stopped sums of independent r.v.'s (see Remark 3.3 in \cite{VPRI}).

The hypothesis of no predictability of $(y_t)$ corresponds to the hypothesis $\beta=0$ in predictive model \eqref{PredRegr1}. It is well-known that the standard OLS-based $t$-test is not robust with respect to a wide range of statistical problems in predictive regression data. For instance, the standard OLS estimator of $\beta$ is not asymptotically Gaussian under $H_0: \beta=0$ if $(x_t)$ is endogenous and (nearly) nonstationary (see \citeauthor{elliott1994inference}, \citeyear{elliott1994inference}, \citeauthor{PhillipsNear}, \citeyear{PhillipsNear}, \citeauthor{PM}, \citeyear{PM}) or is stationary with infinite second moment (e.g., \citeauthor{GrangerOrr}, \citeyear{GrangerOrr}, \citeauthor{EKM}, \citeyear{EKM}, \citeauthor{ibragimov2015heavy}, \citeyear{ibragimov2015heavy}, and references therein), even when there is no heteroskedasticity  and $v_t = \sigma$ is constant.\footnote{The endogeneity of the covariate $x_t$ refers to the existence of nonzero long run covariance between innovations of $u_t$ and $x_t.$} 

In the case of predictive regressions for stock returns, the returns process $(y_t)$ is widely believed to have time-varying stochastic volatility (see \cite{CJP2016} and references therein). Moreover, the volatility process is typically very persistent. For example, many authors have found that the autoregressive parameter for the dynamics of the volatility process is close to one under some appropriate functional transformations. In particular, \cite{jacquier2004bayesian} and \cite{hansen2014estimating}, provide convincing evidence that the logarithm of the volatility process follows a near unit root process for a wide range of equity and foreign exchange rate time series. It is well known that the presence of  persistent volatility may cause the distribution of the standard $t$-statistic to be far from standard normal, yielding a substantial distortion in testing relying on standard normal critical values (see, e.g., \cite{chung2007nonstationary}, \cite{CJP2016} and \cite{KimPark1}).

\subsection{The Cauchy Estimator}

Our inference method is based on the Cauchy estimator. To effectively explain the main idea, we consider model \eqref{PredRegr1} with no intercept term, i.e., $\alpha=0$, and introduce the Cauchy estimator $\check{\beta}$ for $\beta$, which is given by
\begin{align*}
\check{\beta} 
= \left(\sum_{t=1}^T  |x_{t-1}| \right)^{-1}\sum_{t=1}^T sign(x_{t-1}) y_{t},
\end{align*}
where $sign(\cdot)$ is the sign function defined as $sign(x)=1$ for $x\geq0$, and $sign(x)=-1$ for $x<0$. Thus, $\check{\beta}$ is an instrumental variable (IV) estimator with the instrument $sign(x_{t-1}).$  This particular IV estimator was first proposed by \cite{cauchy1836lxxviii}. See, among others, \cite{SoShin}, \cite{PPC}, \cite{CJP2016} and \cite{kim-meddahi-2019} for econometric applications of the Cauchy estimator.  

Under Assumption \ref{assumption-mds} (b), not only $\varepsilon_t$, but also $sign(x_{t-1})\varepsilon_t$, hereafter denoted by $\xi_t$, is an MDS with respect to the filtration $(\mathcal{F}_t)$ with $E(\xi_t^2|\mathcal{F}_{t-1})=1$. Let us define a continuous time partial sum process $(W_T(r), 0\leq r\leq 1)$ by
\begin{equation}\label{BMs}
W_T(r) = \frac{1}{T^{1/2}}\sum_{t=1}^{[Tr]} \xi_t.
\end{equation}
The stochastic process $(W_T(r))$ takes values in $\mathbf{D}_{\mathbb{R}}[0,1]$, where $\mathbf{D}_E[0,1]$ denotes the space of c\`{a}dl\`{a}g functions from $[0,1]$ to $E\subset \mathbb{R}^d$ for some positive integer $d.$ Under Assumption \ref{assumption-mds} (b)-(c), the partial sum process $(W_T(r))$ follows the usual functional central limit theorem (CLT) for martingales, see, e.g., Theorem 18.2 of \cite{billingsley1986convergence}, that is,
\[
W_T \to_d W
\]
in $\mathbf{D}_\mathbb{R}[0,1]$, where $W$ is a standard Brownian motion. The convergence $W_T \to_d W$ is to be interpreted as the weak convergence in the probability measures on $\mathbf{D}_\mathbb{R}[0,1]$. In our context, it is more convenient, and so is assumed, to endow $\mathbf{D}_E[0,1]$ with the uniform topology rather than the usual Skorohod topology (see \citet[pp. 150-152]{billingsley1986convergence}).

The use of the Cauchy estimator in our inference method is motivated by the above functional CLT for $(W_T(r))$. To convey the main idea,  assume that the volatility process $(v_t)$ is observable. Recall that the numerator of $\check{\beta}$ is $\sum_{t=1}^T sign(x_{t-1})y_t$, and it becomes $\sum_{t=1}^T v_t \xi_t$ under $\beta=0$, where $\xi_t = sign(x_{t-1})\varepsilon_t$. One then may construct a robust test for the null hypothesis $H_0: \beta=0$ against the alternative $H_1: \beta\neq 0$ using the following statistic
\begin{align}\label{test-infeasible}
\tau(v) = \frac{1}{T^{1/2}}\sum_{t=1}^T sign(x_{t-1})\frac{y_t}{v_t}.
\end{align}
In particular, for $\beta=0$, 
\[
\tau(v) = \frac{1}{T^{1/2}} \sum_{t=1}^T \xi_t = W_T(1)  \to_d W(1) = \mathbb{N}(0,1).
\]
In practice, however, the volatility process $(v_t)$ is not observable, and hence, the above inference procedure using $\tau(v)$ is not feasible. In Section 3, a feasible version of the Cauchy based inference method above will be fully addressed under our construction of the persistent volatility introduced in Section 2.3.

\subsection{Persistent Volatility}

This subsection presents a time-varying and persistent volatility model, which is a well-known stylized fact for many financial returns. We define a stochastic process $\sigma_T$ on $\mathbf{D}_{\mathbb{R}^+}[0,1]$ as $\sigma_T(r) = v_{[Tr]}$. We assume that $\sigma_T$ has a limiting process $\sigma$ defined over $0\leq r\leq 1$ such that $(W_T, \sigma_T)$ converges to $(W, \sigma)$ jointly, where $(W_T)$ is defined as in \eqref{BMs}. Specifically, we consider the following assumption.

\begin{ass}\label{assumption-limit}
There exists a positive process $\sigma$ on $\mathbf{D}_{\mathbb{R}^+}[0,1]$ such that
\[
(W_T, \sigma_T) \to_d (W, \sigma)
\]
in $\mathbf{D}_{\mathbb{R} \times \mathbb{R}^+}[0,1]$, where $W$ is a standard Brownian motion with respect to the filtration to which $W$ and $\sigma$ are adapted. 
\end{ass}

The above assumptions hold for wide classes of models, such as models with nonstationary volatility, regime switching, and structural breaks in volatility. It also holds for the processes with $v_t = \sigma(t/T),$ where $\sigma(s)$ is a deterministic function on $[0, 1],$ considered by \cite{cavaliere-taylor-2007}, \cite{xu-phillips-2008} and \cite{harvey-leybourne-zu}, among others.\footnote{Assumption 2.2 is a simplified version of the condition $v_{[Tr]}/a_T\to_d \sigma_r$ considered by Assumption 2 of \cite{cavaliere-taylor-2009}. We rule out the explosive volatility settings with $a_T\to\infty$, and consider the stable volatility processes with $a_T=1$ for simplicity. The results in the paper can be obtained under the explosive volatility assumption with $a_T\to\infty$ at the cost of a more involved analysis.} The assumptions also hold for processes with nonstationary volatilities considered by \cite{hansen1995regression} and  \cite{chung2007nonstationary}, who assume that $v_t^2$ is a smooth positive transformation of a (near) unit root process, i.e., $v_t^2 = \sigma^2(T^{-1/2}z_{t-1})$ for a unit root process $z_t$. One should note that Assumption \ref{assumption-limit} is more general than the volatility models considered in the aforementioned literature and, in particular, allows the volatility to be stochastically discontinuous, which are desirable properties for modelling  financial volatility having structural breaks or regime switching. 

Assumptions \ref{assumption-mds} and \ref{assumption-limit} rule out some cases of globally homoskedastic processes, such as stationary GARCH processes. In Section 4, we generalize the model to have a two-factor volatility, one for a nonstationary long run component and the other one for a stationary short run component, and show the validity of our robust method introduced in Section 3 for the generalized model with the two-factor volatility.

Under our construction of the persistent volatility, the asymptotic behavior of the Cauchy estimator can be obtained immediately. The asymptotics of the Cauchy estimator $\check{\beta}$ are mainly determined by $\sum_{t=1}^T v_t\xi_t$ since $\check{\beta} = \beta + \sum_{t=1}^T v_t\xi_t/\sum_{t=1}^T |x_{t-1}|$. Note that $T^{-1/2}\sum_{t=1}^{[Tr]} v_t\xi_t = \int_0^r \sigma_T(s)dW_T(s)$ for $r\in[0,1]$, and the weak convergence of the stochastic integral $\int \sigma_T(r)dW_T(r)$ is well documented in the literature (see, e.g., Theorem 2.1 of \cite{hansen-1992} and Theorem 4.6 of \cite{KP}), and we have $(\int \sigma_T(r) dW_T(r)) \to_d (\int \sigma(r) dW(r))$. 

\begin{lem}\label{lemma-cauchy}
Under Assumptions \ref{assumption-mds} and \ref{assumption-limit},
\[
\left(\sum_{t=1}^T  |x_{t-1}|/\sqrt{T} \right)\left(\check{\beta}-\beta\right) \to_d \int_0^1 \sigma(r)dW(r).
\]
\end{lem}

Two main implications of Lemma \ref{lemma-cauchy} are (i) the limit distribution of the Cauchy estimator is generally non-Gaussian and (ii) the rate of convergence of the Cauchy estimator is nonstandard and unknown. These asymptotic properties of the Cauchy estimator subsequently imply that the usual $t$-test relying on the standard normal table becomes an invalid testing procedure for the null hypothesis of $\beta = 0$. The limit $\int_0^1 \sigma(r)dW(r)$ is Gaussian if and only if the limiting volatility process $\sigma$ is independent of $W$. In this case, $\int \sigma(r)dW(r)$ has a mixed normal distribution and $\int_0^s \sigma(r)dW(r) =_d \mathbb{MN}(0,\int_0^s \sigma^2(r)dr)$. If the independence condition is violated, then $\int \sigma(r)dW(r)$ becomes a non-Gaussian martingale in general.

Clearly, $\check{\beta}$ requires an extremely mild condition for consistency, that is $\sum_{t=1}^T  |x_{t-1}|/\sqrt{T}\to_p\infty$. For example, if there exists a sequence $p_T$ of positive numbers such that
\[
\left(p_T^{-1}\sum_{t=1}^T|x_{t-1}|\right)^{-1} = O_p(1),
\]
then $\check{\beta} - \beta = O_p(T^{1/2}/p_T)$ by Lemma \ref{lemma-cauchy}. For a wide class of time series, the consistency condition $T^{1/2}/p_T\to0$ is satisfied since $p_T\geq T$ unless $x_t\approx 0$ for most $t=1,\cdots,T$. Though it is not necessary in our subsequent theory, one may explicitly obtain the sequence $p_T$ for some time series satisfying required regularity conditions.

\begin{example}\label{example-1-1}
(a) For weakly stationary processes $(x_t)$ with $E|x_t|<\infty$, $p_T=T$;

(b) For stationary $\alpha$-stable $(x_t)$ with $0<\alpha<1$, $p_T=T^{1/\alpha}\ell(T)$ for some slowly varying function $\ell$ (see \citeauthor{EKM}, \citeyear{EKM}, \citeauthor{PSolo}, \citeyear{PSolo}, and references therein);

(c) For the case of unit root and near unit root time series $(x_t)$, $p_T=T^{3/2}$ (see \citeauthor{PhillipsUR}, \citeyear{PhillipsUR}, \citeyear{PhillipsNear}, \citeauthor{PM}, \citeyear{PM},  \citeauthor{IP}, \citeyear{IP}, and references therein);

(d) For fractionally integrated $I(d)$ processes $(x_t)$ with $1/2<d<3/2,$ $p_T=T^{d+1/2}\ell(T)$ for some slowly varying function $\ell$ (see \citeauthor{Baillie}, \citeyear{Baillie}; Lemma 3.4 in \citeauthor{PhillipsFrac1}, \citeyear{PhillipsFrac1}; \citeauthor{PhillipsFrac2}, \citeyear{PhillipsFrac2}, \cite{wangetal} and  \cite{chanwang} and references therein);
\end{example}

\section{New Robust Inference Approach}

Now we introduce our test for no predictability in the regression (\ref{PredRegr1}). The test is motivated by $\tau(v)$ in \eqref{test-infeasible}. Since $(v_t)$ is not observable, we replace $v_t$ by its consistent estimator $\hat{\sigma}((t-1)/T)$, and we consider the test statistic $\tau(\hat{\sigma})$ defined as
\begin{align}\label{test}
\tau(\hat{\sigma}) 
= \frac{1}{T^{1/2}}\sum_{t=1}^T sign(x_{t-1})\frac{y_t}{\hat{\sigma}((t-1)/T)},
\end{align}
where
\begin{align}\label{volest}
\hat{\sigma}^2(r)=\frac{\sum_{t=1}^T\hat{u}_t^2 K_h(r - t/T) }{ \sum_{t=1}^T K_h(r-t/T)}, 
\quad h\leq r\leq 1;
\qquad
\hat{\sigma}^2(r) = \hat{\sigma}^2(h),  
\quad 0\leq r<h,
\end{align}
where $\hat{u}_t$ the OLS residuals given as $\hat{u}_t = y_t -  \hat{\beta} x_{t-1}$ with the OLS estimator $\hat{\beta}$. Here $K_h(t)=K(t/h)$ with a kernel function $K$ and bandwidth $h$.

The validity of our approach requires that $\hat{\sigma}(r)$ is close enough to $\sigma_T(r)=v_{[Tr]}$ for most $r\in[0,1]$. We first establish a uniform convergence result
\begin{align}\label{eq-uniform}
\sup_{r\in \mathcal{C}_h}\left|\hat{\sigma}^2(r) - \sigma_T^2(r)\right|=o_p(1)
\end{align}
for some $\mathcal{C}_h\subset [0,1]$. Invoking the convergence in Assumption \ref{assumption-limit}, $\sigma_T\to_d \sigma$ is interpreted as the weak convergence in the probability measures on $\mathbf{D}_{\mathbb{R}^+}[0,1]$ endowed with the uniform topology. By virtue of the so-called Skorohod representation theorem (e.g., Pollard (1984), pp. 71-72), it is indeed possible to construct $\sigma_T$ and $\sigma$ on a common probability space, up to the distributional equivalence, so that $\sigma_T\to_{a.s.}\sigma$ uniformly on $[0,1]$. For our development of the uniform convergence results \eqref{eq-uniform}, we assume that $\sigma_T$ is defined up to the distributional equivalence such that $\sigma_T\to_{a.s.}\sigma$ uniformly on $[0,1]$. This assumption is not restrictive since we are interested in the convergence of $\hat{\sigma}^2$ to $\sigma_T^2$ rather than $\sigma^2$.

For the nonparametric estimator $\hat{\sigma}$, we assume the kernel function $K$ satisfies the following assumption.

\begin{ass}\label{assumption-kernel}
(a) a nonnegative kernel $K$ has a compact support $[0,1]$ with $\int_0^1 K(r)dr=1$, (b) $|K(r) - K(s)|\leq \bar{K}|r-s|$ for all $r,s\in\mathbb{R}$, and $\sup_r K(r)<\bar{K}$ for some $0<\bar{K}<\infty$, 
\end{ass}

The condition (b) in Assumption \ref{assumption-kernel} is standard in the investigation of uniform consistency. In Assumption \ref{assumption-kernel} (a), we assume a nonstandard assumption that $K$ is a one-sided kernel, which is unnecessary in developing the uniform consistency \eqref{eq-uniform}. When we establish $\tau(\hat{\sigma})\to_d \mathbb{N}(0,1)$ under $\beta=0$, however, it is important to make $\hat{\sigma}(t/T)$ measurable with respect to $\mathcal{F}_{t+1}$ so that we can apply a martingale CLT to $\tau(\hat{\sigma})$.\footnote{For the same reason, \cite{hansen1995regression} considered a one-sided kernel.} For a more precise explanation, we write 
\begin{equation}\label{decomposition}
\hat{\sigma}^2(r) = \hat{\sigma}_1^2(r) + \hat{\sigma}_2^2(r) + \hat{\sigma}_3^2(r) + \hat{\sigma}_4^2(r),
\end{equation}
where
\begin{align*}
&\hat{\sigma}_1^2(r)=  \frac{\sum_{t=1}^T E(u_t^2 | \mathcal{F}_{t-1}) K_h(r-t/T)}{\sum_{t=1}^T K_h(r-t/T)},
&&\hat{\sigma}_2^2(r)=\frac{\sum_{t=1}^T (u_t^2 - E(u_t^2 | \mathcal{F}_{t-1})) K_h(r-t/T)}{\sum_{t=1}^T K_h(r-t/T)},\\
&\hat{\sigma}_3^2(r)=(\hat{\beta} - \beta)^2\frac{ \sum_{t=1}^T x_{t-1}^2 K_h(r-t/T)}{\sum_{t=1}^T K_h(r-t/T)},
&&\hat{\sigma}_4^2(r)=2(\hat{\beta} - \beta)\frac{ \sum_{t=1}^T x_{t-1} u_t K_h(r -t/T)}{\sum_{t=1}^T K_h(r-t/T)}.
\end{align*}

In the decomposition \eqref{decomposition}, $\hat{\sigma}_1^2(r) - \sigma_T^2(r)$ is a bias term since $E(u_t^2 | \mathcal{F}_{t-1}) = v_t^2 = \sigma_T^2(t/T)$, whereas $\hat{\sigma}_2^2(r)$ is a variance term involving a martingale.  On the other hand, $\hat{\sigma}_3^2$ and $\hat{\sigma}_4^2$ are error components induced by using $\hat{u}_t$, instead of $u_t$, in the kernel estimation of $\sigma_T^2$. Under Assumption \ref{assumption-kernel} (a), $\hat\sigma_1^2(t/T)$ is $\mathcal{F}_{t-1}$-adapted, whereas $\hat\sigma_2^2(t/T)$ is $\mathcal{F}_t$-adapted. Consequently, $\tilde{\sigma}^2(t/T)$, where $\tilde{\sigma}^2 = \hat\sigma_1^2 + \hat\sigma_2^2$, is $\mathcal{F}_t$-adapted, from which one may show that $\tau(\tilde{\sigma})\to_d \mathbb{N}(0,1)$, where $\tau(\tilde{\sigma})$ is defined as $\tau(\hat{\sigma})$ but with $\hat{\sigma}$ replaced $\tilde{\sigma}$, by a martingale CLT as long as $|\tilde{\sigma}^2(r)-\sigma_T^2(r)|=o_p(1)$ for most $r\in[0,1]$. However, $\hat{\sigma}_3^2(t/T)$ and $\hat{\sigma}_4^2(t/T)$ are not $\mathcal{F}_t$-measurable since $\hat{\beta}$ is not $\mathcal{F}_t$-measurable for any $t<T$. Consequently, we cannot directly apply a martingale CLT to show $\tau(\hat{\sigma})\to_d \mathbb{N}(0,1)$. Alternatively, for these two terms, it is shown that they have negligible effects in the test statistic $\tau(\hat{\sigma})$, and we have $\tau(\hat{\sigma}) = \tau(\tilde{\sigma})(1+o_p(1))$ as long as $\hat{\beta}\to_p\beta$ sufficiently quickly. For the asymptotic negligibilities of $\hat{\sigma}_3^2$ and $\hat{\sigma}_4^2$, we assume

\begin{ass}\label{assumption-ols}
For any deterministic sequence $(c_t)_{t=1}^T$ such that $0\leq c_t\leq 1$ for all $t$, $\sum_{t=1}^T c_t x_{t-1}u_t = O_p\left(T^p \left(\sum_{t=1}^T x_{t-1}^2\right)^{1/2}\right)$ for some $p\in[0,1/8)$.
\end{ass}

Assumption \ref{assumption-ols} is very general and many time series models satisfy the condition. In particular, it holds with $p=0$ if $(x_t)$ is either (near) unit root or stationary with finite variance. Moreover, if $(x_t)$ is stationary with unbounded variance, then the condition holds with $p=0$ under some additional conditions on $(x_t)$ and $(u_t)$ (see, e.g., \cite{Samorodnitsky2007}). 

\begin{lem}\label{lemma-ols}
If Assumptions \ref{assumption-mds} and \ref{assumption-ols} hold, then $|\hat{\beta} - \beta| = O_p\left(T^p \left(\sum_{t=1}^T x_{t-1}^2\right)^{-1/2}\right)$.
\end{lem}

We will show below that the rate of convergence of $\hat{\beta}$ in Lemma \ref{lemma-ols} is enough to obtain the required uniform convergences of $\hat{\sigma}_3^2$ and $\hat{\sigma}_4^2$ as well as their asymptotic negligibility in the test relying on the statistic \eqref{test}.


On the other hand, the convergence $|\hat{\sigma}_1^2(r)- \sigma_T^2(r)|\to_p0$ requires that $\sigma_T^2$ be left-continuous at $r$ due, in particular, to the fact that $K$ is a one-sided kernel having  support $[0,1]$. However, $\sigma_T$ may have countably many jumps since $\sigma_T\in\mathbf{D}[0,1]$. In particular, at a discontinuity point $r$ with $\sigma(r)\neq \sigma(r-)$, we have $|\hat{\sigma}_1^2(r)- \sigma_T^2(r-)|\to_p0$ instead of $|\hat{\sigma}_1^2(r)- \sigma_T^2(r)|\to_p0$. Therefore, the set $\mathcal{C}_h$ in \eqref{eq-uniform} should effectively exclude a set of discontinuity points as well as its neighborhoods so that the uniform convergence result holds. Under our convention of $\sigma_T\to_{a.s.}\sigma$ uniformly on $[0,1]$, we only need to consider $\sigma$'s discontinuity points, and  we define
\begin{align}\label{set}
\mathcal{C}_h= [h,1]\setminus \mathcal{J}_h,\quad \text{where} \quad \mathcal{J}_h = \{[r, r+h)\subset[0,1] | \sigma(r)\neq \sigma(r-)\}.
\end{align}
Clearly, $\mathcal{C}_h$ is a set of left-continuity points, and we establish the uniform convergence result \eqref{eq-uniform} over $\mathcal{C}_h$.

A martingale exponential inequality can be used to show the asymptotic negligibility of the variance component $\hat{\sigma}_2^2(r)$ uniformly in $r$ (see, e.g., \cite{victor1999general} and \citet{bercu2008exponential}). In this paper, we use the two-sided exponential inequality in \cite{bercu2008exponential} under which we can relax the moment condition for $(\varepsilon_t)$ at the cost of an assumption on the stochastic order of the extremal process of $(\varepsilon_t)$.
     
\begin{ass}\label{assumption-error}
For some $q\in[0,1/8)$, $\max_{1\le t\le T}|\varepsilon_t| = O_p(T^q)$.
\end{ass}

Assumption \ref{assumption-error} is not stringent, and a wide class of time series models for $\varepsilon_t$ satisfies the condition.\footnote{Instead of Assumption \ref{assumption-error}, one may obtain the subsequent results by assuming an additional moment condition, i.e., $E|\varepsilon_t|^{4r}<\infty$ for some $r\geq1$. For a relevant approach, the reader is referred to, e.g., Theorem 2.1 of \cite{wang2014uniform}.} For instance, if $\varepsilon_t$ is a Gaussian process with $cov(\varepsilon_1, \varepsilon_T)\log T\to0$, then $\max_{0\leq t\leq T}|\varepsilon_t| = O_p(\sqrt{\log T})$ and the condition (a) holds for any $q>0$ (see, e.g., Theorem 2.5.2 of \cite{leadbetter1988extremal}).  


\begin{ass}\label{assumption-volest}
As $h\to0$ and $T\to\infty$, (a) $h T^{1/2-2p}\to\infty$ where $p\in[0,1/8)$ is defined as in Assumption \ref{assumption-ols}, and  (b) $h T^{1-4q}\to\infty$ and $hT^{2q}\to0$, where $q\in[0,1/8)$ is defined as in Assumption \ref{assumption-error}.
\end{ass} 

Assumption \ref{assumption-volest} provides the connections among the stochastic orders in Assumptions \ref{assumption-ols}-\ref{assumption-error} and the bandwidth $h$. If we let $h=c T^{-\alpha}$ for $c,\alpha>0$ as in the typical situation, then Assumption \ref{assumption-volest} holds for $2q< \alpha <\min\{1/2 -2p, 1-4q\}$. Note that such $\alpha$ always exists for $p,q\in[0,1/8)$. In particular, if $p=q=0$, then $h=c T^{-\alpha}$ satisfies Assumption \ref{assumption-volest} for $0<\alpha<1/2$. We note that the condition (a) is used to guarantee $\hat{\sigma}_3^2$ and $\hat{\sigma}_4^2$ being asymptotically negligible in our inference method. In condition (b),  $h T^{1-4q}\to\infty$ is needed for the uniform convergence of $\hat{\sigma}_2^2$, whereas $hT^{2q}\to0$ is used to effectively handle discontinuity points of $\sigma^2$ at which $\hat{\sigma}^2$ becomes inconsistent.

\begin{prop}\label{proposition-vol}
Let Assumptions \ref{assumption-mds}-\ref{assumption-limit} and \ref{assumption-kernel}-\ref{assumption-volest} hold. As $h\to0$ and $T\to\infty$, we have
\begin{align*}
&(a)
\quad \sup_{r\in \mathcal{C}_h} |\hat{\sigma}_1^2(r)-\sigma_T^2(r)|=o_p(1),
&&(b)
\quad \sup_{h\leq r\leq 1} |\hat{\sigma}_2^2(r)|=O_p\left(T^{2q} \left(\log (hT)/ (hT)\right)^{1/2}\right),\\
&(c)
\quad \sup_{h\leq r\leq 1} |\hat{\sigma}_3^2(r)|=O_p\left(T^{2p} /(hT)\right),
&&(d)
\quad \sup_{h\leq r\leq 1} |\hat{\sigma}_4^2(r)|= O_p\left(T^{2p} /(hT)\right),
\end{align*}
and the uniform convergence result \eqref{eq-uniform} holds.
\end{prop}

Under Assumption \ref{assumption-volest}, we have
\[
T^{2p} /(hT) = o\left( T^{2q} \left(\log (hT)/ (hT)\right)^{1/2}\right),
\]
from which we can see that the error components $\hat{\sigma}_3^2$ and $\hat{\sigma}_4^2$ have smaller orders than $\hat{\sigma}_2^2$. Indeed, it is shown in the proof of Theorem \ref{theorem-vol} that $\hat{\sigma}_3^2$ and $\hat{\sigma}_4^2$ have negligible effects in the test statistic $\tau(\hat{\sigma})$, and we have $\tau(\hat{\sigma})=\tau(\tilde{\sigma})(1+o_p(1))$, where $\tilde{\sigma}^2 = \hat{\sigma}_1^2 + \hat{\sigma}_2^2$. However, the orders of  $\hat{\sigma}_1^2$ and $\hat{\sigma}_2^2$ are not sufficiently small to show directly that $\tau(\tilde{\sigma}) = \tau(\sigma_T)(1+o_p(1))$ even though $\tilde{\sigma}^2$ converges uniformly to $\sigma_T^2$. That is because the convergence rate of $\hat{\sigma}_2^2(r)\to_p0$ is not fast enough for the direct approximation $\tau(\tilde{\sigma}) = \tau(\sigma_T)(1+o_p(1))$, and the convergence rate of $|\hat{\sigma}_1^2(r)-\sigma_T^2(r)|\to0$ depends on the degrees of left-continuity of $\sigma_T^2$ which are unknown in general. 

Alternatively, we use the weak convergence of the stochastic integral $\int_0^1 (\sigma_T(r)/\tilde{\sigma}(r))dW_T(r)$, as in Lemma \ref{lemma-limit}, jointly with the facts that $\tilde{\sigma}(t/T)$ is $\mathcal{F}_t$-adapted and $\sup_{r\in \mathcal{C}_h}|\tilde{\sigma}^2(r)-\sigma_T^2(r)|=o_p(1)$. Here, in particular, we require $\mathcal{C}_h\to_{a.s.}[0,1]$ which holds when $\sigma$ has finitely many jumps almost surely.

\begin{ass}\label{assumption-vol}
$\sigma$ has finitely many jumps almost surely.\footnote{In other words, $\sigma$ is of finite activity in the sense that the probability measure of any set $\{\omega: r\mapsto \sigma(r,\omega) \text{ has finitely many jumps in }r\in[0,1]\}$ is one.}
\end{ass}

\begin{thm}\label{theorem-vol}
Let Assumptions \ref{assumption-mds}-\ref{assumption-limit} and \ref{assumption-kernel}-\ref{assumption-vol} hold. As $h\to0$ and $T\to\infty$, we have the following.

(a) Under $\beta=0$, $\tau(\hat{\sigma})\to_d \mathbb{N}(0,1)$.

(b) If $\beta\neq0$ and $\sum_{t=1}^{T-1}|x_t|/\sqrt{T}\to_p \infty$, then \[
|\tau(\hat{\sigma})|\geq \frac{|\beta|}{\underline{v}} \frac{1}{\sqrt{T}}\sum_{t=1}^{T-1}|x_t| + O_p(1)\to_p \infty,  
\]
and hence, $P\left[|\tau(\hat{\sigma})|>c\right]\to1$ for any positive constant $c$.
\end{thm}

\begin{remark}
The asymptotic power result in Theorem \ref{theorem-vol} (b) implies that the testing power is mainly determined by the asymptotic behavior of $\sum_{t=1}^{T-1}|x_t|$. As an illustration, assume that there exists a diverging sequence $p_T$ such that $p_T/T^{1/2}\to\infty$ and
\begin{align}\label{eq-local}
\frac{1}{p_T}\sum_{t=1}^T |x_{t-1}|\to_d P
\end{align}
for a random variable $P>0$. Clearly, a wide class of time series satisfies \eqref{eq-local} (see, e.g., Example \ref{example-1-1}). It then follows immediately from Theorem \ref{theorem-vol} (b) that $|\tau(\hat{\sigma})|\to_p\infty$ and the speed of divergence is no slower than $p_T/\sqrt{T}$. 

Heuristically, we may consider the power of the proposed test by considering local alternatives in which $\beta\neq0$, but $\beta\to0$ at an appropriate rate. For our purpose, let $(P_T(r), 0\leq r\leq 1)$ be a continuous time process defined as
\[
P_T(r) = \frac{1}{p_T}\sum_{t=1}^{[Tr]} \frac{|x_{t-1}|}{\sigma_T(t/T)}
\]
for a diverging sequence $p_T$ such that $p_T/T^{1/2}\to\infty$. We then assume instead of \eqref{eq-local} that
\begin{align}\label{eq-local2}
P_T
\to_d P
\end{align}
for a stochastic process $(P(r),0\leq r\leq 1)$ having a positive support. If the convergence \eqref{eq-local2} holds jointly with the convergence in Assumption \ref{assumption-limit}, then we may develop the power of the proposed test by considering the following alternative hypothesis
\begin{align}\label{eq-localh}
\beta = \bar{\beta} \times \frac{\sqrt{T}}{p_T}
\end{align}
for a constant $\bar{\beta}\in\mathbb{R}\setminus \{0\}$. Clearly, the hypothesis \eqref{eq-localh} can be interpreted as a local alternative since $\beta\neq0$ and $\beta\to0$. Our construction of the local alternative \eqref{eq-localh} is useful to develop the asymptotic power result in a unified framework, especially, when the covariate $(x_t)$ is general, but satisfies \eqref{eq-local2}. Under the local alternative \eqref{eq-localh}, one can easily deduce from the proof of Theorem \ref{theorem-vol} with \eqref{eq-local2} that
\[
\tau(\hat{\beta})\to_d \bar{\beta}P(1) + \mathbb{N}(0,1).
\]
Clearly, $\tau(\hat{\sigma})$ under \eqref{eq-localh} is not Gaussian asymptotically unless $P(1)$ is either constant or Gaussian.

\begin{example}\label{example-1-2}
Let $\sigma_T(r)=\sigma$ for all $r\in[0,1]$.

(a) $(x_t)$ be a stationary process such that $T^{-1}\sum_{t=1}^T |x_t|\to_p E|x_1|<\infty$. Under \eqref{eq-localh} with $p_T=T$, $\tau(\hat{\sigma})\to_d \mathbb{N}\big((\bar{\beta}/\sigma)E|x_1|,\,\, 1\big)$. 

(b) $(x_t)$ be a unit root (or near unit root) process such that $T^{-3/2}\sum_{t=1}^T |x_t|\to_d \int_0^1 |X(r)|dr$, where $(X(r))$ is the limiting Brownian motion (or Ornstein-Uhlenbeck process) of $(x_t)$ such that $X_T\to_d X$ for $X_T(r) = T^{-1/2}x_{[Tr]}$.\footnote{The reader is referred to \cite{PhillipsNear} and \cite{park2003strong} for more discussion about the near unit root process and its limiting behaviors. Here, in particular, the Ornstein-Uhlenbeck process $X$ follows
\[
dX(r) = c X(r) dr + dV(r), \quad X(0)=0,
\]
where $V$ is Brownian motion.}  Under \eqref{eq-localh} with $p_T=T^{3/2}$, $\tau(\hat{\sigma})\to_d (\bar{\beta}/\sigma)\int_0^1 |X(r)|dr + \mathbb{N}(0, 1)$.   
\end{example}

When $\sigma_T(r)=\sigma$ for all $r\in[0,1]$ and $(x_t)$ is stationary with $E(x_t^2)<\infty$, the usual $t$-test procedure is a valid hypothesis testing procedure for the model \eqref{PredRegr1}. In this case, the asymptotic power property of the usual $t$-test is also well known under the local alternative hypothesis \eqref{eq-localh} with $p_T = T$, and is given by
\[
\text{$t$-statistic}\to_d
\mathbb{N}\big((\bar{\beta}/\sigma) (E(x_t^2))^{1/2},\,\, 1\big).
\]
The ratio of the asymptotic biases of the $t$-test and our test, obtained in Example \ref{example-1-2} (a), is given by $(E(x_t^2))^{1/2} / E|x_t|$. Importantly,  the ratio is always greater than one as long as $E(x_t^2)<\infty$ due to Jensen's inequality. This implies that the usual $t$-test is more powerful than our test under the ideal assumptions, even though the statistics in both tests diverge at the same rate $T^{1/2}$ under a fixed alternative hypothesis. However, when one or more of the ideal assumptions are violated, our test remains valid, whereas the usual $t$-test becomes invalid. This is another example of the traditional issue of trade-off between efficiency and robustness.
\end{remark}

\begin{remark} A number of works in statistics and econometrics have focused on robust inference using sign tests 
applied to different models, including time series regressions (see, among others, \cite{DH}, \cite{CD}, \cite{SoShin1}, \cite{delaPena}, \cite{IB}, \cite{kim-meddahi-2019}, and references therein). For instance, \cite{CD} propose sign tests for testing independence of a zero median time series $Y_t$ with $P(Y_t=0)=0,$ e.g., a time series with continuous distributions symmetric about zero, of past values of $Y_t$ and another time series $X_t.$ The tests in \cite{CD} are based on the observation that, under the above independence/orthogonality hypothesis, for any $T\ge 1,$ the sign statistic like $S_0=0.5(\sum_{t=1}^T sign(Y_tX_{t-1})+T)$ and its more general analogues follow a Binomial distribution with parameters $T$ and 0.5: $S_0\sim Bi(T, 0.5)$ (the results in \cite{IB} imply that sign tests for general zero median or symmetric processes $Y_t$ can be based on similar statistics with randomization over zero values of $Y_t$). \cite{Efron}, \cite{Edelman}, \cite{Pinelis}, \cite{DH}, and \cite{delaPena} consider related testing procedures based on bounds for tail probabilities of $t$-statistics of a parameter of interest (e.g., a location parameter or a regression/autoregression coefficient) under symmetry assumptions implied by (sharp) bounds on tail probabilities of weighted sums of i.i.d. symmetric Bernoulli r.v.'s.  

Naturally, in the time series regression context, the above sign-based inference approaches are more robust to moment assumptions and heavy tails than the inference procedures based on the Gaussian asymptotics for the full-sample OLS and Cauchy estimators. Typically, the sign-based tests can be used without any moment conditions on the time series considered, e.g., under infinite variances. However, they usually require symmetry or zero median assumptions on the processes. Such assumptions are often too restrictive in empirical applications, including the analysis of financial markets due to the stylized fact of gain-loss asymmetry in financial returns (see, among others, \cite{cont} and references therein). Further, sign-based tests are less efficient than those on the Gaussian asymptotics for the OLS estimator under the validity of the latter tests.   

\end{remark}

\begin{remark}
Our method can be applied to a discrete time model and a discrete sample collected from an underlying continuous time model as in  \citet{CJP2016}.  The main difference between our approach to \citet{CJP2016} is that we do not require the assumption $\delta\to0$, where $\delta$ is the sampling interval of the discrete samples. Clearly, the method of \citet{CJP2016} is applicable to high frequency data. Therefore, we may say that our method is more flexible since it can be applied to both high and low frequency data. The price we have to pay for the flexibility is the persistent volatility assumption $\sigma_T \to_d \sigma$ in Assumption \ref{assumption-limit}. Persistent volatility is a well-known stylized fact of financial time series and, in our view, is best considered within the model formulation.

Our method is also comparable to the IVX approach proposed by \cite{PhillipsMagdalinos2009}. The IVX approach is based on a self-generated instrument obtained by differencing the predictor $x_t$ and using an autoregressive filter to construct the instrument. As is shown in \cite{PhillipsMagdalinos2009}, the IVX approach is robust to a (near) unit root or mildly explosive predictor. The Cauchy estimator approaches to inference, including ours and \cite{CJP2016}, are restricted to a single regressor.\footnote{In general, a test relying on a single regressor exhibits size distortion when some relevant regressors are omitted. To overcome the issue induced by a single regressor in our approach, one may extend our approach to a multivariate setting based on the recent paper by \cite{shephard2020} in which a multivariate extension of the Cauchy estimator is proposed. An alternative extension is to use the parsimonious system approach (see \cite{Ghysels-Hill-Motegi-2020} and \cite{Xu-Guo-2022}) which is based on a set of misspecified regression models with only one group of regressors, allowing a single regressor for each regression. We leave these extensions for future research.} 
Unlike the Cauchy based inferences, the IVX approach is applicable to predictive regressions with multiple regressors. We also note that the IVX approach allows for conditional heteroskedasticity. However, to our knowledge, it is not known whether the IVX approach is valid when the volatility is persistent or the predictor is heavy-tailed with infinite second moments, and a continuous time extension of the IVX approach is not available in the literature. Therefore, our method and the IVX may be regarded as complementing each other. 

\cite{hansen1995regression} provides a nonparametric GLS method for regression models with nonstationary volatility using the estimator $\hat{\sigma}$ to correct the heteroskedasticity.  One should note that the assumptions on the limiting volatility $\sigma$ are more general than those in \cite{hansen1995regression} and other work in the literature on the topic. In particular, the assumptions in \cite{hansen1995regression} do not allow for structural changes or regime switching in the  volatility process as the limiting volatility is assumed to have continuous sample
paths almost surely. In contrast, the limiting volatility is allowed to have an arbitrary number of jumps in this paper, and hence, structural changes or regime switching are allowed. Moreover, we further extend our model to have a two-factor volatility in Section 4. 
\end{remark}

\section{An Extension to Two-Factor Volatility Models}
In this section, we generalize the model \eqref{PredRegr1} to have a two-factor volatility in the regression error $(u_t)$. More specifically, we assume that $(\varepsilon_t)$ is conditionally heteroskedastic, rather than conditional homoscedastic as is assumed in Assumption \ref{assumption-mds} (b).

\begin{ass}\label{two factor}
(a) $E(\varepsilon_t^2|\mathcal{F}_{t-1}) = w_t^2$ and $E(w_t^2)=1$, (b) $\max_{t\geq 1}E(|w_t|^{2\eta_1})<\infty$ for some $\eta_1>2$, (c) $(w_t)$ is $\alpha$-mixing such that the mixing coefficient $\alpha$ satisfies $\alpha(k)\leq Ak^{-\eta_2}$ for some $A<\infty$ and $\eta_2 > (2\eta_1  + 2)/(\eta_1 - 2)$, and (d) $hT^{\eta_3}\to\infty$ for some $\eta_3> (\eta_2(1-2/\eta_1) - 2/\eta_1 -2)/(\eta_2 +2)$.
\end{ass}

Under Assumptions \ref{assumption-mds} (a) and \ref{two factor}, the regression error $(u_t)$ in the model \eqref{PredRegr1} can be written as $u_t = v_t w_t e_t$, where $(e_t)$ is an MDS with respect to $(\mathcal{F}_t)$ such that $E(e_t^2 | \mathcal{F}_{t-1})=1$. Clearly, $(u_t)$ has two volatility factors, $(v_t)$ and $(w_t)$, where $(v_t)$ is the long run component by Assumption \ref{assumption-limit} and $(w_t)$ is the short run  component by Assumption \ref{two factor} (c).  Moreover, under Assumption \ref{assumption-limit} and the condition $E(w_t^2)=1$ in Assumption \ref{two factor}, we can identify and estimate the persistent volatility component $v_t$ by the nonparametric estimator \eqref{volest}.  In particular, under Assumption \ref{two factor}, we may show that
\[
\sup_{h\leq r\leq 1} \left|\frac{1}{hT}\sum_{t=1}^T (w_t^2 - 1) K_h(r-t/T)\right| = O_p\left((\log T/ (hT))^{1/2}\right)
\]
using an exponential inequality for a strongly mixing process (see, e.g., \citet[Theorem 2.1]{liebscher1996strong},  \citet[Theorem 4.1]{vogt2012nonparametric}) and \citet[Theorem 1]{kristensen2009uniform}). The above uniform convergence result for the mixing process is sufficient to develop the required uniform convergence of the volatility estimator as well as the validity of the inference procedure proposed in Section 3. We also note that the conditions for $h$ and $T$ in Assumption \ref{two factor} (d) and Assumption \ref{assumption-volest} hold simultaneously for any $p,q\in[0,1/8)$ as long as Assumption \ref{two factor} (d) holds for some $\eta_3>1/4$, which is not stringent. For instance, if $(w_t)$ is a stationary GARCH(1,1) process and $\beta$-mixing with exponential decay, which hold under some mild conditions (see, e.g., \cite{carrasco2002mixing} and \cite{francq2006mixing}), then both Assumption \ref{two factor} (d) and Assumption \ref{assumption-volest} hold for any $p,q\in[0,1/8)$ since Assumption \ref{two factor} (d) holds for any $\eta_3 >1$.

For our purpose, we again consider the decomposition \eqref{decomposition} of the nonparametric estimator \eqref{volest}, and write $\hat{\sigma}^2 = \sum_{k=1}^4 \hat{\sigma}_k^2$. We then can obtain the uniform convergence rate of each component $\hat{\sigma}_k^2$ for $k=1,2,3,4$ as in Proposition \ref{proposition-vol}, and establish the validity of the inference method relying on the test statistic \eqref{test} as in Theorem \ref{theorem-vol}. 

\begin{cor}\label{cor-vol}
Let Assumptions \ref{assumption-mds} (a) and (c), \ref{assumption-limit}, \ref{assumption-kernel}-\ref{assumption-vol} and \ref{two factor} hold. As $h\to0$ and $T\to\infty$, Proposition \ref{proposition-vol} and Theorem \ref{theorem-vol} remain valid.
\end{cor}

\section{Monte-Carlo Simulations}

This section provides the numerical results on finite sample performance of the proposed robust test based on $\tau(\hat{\sigma})$. We present the comparisons of the finite sample properties of the test with the test proposed by \citet{CJP2016} (denoted as Cauchy RT; RT for random time) and also two other tests considered in \citet{CJP2016}: the Bonferroni $Q$-test of \cite{CampbellYogo2006} (denoted as BQ) and the restricted likelihood ratio test of \cite{ChenDeo2009} (denoted as RLRT).


We consider two different settings for simulation models: continuous time and discrete time DGPs. As for the continuous time DGPs, we follow the simulation designs of \citet{CJP2016}. The data is generated in the continuous time setting using the following DGP:
\begin{eqnarray}
dY_t&=&\frac{\bar{\beta}}{T} X_tdt+dU_t,
\quad\quad dU_t=\sigma_t\left(dW_{1t}+\int_{\mathbb R}{x\Lambda(dt,dx)}\right),\label{MC1}\\ \label{Xt}
dX_t&=&-\frac{\bar{\kappa}}{T} X_tdt+\sigma_tdW_{2t},
\end{eqnarray}
where $W_{1t}$ and $W_{2t}$ are Brownian motions with $E(W_{1t}W_{2t})=-0.98t$. We set the constant term in the predictive regression to be zero and use recursive de-meaning. We assume that the continuous time models are observed at $\delta$-intervals over $T$ years with $\delta=1/252$, which corresponds to daily observations of size $252 T$.

The volatility process considered in the numerical results is assumed to follow one of the following models:
\begin{itemize}
\item Model CNST. \textit{Constant volatility}: $
\sigma_t^2=\sigma_0^2$, $\sigma_0=1$.
\item Model SB. \textit{Structural break in volatility}: $\sigma_0+(\sigma_1-\sigma_0) 1\{t/T\geq 4/5\}$ with $\sigma_0=1$ and $\sigma_1=4$. 
\item Model GBM. \textit{Geometric Brownian motion}: $d\sigma_t^2=\frac{1}{2}\frac{\bar{\omega}^2}{T}\sigma_t^2dt+\frac{\bar{\omega}^2}{\sqrt{T}}\sigma_t^2dZ_t$, where $Z_t$  is a Brownian motion with $E(W_{1t} Z_t) = -0.4 t$, and $\bar{\omega} = 9$.
\item Model RS. \textit{Regime switching}: $\sigma_t=\sigma_0(1-s_t)+\sigma_1s_t$, where $s_t$ is a homogeneous Markov process indicating the current state of the world which is independent of both $Y_t$ and $X_t$ with the state space $\{0,1\}$ and the  transition matrix 
\[P_t=\begin{pmatrix}
0.8 & 0.2\\
0.8 & 0.2
\end{pmatrix}
+
\begin{pmatrix}
0.2 & -0.2\\
-0.8 & 0.8
\end{pmatrix}\exp\left(-\frac{\bar{\lambda}}{T}t\right),
\]
where $\bar{\lambda}=60$, $\sigma_0=1$ and $\sigma_1=4$. The process $s_t$ is initialized by its invariant distribution.
\end{itemize}
We set the number of years $T\in\{5,20,50\}$ (which corresponds to 60, 240 and 600 monthly data) and consider the values $\bar{\kappa}\in\{0,5,10\}$ for the persistence parameter $\bar{\kappa}$  of $X_t$ in (\ref{Xt}).   

As indicated before, in contrast to the Cauchy RT test in \citet{CJP2016}, our test is applicable, not only in the continuous time models, but also in the discrete time framework. We consider the following discrete time models in the analysis of the finite sample performance of the tests:
\begin{eqnarray}
y_t=\frac{\bar{\beta}}{T}x_{t-1}+\sigma_{\varepsilon,t}\varepsilon_t,\qquad
x_t=\left(1-\frac{\bar{\kappa}}{T}\right) x_{t-1}+\sigma_{\eta,t}\eta_t,\label{MC2}
\end{eqnarray}
for $t=2, \dots,T$, where $T\in\{60, 240, 600\}$ (the same number of monthly observations as in continuous time simulations) and the same values of $\bar{\beta}$ and $\bar{\kappa}$. Here the innovations $(\varepsilon_t, \eta_t)$ are assumed to be multivariate normal with the correlation coefficient $-$0.98.

For the volatility processes in the discrete time setting, we consider three specifications: Model CNST and Model SB as in the continuous time setup, and GARCH volatility dynamics with
\begin{align*}
\sigma_{\varepsilon,t}^2
=1+\alpha\varepsilon_{t-1}^2+\theta\sigma_{\varepsilon,t-1}^2,
\qquad
\sigma_{\eta,t}^2
=1+\alpha \eta_{t-1}^2+\theta\sigma_{\eta,t-1}^2.
\end{align*}
In the numerical analysis, we consider the ARCH(1) processes with $\theta=0,$ $\alpha=0.5773$ (stationary with infinite fourth moment); $\theta=0,$ $\alpha=0.7325$ (stationary with infinite third moment); IGARCH(1,1) models with $\alpha=0.9$, $\theta=0.1$ and $\alpha=0.1$, $\theta=0.9$ (nonstationary). Note that the ARCH(1) processes in our simulations violate the moment conditions in Assumption \ref{two factor}. As shown in our simulation results below, our approach has reliable size and power properties even though the required moment conditions are violated.\footnote{See, among others, \cite{MS}, \cite{DM}, \cite*{IPS}, and references therein for the results on moment properties of GARCH processes and their importance in robust econometric inference.}


\subsection{Finite Sample Size Properties}

In this section, we analyze finite sample size properties of the no predictability tests by setting $\bar{\beta}=0$ in the regression models \eqref{MC1} and \eqref{MC2}. The numerical results on the finite sample size properties are presented in Tables \ref{tab1}-\ref{tab2}. 

Table \ref{tab1} provides the finite sample size results for models CNST, SB, GBM and RS in the continuous time setting. The finite sample size values for the OLS, BQ, RLRT and Cauchy RT tests are exactly the same as those reported in \citet{CJP2016}. These numerical results show that the size of the OLS, BQ and RLRT tests is highly distorted for most of the time-varying volatility models considered. In contrast, the rejection probabilities of the proposed test are very close to their nominal levels, such as the Cauchy RT test, regardless of the values of $\bar{\kappa}$ and $T$, and the volatility models we consider in our simulations. For the $5\%$ test, rejection probabilities stay between $4\%$ to $8\%$ without any exception. 

As mentioned before, the Cauchy RT test is inapplicable in the discrete time settings. Table \ref{tab2} provides the numerical results on finite sample size properties of all the tests except Cauchy RT under the discrete time settings. The quantitative and qualitative comparisons of the size properties of the tests are similar to the continuous time case. In summary, the finite sample size properties reported in Tables \ref{tab1}-\ref{tab2} show that the proposed test has a reliable size performance and is widely applicable for both discrete and continuous time settings.

\subsection{Finite Sample Power Properties}

Figures \ref{fig1}-\ref{fig6} present the results on finite sample power properties of the tests considered.\footnote{We report the power properties for Models SB, GBM and RS (continuous time) as well as Models SB and GARCH (discrete time) in Section 5.2. The power properties for the other models, Model CNST (continuous time) as well as Models CNST and ARCH (discrete time), are presented in the Supplementary Online Appendix.} In our simulations, we  consider the DGPs in \eqref{MC1} in continuous time and  \eqref{MC2} in discrete time with $\bar{\beta}$ ranging from 0 to 20. All the power curves presented in the figures are size-adjusted. Taking into account the results of finite sample size performance of the tests and their comparisons, we mainly focus on two tests: the Cauchy RT and our test, in the analysis of finite sample power properties. 
For comparison, the analysis also provides the numerical results on the finite sample power of the OLS, BQ and RLRT tests.

In Figure \ref{fig1} for the case of the structural break in volatility, one observes that the Cauchy RT test appears to be superior to other testing approaches (except in the cases with $\bar{\kappa}=0$ and large $\bar{\beta}$). At the same time, the proposed test $\tau(\hat{\sigma})$  also appears to have good finite sample power properties especially in the case of highly persistent predictors.

Figure \ref{fig2} provides the numerical results on finite sample properties of the tests in the geometric Brownian motion case.  For the case of a unit root regressor, the power properties of the proposed test based on $\tau(\hat{\sigma})$ appear to outperform those of the Cauchy RT test which in turn outperforms other tests considered. However, the finite sample power performance of the Cauchy RT test improves in the case of near unit root regressor with $\bar{\kappa}=5$ and $\bar{\kappa}=20$. 

The power curves for the regime switching case presented in Figure \ref{fig3} demonstrate that the test based on the proposed test $\tau(\hat{\sigma})$ has better power  properties than other tests in the case $\bar{\kappa}=0.$ For the case of the near unit root persistence in the regressor, the power properties of the Cauchy RT test appear to be better than those of the test based on $\tau(\hat{\sigma})$ for relatively small sample sizes (small values of $T$). However, as the sample size increases, the test based on $\tau(\hat{\sigma})$ becomes more powerful than the Cauchy RT test (see, e.g., Figure \ref{fig3} for the case $\bar{\kappa}=5$ and $T=50$). For large deviations from a unit root regressor ($\bar{\kappa}=20$), the Cauchy RT is more powerful than other tests, but the power curves appear to be very similar.

Figures \ref{fig4}-\ref{fig6} present the numerical results on power properties under discrete time settings for all the tests considered except Cauchy RT which is inapplicable in  discrete time settings. Results in the figures are provided for the cases of the structural break in volatility (Figure \ref{fig4}); the GARCH cases with $(\alpha,\theta)=(0.9,0.1)$ (Figure \ref{fig5}) and $(\alpha,\theta)=(0.1,0.9)$ (Figure \ref{fig6}). For all the cases, the conclusions on power properties of the tests and their comparisons are virtually the same as in the continuous time framework. 


Overall, the numerical results on finite sample properties of the tests indicate good performance of the test based on $\tau(\hat{\sigma})$ in comparison to the Cauchy RT. Again, the latter test is inapplicable in the discrete time settings. Their relative finite sample performances vary across different models. Which test should be used in practice depends on the availability of high frequency data as well as the size-power trade-off for a specific model. 
Therefore, the test proposed in this paper and the Cauchy RT complement rather than substitute one another.

\setlength{\extrarowheight}{0.3ex}
\newcolumntype{C}{>{\centering\arraybackslash}X}

\section{Conclusion}

Endogenously persistent regressors have been extensively analyzed in the predictive regression literature. A widely believed characteristic of stock returns is heteroskedastic and persistent volatility, which is often ignored in the predictive regression literature except for  \cite{CJP2016}. These two characteristics cause standard hypothesis tests to become substantially biased and often over-reject the null of no predictability. The main contribution of this paper is to provide an inference method that is designed to be robust to these problematic characteristics of predictive regression data. The proposed method relies on the Cauchy estimator and a kernel-based nonparametric correction of volatility. Its theoretical validity is provided by analyzing the asymptotic size and power properties. Moreover, it is shown through a simulation study that the proposed method has a reliable finite sample performance compared to the most advanced existing inference methods.  

Our inference method is comparable to the method proposed by \cite{CJP2016}. Similar to our method, their approach relies on the Cauchy estimator and a nonparametric volatility correction. However, their approach to the volatility correction is quite different from ours, and its applicability is limited to a predictive regression equipped with appropriate high frequency data. In terms of finite sample properties, our method and the method by \cite{CJP2016} perform well and have good size and power performances under continuous time settings. However, unlike our method, \cite{CJP2016} method is not applicable under discrete time settings. In contrast, our method can be applied to a discrete time model as well as a discrete sample collected from an underlying continuous time model. Therefore, our method is more flexible and widely applicable since it can be applied to both high and low frequency data.

A further approach to robust inference in predictive regressions under heterogeneous and persistent volatility as well as endogenous, persistent or heavy-tailed regressors is provided by the simple to implement robust $t$-statistic inference approach (see \citet{IbragimovMuller2010}) based on asymptotically normal group Cauchy estimators of a regression parameter of interest. This  approach will be explored in a companion paper now in preparation.

\section*{Online Supplementary Material}
Ibragimov, R., Kim, J, and Skrobotov, A. (2022): Supplement to ``New robust inference for predictive regressions,'' Econometric Theory Supplementary Material. To view, please visit:

\bibliographystyle{agsm} 
\bibliography{Predictive}

%
%
%
%


\newpage
\renewcommand{\baselinestretch}{1}
\renewcommand{\thesection}{\Alph{section}}
\renewcommand{\theequation}{\Alph{section}.\arabic{equation}}
\setcounter{section}{0}
\setcounter{equation}{0}

\begin{table}[h!]
\begin{center}
\footnotesize
\caption{Size for the continuous time models\label{tab1}}

\begin{tabularx}{0.95\textwidth}{CcCCCCCCcCC} \toprule
&&\multicolumn{3}{c}{$\bar{\kappa}=0$}&\multicolumn{3}{c}{$\bar{\kappa}=5$}&\multicolumn{3}{c}{$\bar{\kappa}=20$}\\
\cmidrule(r){3-5}\cmidrule(r){6-8}\cmidrule(r){9-11}
T&&5&20&50&5&20&50&5&20&50\\\hline
CNST&OLS&42.2&42.0&43.0&19.5&19.5&19.7&11.1&11.2&10.9\\
&BQ&8.6&4.9&4.3&7.5&4.5&4.2&8.6&4.1&3.2\\
&RLRT&8.5&7.7&8.1&5.4&5.9&5.6&4.8&5.2&5.3\\
&Cauchy RT&5.3&4.9&5.3&5.2&5.4&4.7&5.5&5.1&5.1\\
&$\tau(\hat{\sigma})$&5.6&5.0&5.3&5.4&5.0&5.1&5.4&5.0&4.8\\
\hline
SB&OLS&38.3&38.8&39.9&29.6&30.8&31.2&24.3&26.4&26.0\\
&BQ&18.1&12.9&11.9&17.0&15.1&14.1&17.4&14.8&14.3\\
&RLRT&23.8&22.8&23.6&21.0&21.9&21.8&22.4&24.5&23.6\\
&Cauchy RT&5.6&5.0&5.1&5.2&5.3&5.0&5.4&5.0&4.9\\
&$\tau(\hat{\sigma})$&8.0&6.7&6.3&7.8&6.5&6.0&7.9&6.4&6.0\\
\hline
RS&OLS&42.9&43.6&44.6&22.0&23.4&24.5&14.9&18.9&19.5\\
&BQ&8.8&6.3&6.0&9.8&7.2&6.8&12.6&8.9&8.4\\
&RLRT&9.3&10.0&10.7&7.5&9.4&9.6&9.6&13.0&14.2\\
&Cauchy RT&5.0&4.8&5.2&4.9&4.9&4.9&5.4&5.1&4.8\\
&$\tau(\hat{\sigma})$&5.2&5.4&6.1&5.2&5.1&5.8&5.6&5.8&5.8\\
\hline
GBM&OLS&52.2&53.7&53.1&28.6&30.2&30.9&23.2&26.0&27.0\\
&BQ&16.8&12.5&11.3&13.9&12.4&13.2&15.8&11.7&12.0\\
&RLRT&21.7&22.3&21.9&16.3&17.8&19.0&21.4&23.3&23.5\\
&Cauchy RT&4.4&4.7&4.4&4.3&4.5&4.4&4.6&4.5&4.5\\
&$\tau(\hat{\sigma})$&5.4&5.5&6.1&5.7&5.7&5.9&5.7&5.9&6.5\\
 \bottomrule\smallskip
\end{tabularx}
\end{center}
The parameter $\bar{\kappa}$ measures the degree of persistence in the predictor. The sample size corresponds to $T$ yearly observations (total $12T$ observations). CNST, SB, GBM, and RS denote respectively constant volatility, structural break, geometric Brownian motion, and regime switching in volatility. 
\end{table}

\newpage

\begin{table}[h!]
\begin{center}
\footnotesize
\caption{Size for the discrete time models with CNST and SB\label{tab2}}

\begin{tabularx}{0.95\textwidth}{ccCCCCCCcCC} \toprule
&&\multicolumn{3}{c}{$\bar{\kappa}=0$}&\multicolumn{3}{c}{$\bar{\kappa}=5$}&\multicolumn{3}{c}{$\bar{\kappa}=20$}\\
\cmidrule(r){3-5}\cmidrule(r){6-8}\cmidrule(r){9-11}
T&&5&20&50&5&20&50&5&20&50\\\hline
CNST&OLS&43.9&43.8&44.7&19.4&19.8&20.1&9.7&11.2&10.8\\
&BQ&8.4&5.2&4.8&7.8&4.9&4.5&9.2&4.1&3.4\\
&RLRT&8.3&8.0&8.1&5.2&5.4&5.3&4.1&5.4&5.3\\
&$\tau(\hat{\sigma})$&5.5&5.1&5.0&5.5&4.8&5.1&5.1&5.2&5.2\\
\hline
SB&OLS&38.0&39.6&40.0&29.1&31.1&31.4&22.1&26.1&26.8\\
&BQ&17.2&12.8&12.3&16.5&15.1&14.5&17.7&15.0&15.2\\
&RLRT&23.1&23.5&24.2&19.9&21.8&21.4&21.2&24.6&25.1\\
&$\tau(\hat{\sigma})$&8.0&6.7&6.3&7.9&6.2&5.8&7.5&6.5&6.2\\
\hline
ARCH(1)&OLS&45.0&44.1&43.5&23.5&22.5&21.2&17.2&17.0&15.2\\
$\alpha=0.5773$&BQ&9.7&5.4&4.8&9.5&5.8&4.6&13.1&6.3&4.6\\
$\xi=4$&RLRT&9.6&8.8&8.7&9.0&7.9&6.7&13.1&11.3&9.1\\
&$\tau(\hat{\sigma})$&6.1&5.4&6.0&6.1&5.2&5.4&6.0&5.9&6.1\\
\hline
ARCH(1)&OLS&45.8&44.0&43.6&24.4&24.1&22.6&19.7&19.8&18.1\\
$\alpha=0.7325$&BQ&10.2&6.2&5.2&10.4&7.1&5.8&14.7&8.3&6.8\\
$\xi=3$&RLRT&10.7&10.0&9.2&10.6&10.0&8.1&15.9&14.9&12.9\\
&$\tau(\hat{\sigma})$&5.9&5.8&6.5&6.2&5.6&6.0&6.4&6.1&6.1\\
\hline
IGARCH(1,1)&OLS&44.9&45.8&45.6&20.1&21.8&24.3&11.1&14.9&17.3\\
$\alpha=0.9$&BQ&8.9&5.8&6.0&7.8&6.0&6.9&9.3&5.3&6.4\\
$\beta=0.1$&RLRT&9.1&10.1&11.5&5.7&8.3&9.8&6.2&9.0&11.5\\
&$\tau(\hat{\sigma})$&6.2&5.5&5.5&5.8&5.6&6.0&5.9&5.8&5.7\\
\hline
IGARCH(1,1)&OLS&46.0&46.5&45.1&26.9&28.5&28.0&21.6&26.1&26.2\\
$\alpha=0.1$&BQ&11.7&8.3&8.1&12.7&10.4&10.9&16.4&12.2&13.4\\
$\beta=0.9$&RLRT&13.3&13.0&12.7&13.7&14.7&15.1&20.2&23.7&23.2\\
&$\tau(\hat{\sigma})$&6.3&6.4&7.4&6.9&6.7&7.2&6.6&6.9&6.9\\
 \bottomrule\smallskip
\end{tabularx}
\end{center}
The parameter $\bar{\kappa}$ measures the degree of persistence in the predictor. The sample size corresponds to $T$ yearly observations (total $12T$ observations). CNST, SB, GBM, and RS denote respectively constant volatility, structural break, geometric Brownian motion, and regime switching in volatility. 
\end{table}

\begin{landscape}
\begin{figure}[h]%
\begin{center}%
\subfigure[$\bar{\kappa}=0$, $T=5$]{\includegraphics[width=0.30\linewidth]{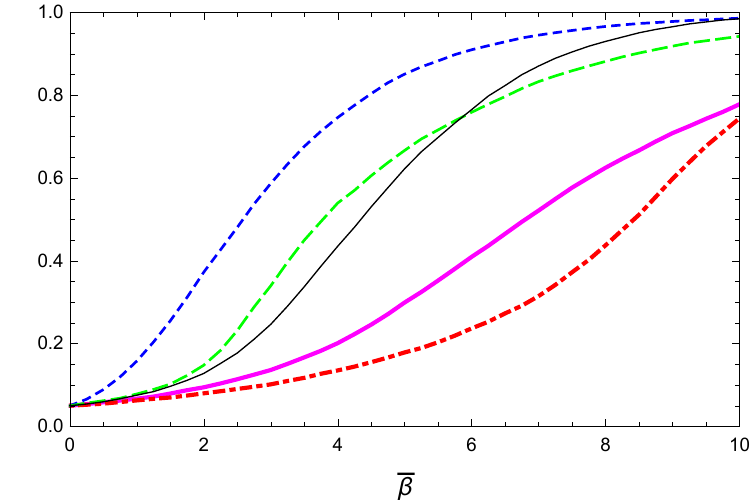}\label{fig:2:1}}
\subfigure[$\bar{\kappa}=0$, $T=20$]{\includegraphics[width=0.30\linewidth]{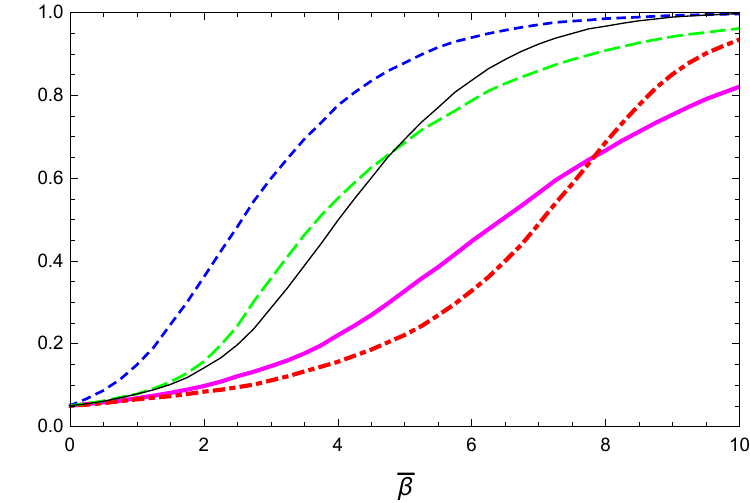}\label{fig:2:2}}
\subfigure[$\bar{\kappa}=0$, $T=50$]{\includegraphics[width=0.30\linewidth]{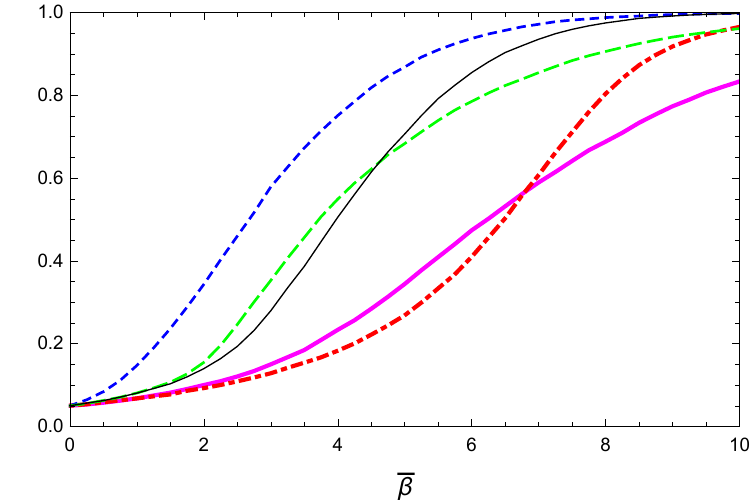}\label{fig:2:3}}\\
\subfigure[$\bar{\kappa}=5$, $T=5$]{\includegraphics[width=0.30\linewidth]{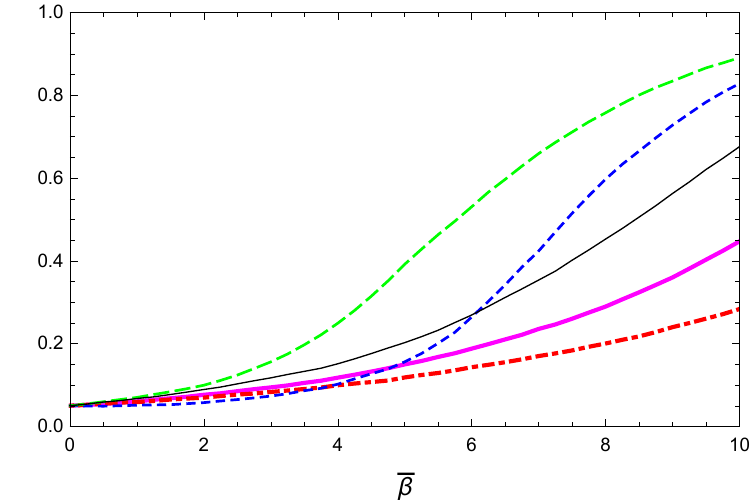}\label{fig:2:4}}
\subfigure[$\bar{\kappa}=5$, $T=20$]{\includegraphics[width=0.30\linewidth]{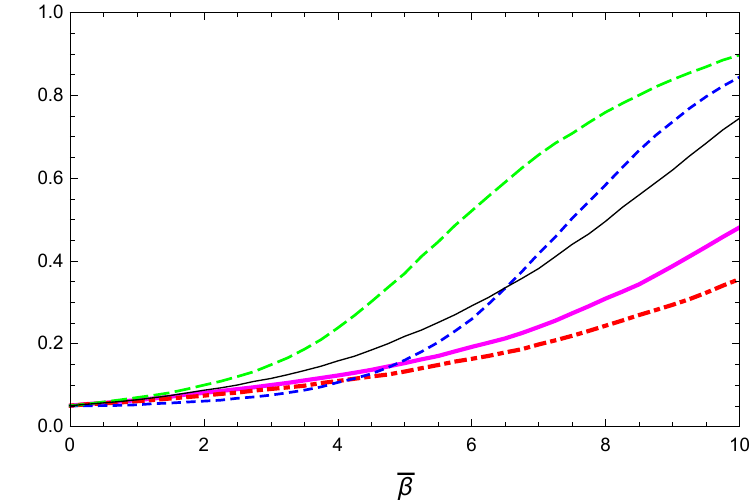}\label{fig:2:5}}
\subfigure[$\bar{\kappa}=5$, $T=50$]{\includegraphics[width=0.30\linewidth]{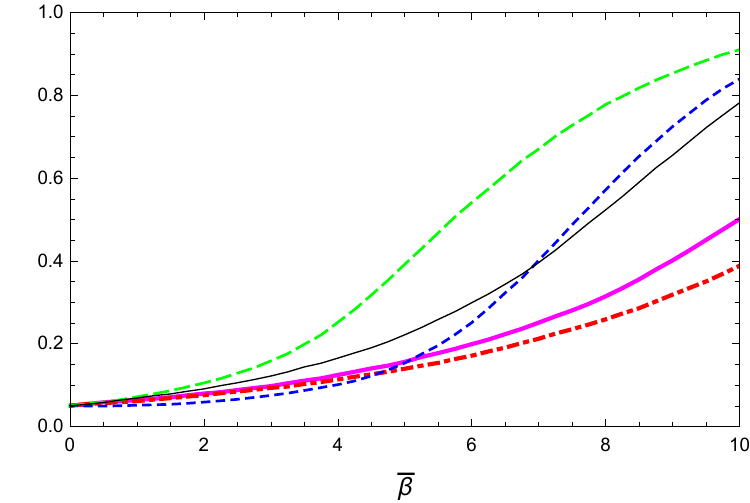}\label{fig:2:6}}\\
\subfigure[$\bar{\kappa}=20$, $T=5$]{\includegraphics[width=0.30\linewidth]{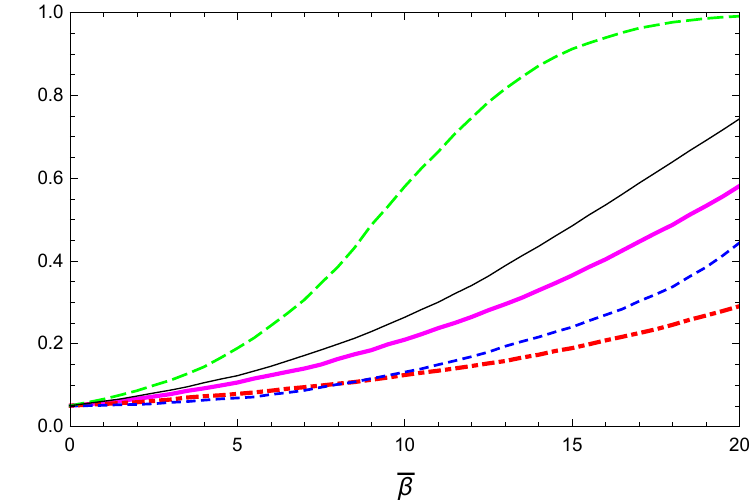}\label{fig:2:7}}
\subfigure[$\bar{\kappa}=20$, $T=20$]{\includegraphics[width=0.30\linewidth]{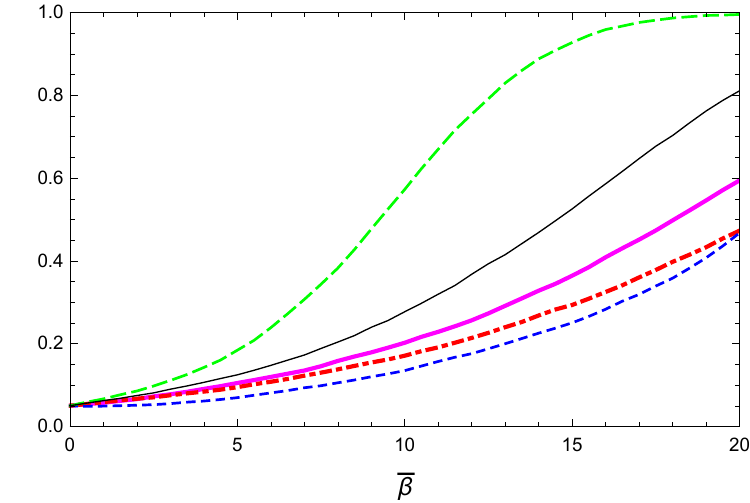}\label{fig:2:8}}
\subfigure[$\bar{\kappa}=20$, $T=50$]{\includegraphics[width=0.30\linewidth]{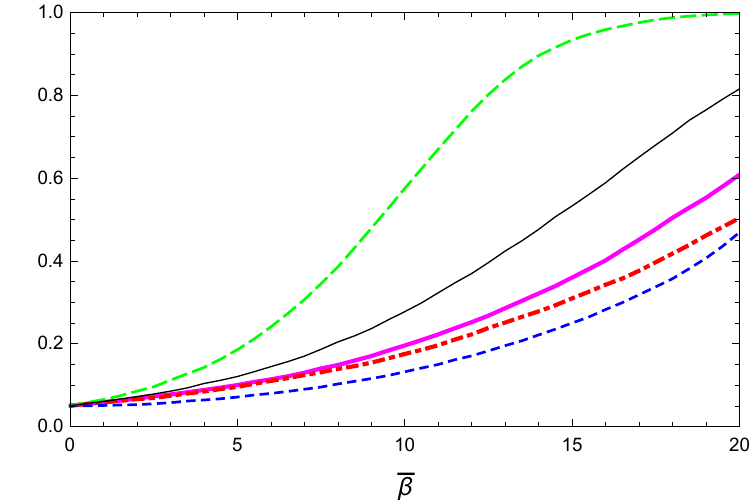}\label{fig:2:9}}
\end{center}%
\caption{Power for SB (continuous time)}
\label{fig1}
\centering
\footnotesize{OLS:$\textcolor{magenta}{\rule[0.25em]{2em}{1.6pt}\ }$,
Bonf. Q:$\textcolor{red}{\rule[0.25em]{0.6em}{1.7pt} \ \mathbf{\cdot} \ \rule[0.25em]{0.6em}{1.7pt} \ }$
RLRT:$\textcolor{blue}{\rule[0.25em]{0.4em}{1.6pt} \ \rule[0.25em]{0.4em}{1.6pt}\ }$, 
Cauchy RT:$\textcolor{green}{\rule[0.25em]{0.8em}{1.6pt} \ \rule[0.25em]{0.8em}{1.6pt}\ }$, 
NP:$\textcolor{black}{\rule[0.25em]{1.9em}{0.5pt}}$}
\end{figure}
\end{landscape}

\begin{landscape}
\begin{figure}[h]%
\begin{center}%
\subfigure[$\bar{\kappa}=0$, $T=5$]{\includegraphics[width=0.30\linewidth]{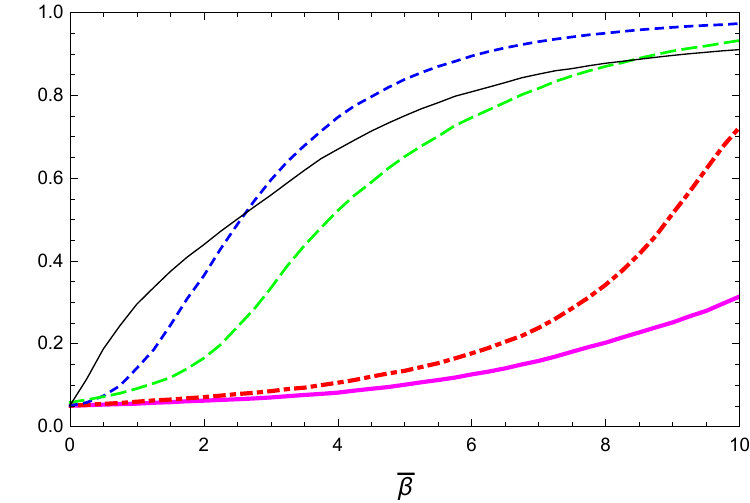}\label{fig:3:1}}
\subfigure[$\bar{\kappa}=0$, $T=20$]{\includegraphics[width=0.30\linewidth]{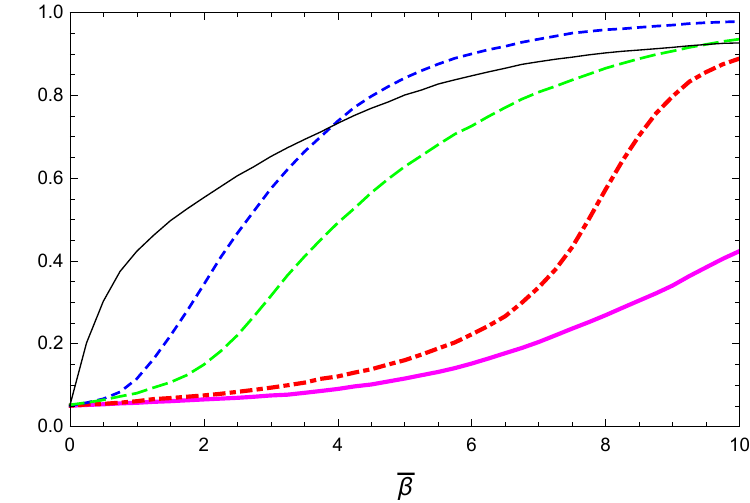}\label{fig:3:2}}
\subfigure[$\bar{\kappa}=0$, $T=50$]{\includegraphics[width=0.30\linewidth]{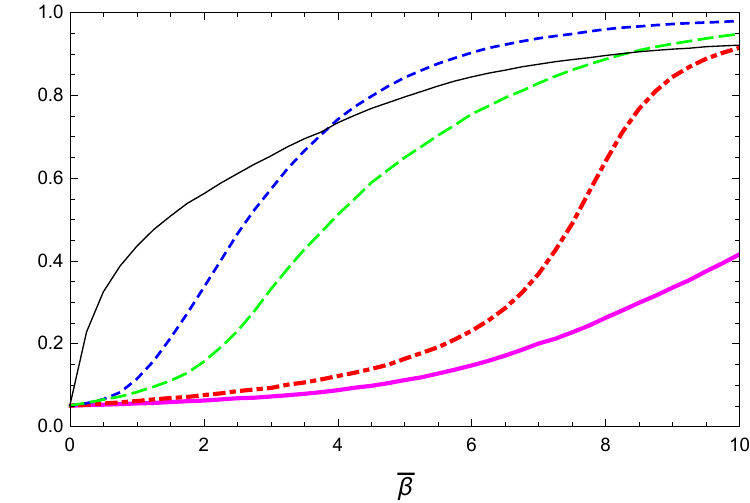}\label{fig:3:3}}\\
\subfigure[$\bar{\kappa}=5$, $T=5$]{\includegraphics[width=0.30\linewidth]{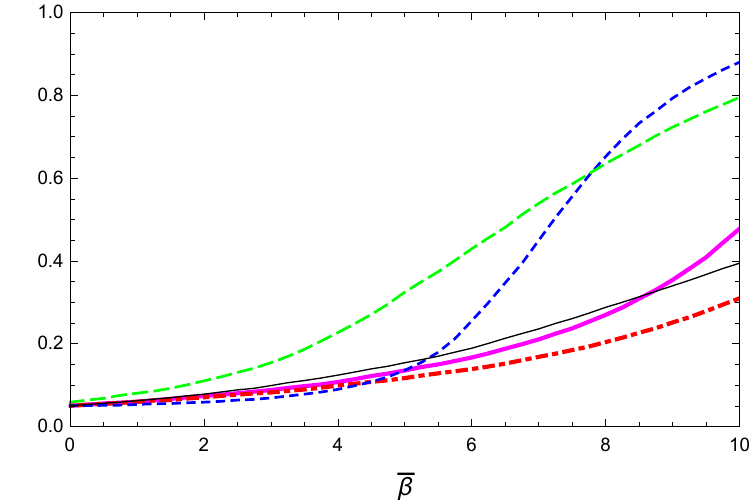}\label{fig:3:4}}
\subfigure[$\bar{\kappa}=5$, $T=20$]{\includegraphics[width=0.30\linewidth]{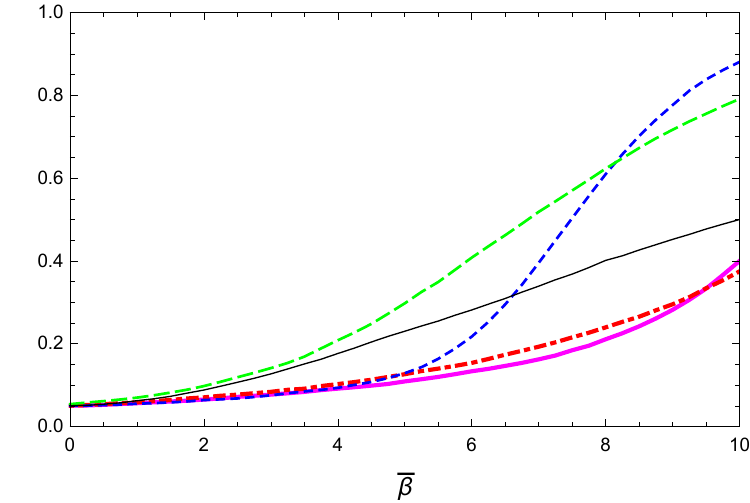}\label{fig:3:5}}
\subfigure[$\bar{\kappa}=5$, $T=50$]{\includegraphics[width=0.30\linewidth]{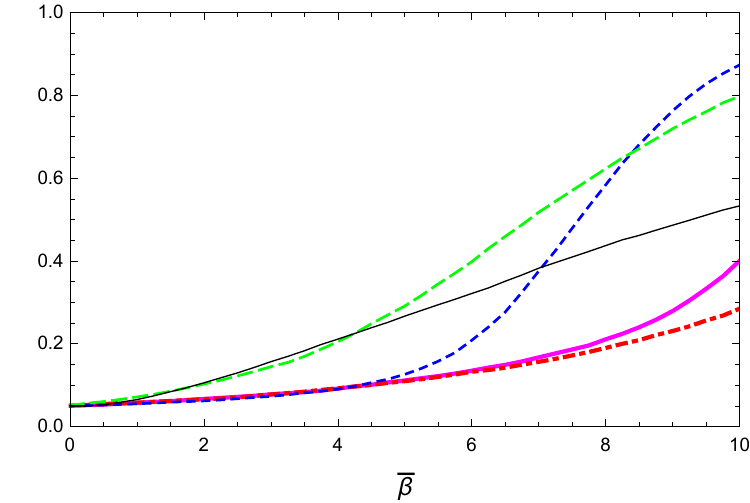}\label{fig:3:6}}\\
\subfigure[$\bar{\kappa}=20$, $T=5$]{\includegraphics[width=0.30\linewidth]{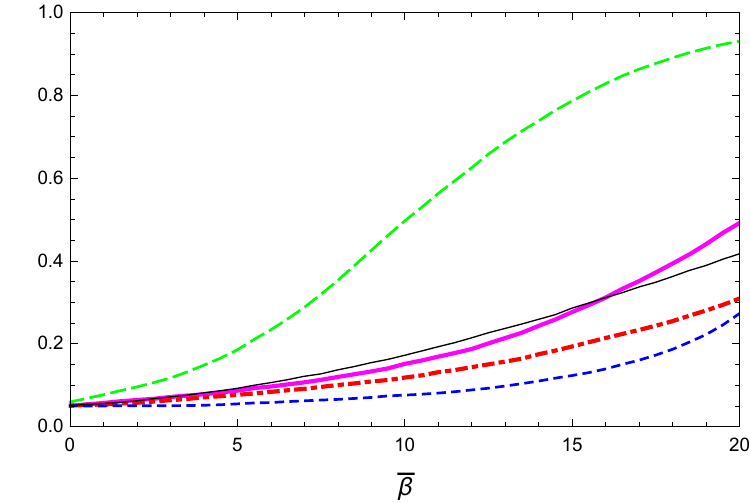}\label{fig:3:7}}
\subfigure[$\bar{\kappa}=20$, $T=20$]{\includegraphics[width=0.30\linewidth]{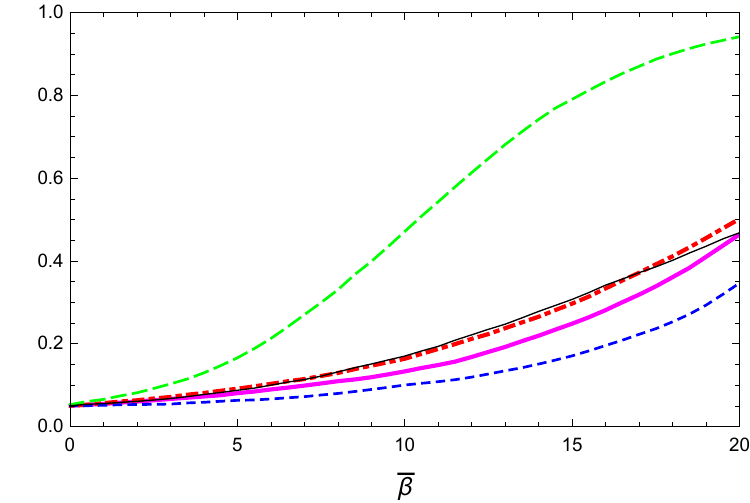}\label{fig:3:8}}
\subfigure[$\bar{\kappa}=20$, $T=50$]{\includegraphics[width=0.30\linewidth]{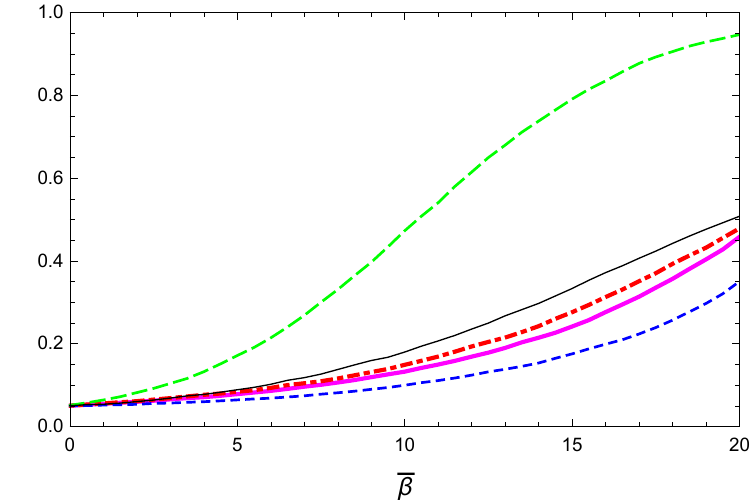}\label{fig:3:9}}
\end{center}%
\caption{Power for GBM (continuous time)}
\label{fig2}
\centering
\footnotesize{OLS:$\textcolor{magenta}{\rule[0.25em]{2em}{1.6pt}\ }$,
Bonf. Q:$\textcolor{red}{\rule[0.25em]{0.6em}{1.7pt} \ \mathbf{\cdot} \ \rule[0.25em]{0.6em}{1.7pt} \ }$
RLRT:$\textcolor{blue}{\rule[0.25em]{0.4em}{1.6pt} \ \rule[0.25em]{0.4em}{1.6pt}\ }$, 
Cauchy RT:$\textcolor{green}{\rule[0.25em]{0.8em}{1.6pt} \ \rule[0.25em]{0.8em}{1.6pt}\ }$, 
NP:$\textcolor{black}{\rule[0.25em]{1.9em}{0.5pt}}$}
\end{figure}
\end{landscape}

\begin{landscape}
\begin{figure}[h]%
\begin{center}%
\subfigure[$\bar{\kappa}=0$, $T=5$]{\includegraphics[width=0.30\linewidth]{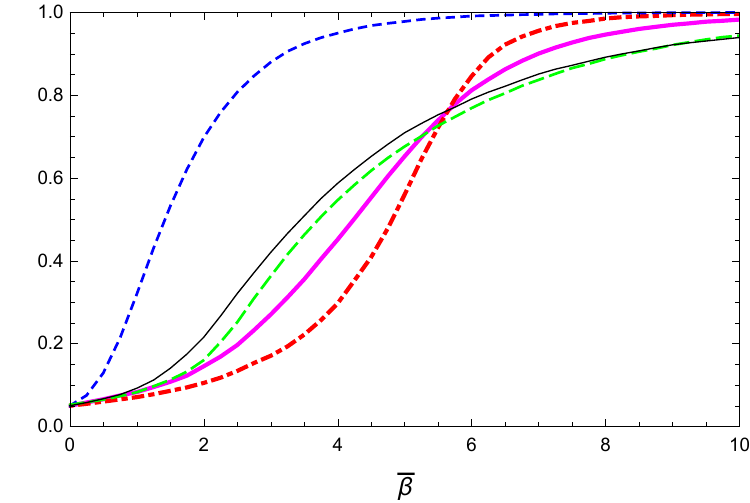}\label{fig:4:1}}
\subfigure[$\bar{\kappa}=0$, $T=20$]{\includegraphics[width=0.30\linewidth]{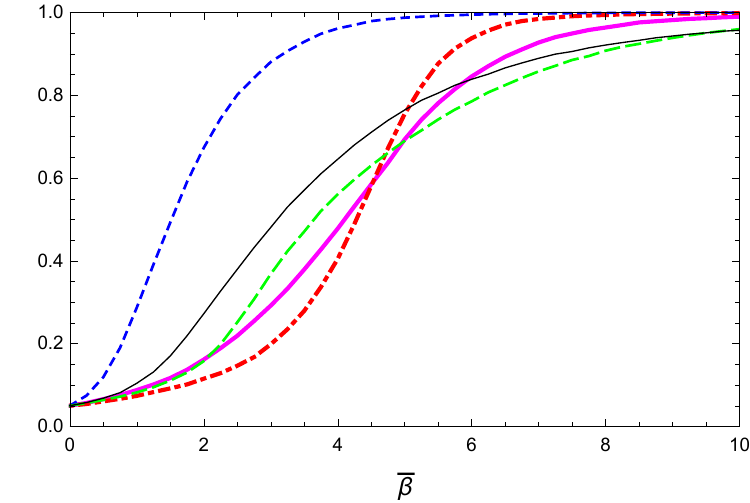}\label{fig:4:2}}
\subfigure[$\bar{\kappa}=0$, $T=50$]{\includegraphics[width=0.30\linewidth]{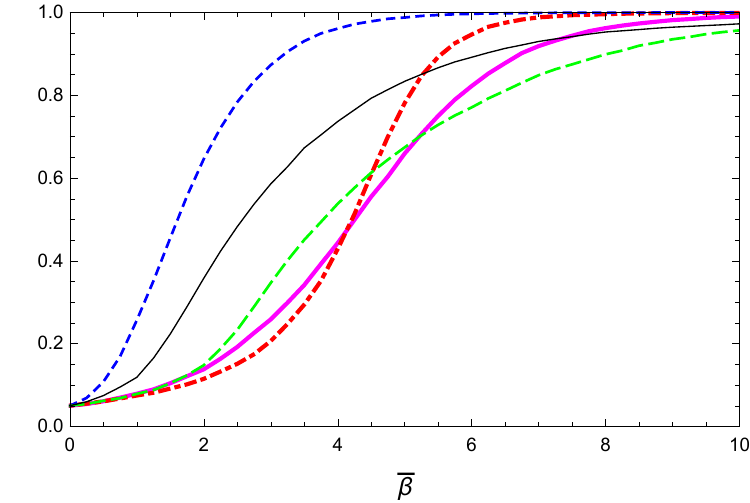}\label{fig:4:3}}\\
\subfigure[$\bar{\kappa}=5$, $T=5$]{\includegraphics[width=0.30\linewidth]{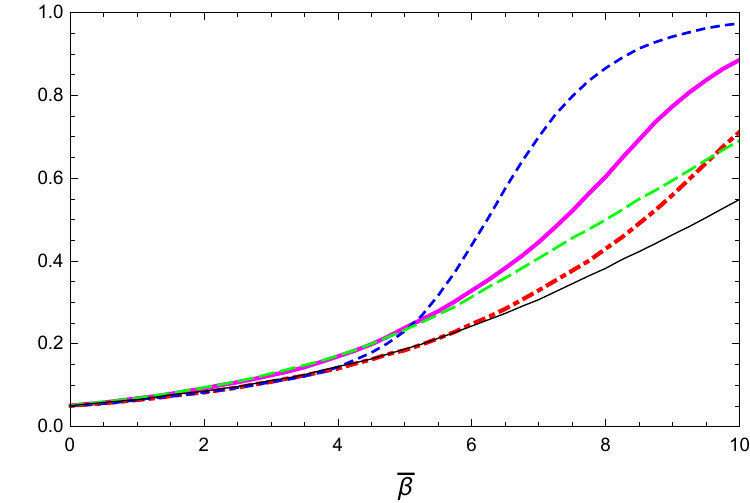}\label{fig:4:4}}
\subfigure[$\bar{\kappa}=5$, $T=20$]{\includegraphics[width=0.30\linewidth]{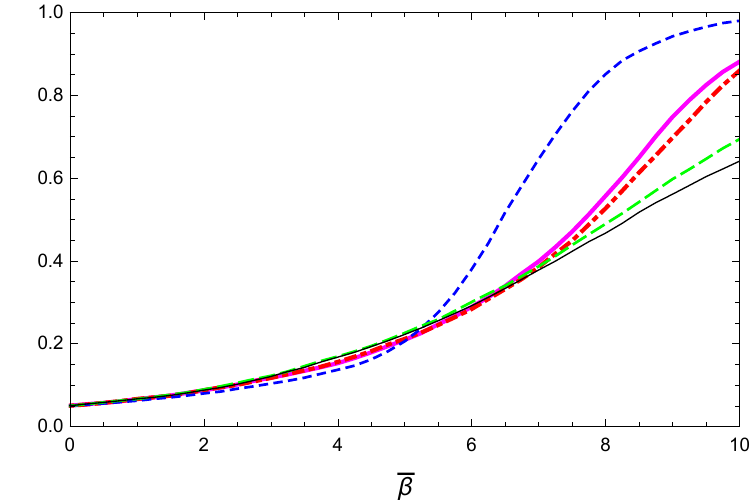}\label{fig:4:5}}
\subfigure[$\bar{\kappa}=5$, $T=50$]{\includegraphics[width=0.30\linewidth]{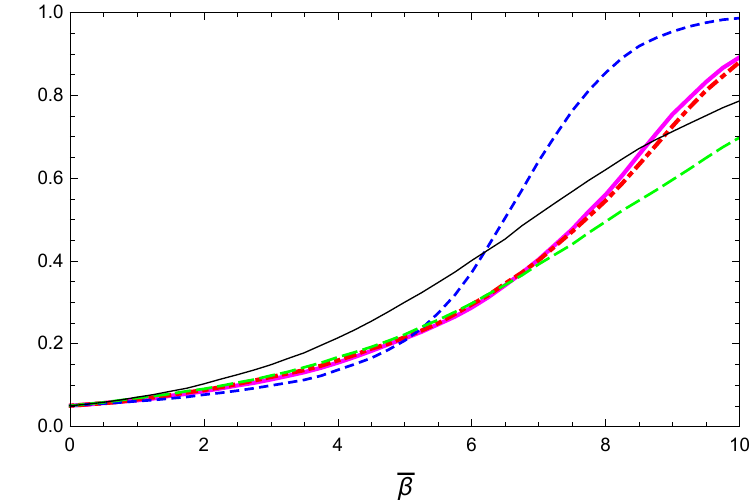}\label{fig:4:6}}\\
\subfigure[$\bar{\kappa}=20$, $T=5$]{\includegraphics[width=0.30\linewidth]{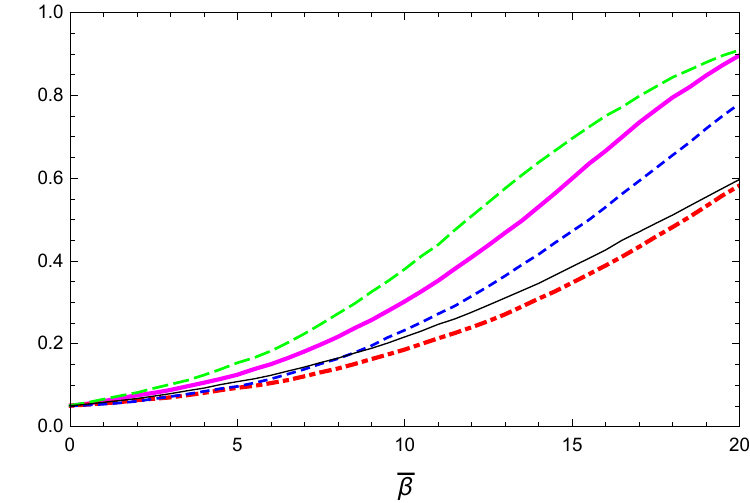}\label{fig:4:7}}
\subfigure[$\bar{\kappa}=20$, $T=20$]{\includegraphics[width=0.30\linewidth]{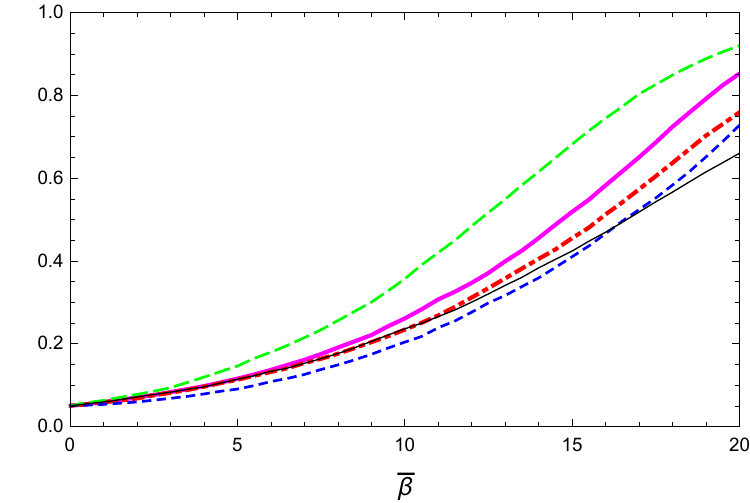}\label{fig:4:8}}
\subfigure[$\bar{\kappa}=20$, $T=50$]{\includegraphics[width=0.30\linewidth]{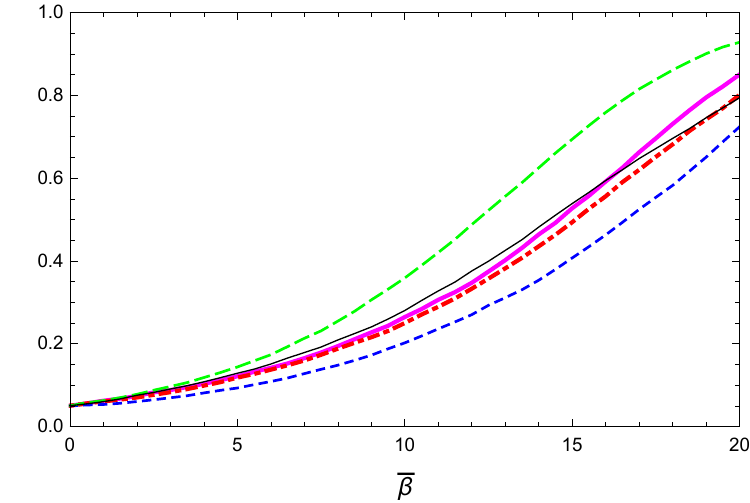}\label{fig:4:9}}
\end{center}%
\caption{Power for RS (continuous time)}
\label{fig3}
\centering
\footnotesize{OLS:$\textcolor{magenta}{\rule[0.25em]{2em}{1.6pt}\ }$,
Bonf. Q:$\textcolor{red}{\rule[0.25em]{0.6em}{1.7pt} \ \mathbf{\cdot} \ \rule[0.25em]{0.6em}{1.7pt} \ }$
RLRT:$\textcolor{blue}{\rule[0.25em]{0.4em}{1.6pt} \ \rule[0.25em]{0.4em}{1.6pt}\ }$, 
Cauchy RT:$\textcolor{green}{\rule[0.25em]{0.8em}{1.6pt} \ \rule[0.25em]{0.8em}{1.6pt}\ }$, 
NP:$\textcolor{black}{\rule[0.25em]{1.9em}{0.5pt}}$}
\end{figure}
\end{landscape}

\begin{landscape}
\begin{figure}[h]%
\begin{center}%
\subfigure[$\bar{\kappa}=0$, $T=60$]{\includegraphics[width=0.30\linewidth]{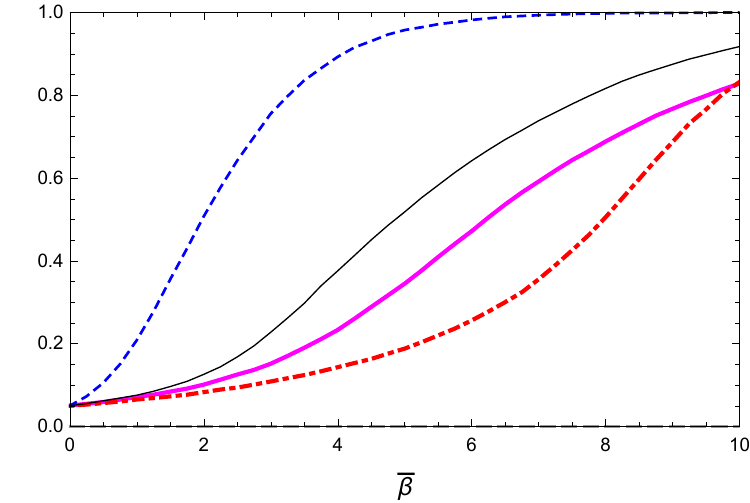}\label{fig:6:1}}
\subfigure[$\bar{\kappa}=0$, $T=240$]{\includegraphics[width=0.30\linewidth]{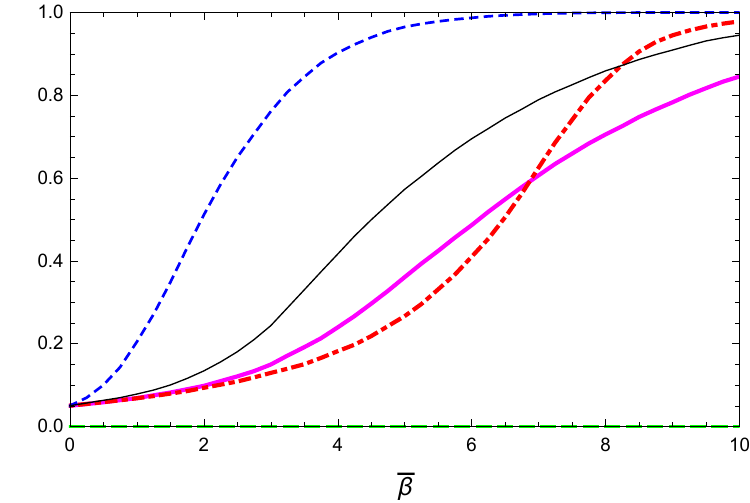}\label{fig:6:2}}
\subfigure[$\bar{\kappa}=0$, $T=600$]{\includegraphics[width=0.30\linewidth]{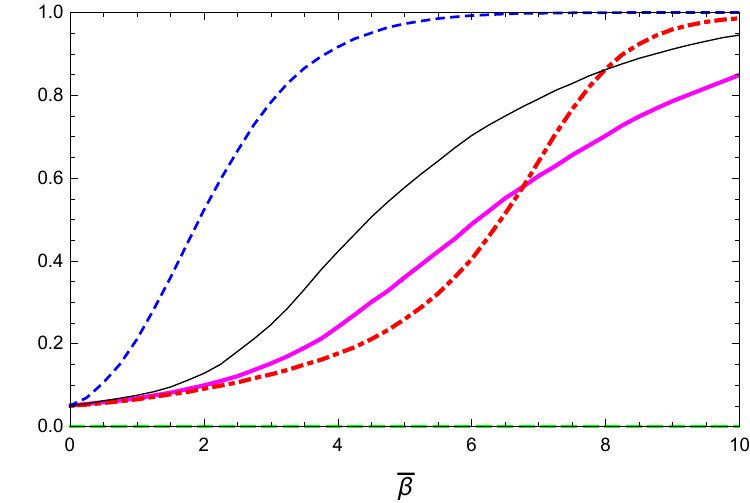}\label{fig:6:3}}\\
\subfigure[$\bar{\kappa}=5$, $T=60$]{\includegraphics[width=0.30\linewidth]{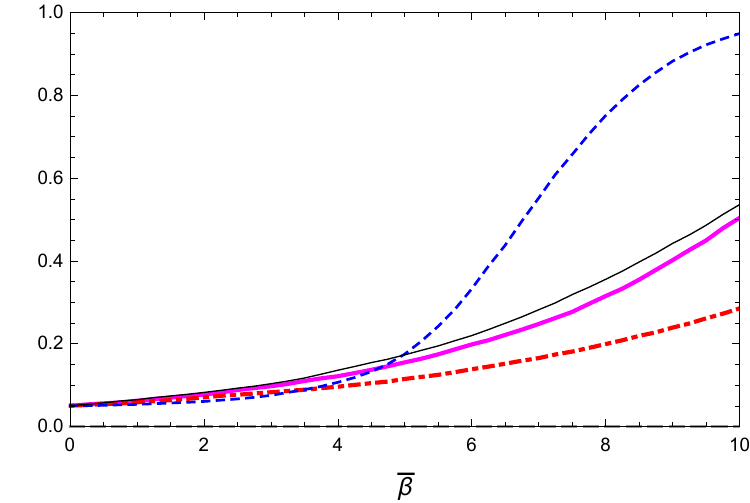}\label{fig:6:4}}
\subfigure[$\bar{\kappa}=5$, $T=240$]{\includegraphics[width=0.30\linewidth]{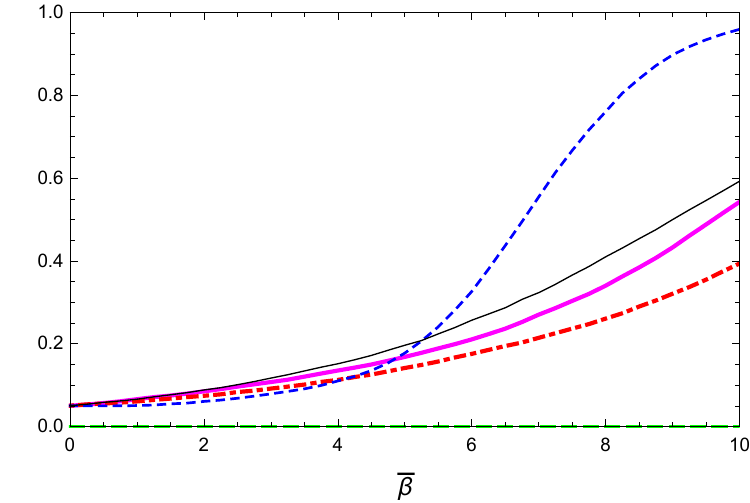}\label{fig:6:5}}
\subfigure[$\bar{\kappa}=5$, $T=600$]{\includegraphics[width=0.30\linewidth]{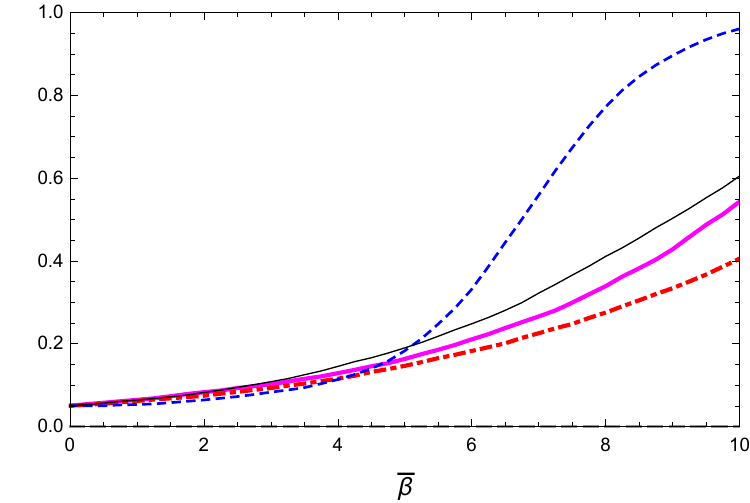}\label{fig:6:6}}\\
\subfigure[$\bar{\kappa}=20$, $T=60$]{\includegraphics[width=0.30\linewidth]{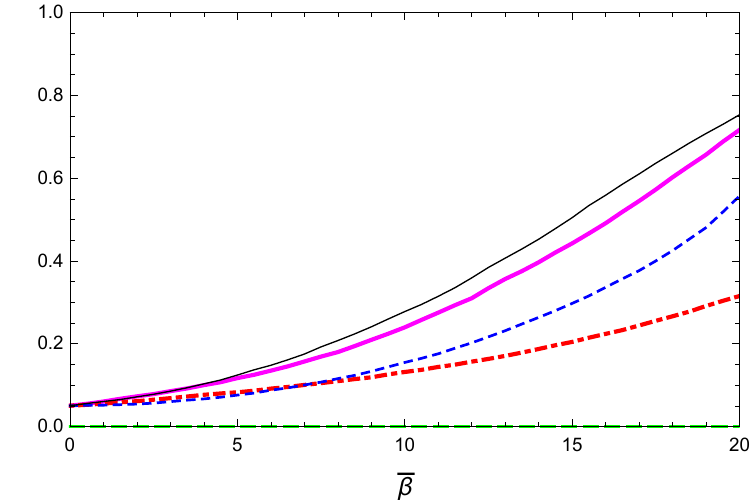}\label{fig:6:7}}
\subfigure[$\bar{\kappa}=20$, $T=240$]{\includegraphics[width=0.30\linewidth]{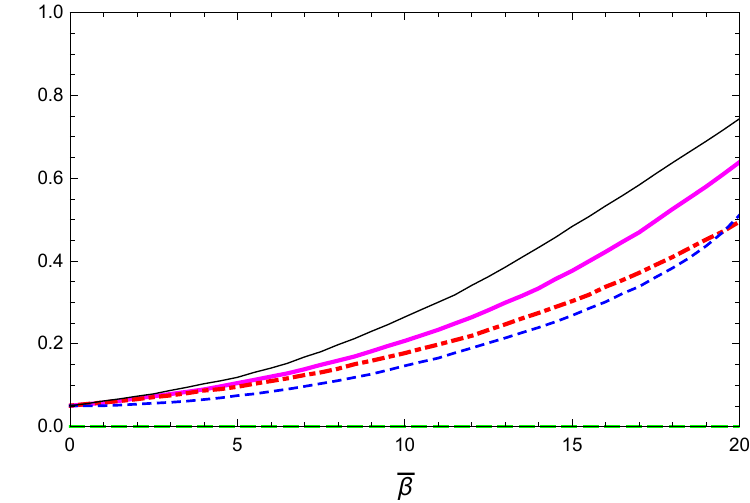}\label{fig:6:8}}
\subfigure[$\bar{\kappa}=20$, $T=600$]{\includegraphics[width=0.30\linewidth]{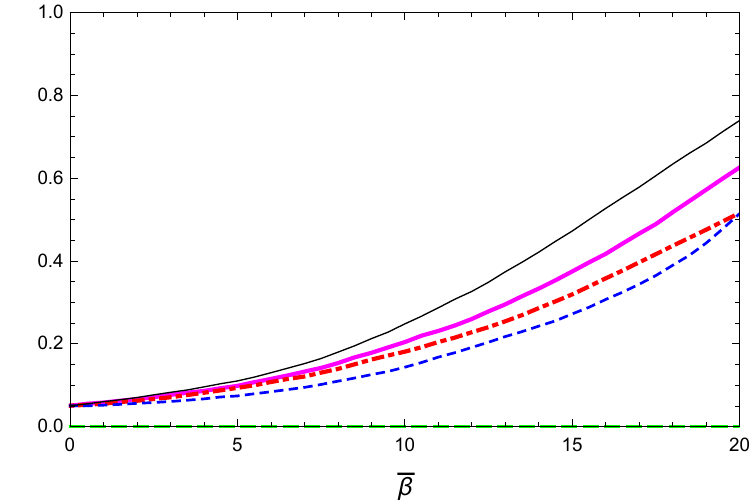}\label{fig:6:9}}
\end{center}%
\caption{Power for SB (discrete time)}
\label{fig4}
\centering
\footnotesize{OLS:$\textcolor{magenta}{\rule[0.25em]{2em}{1.6pt}\ }$,
Bonf. Q:$\textcolor{red}{\rule[0.25em]{0.6em}{1.7pt} \ \mathbf{\cdot} \ \rule[0.25em]{0.6em}{1.7pt} \ }$
RLRT:$\textcolor{blue}{\rule[0.25em]{0.4em}{1.6pt} \ \rule[0.25em]{0.4em}{1.6pt}\ }$, 
NP:$\textcolor{black}{\rule[0.25em]{1.9em}{0.5pt}}$}
\end{figure}
\end{landscape}

\begin{landscape}
\begin{figure}[h]%
\begin{center}%
\subfigure[$\bar{\kappa}=0$, $T=60$]{\includegraphics[width=0.30\linewidth]{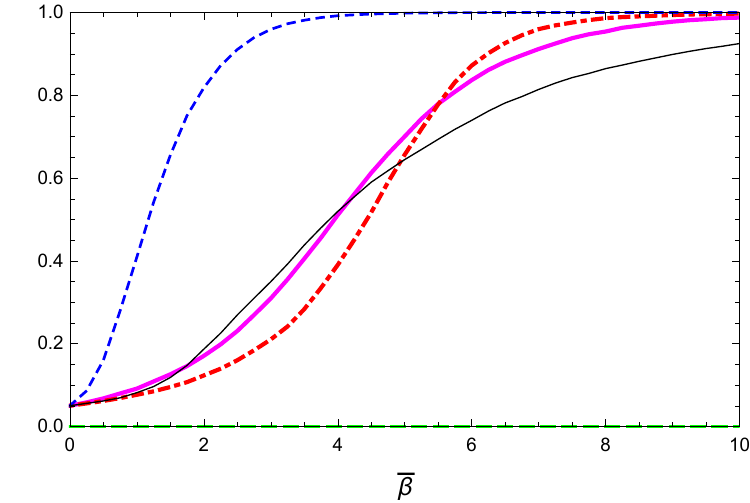}\label{fig:9:1}}
\subfigure[$\bar{\kappa}=0$, $T=240$]{\includegraphics[width=0.30\linewidth]{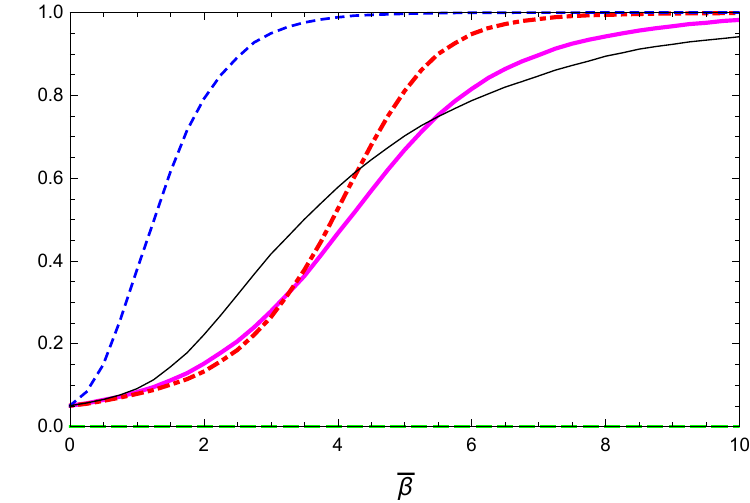}\label{fig:9:2}}
\subfigure[$\bar{\kappa}=0$, $T=600$]{\includegraphics[width=0.30\linewidth]{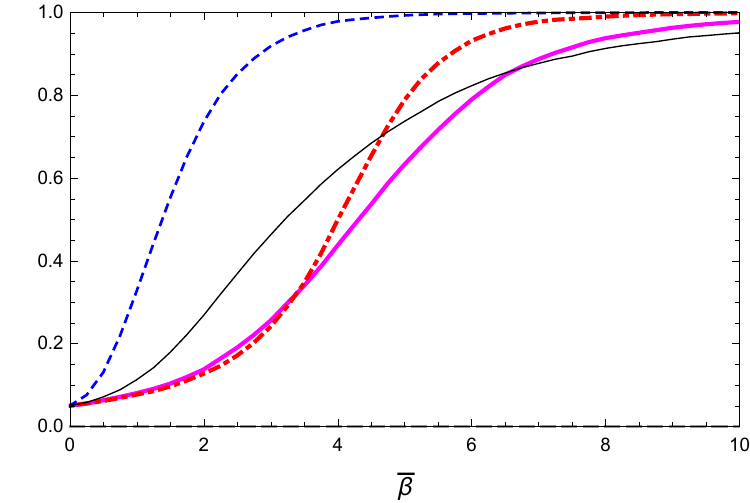}\label{fig:9:3}}\\
\subfigure[$\bar{\kappa}=5$, $T=60$]{\includegraphics[width=0.30\linewidth]{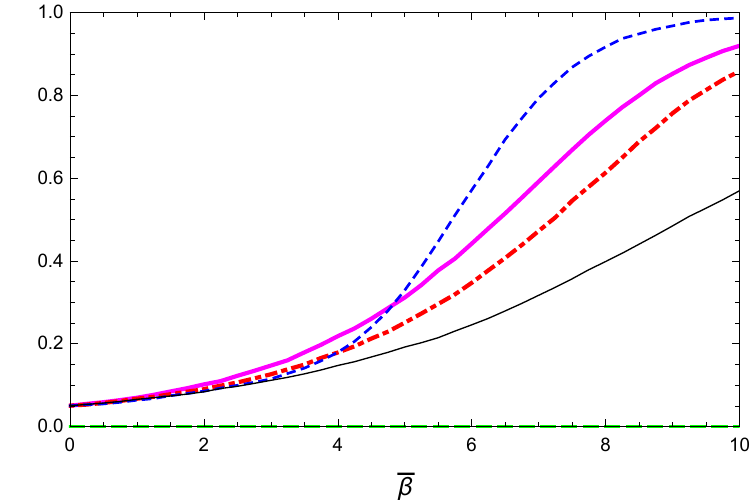}\label{fig:9:4}}
\subfigure[$\bar{\kappa}=5$, $T=240$]{\includegraphics[width=0.30\linewidth]{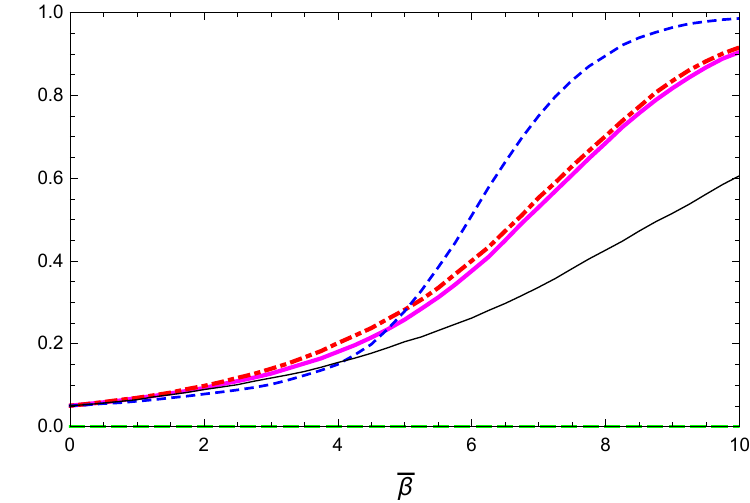}\label{fig:9:5}}
\subfigure[$\bar{\kappa}=5$, $T=600$]{\includegraphics[width=0.30\linewidth]{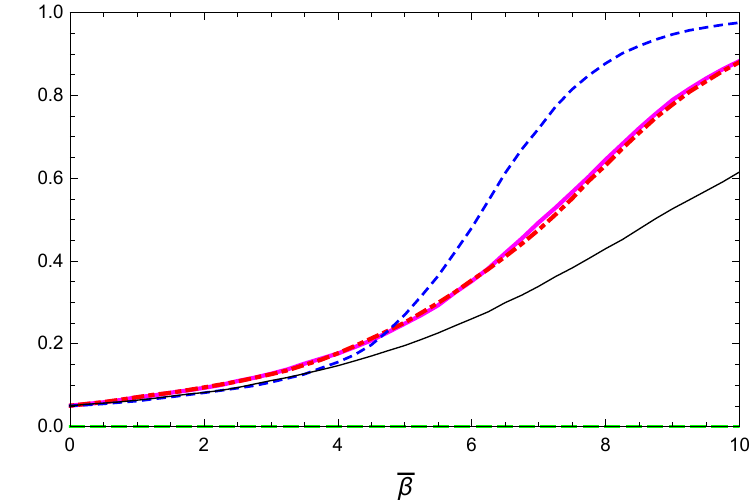}\label{fig:9:6}}\\
\subfigure[$\bar{\kappa}=20$, $T=60$]{\includegraphics[width=0.30\linewidth]{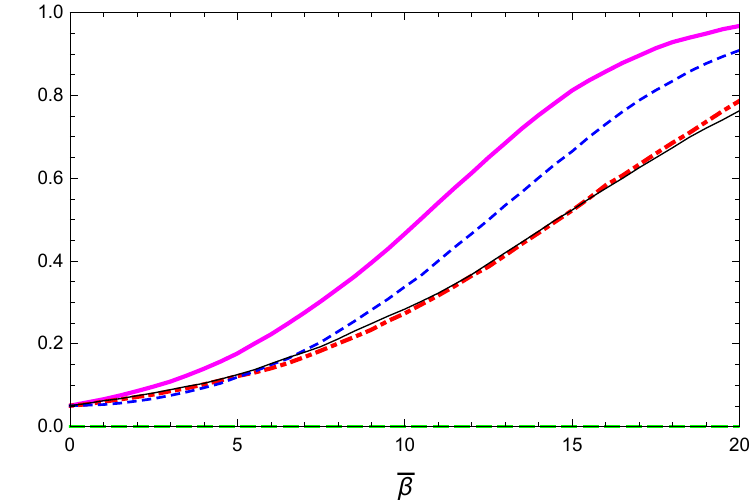}\label{fig:9:7}}
\subfigure[$\bar{\kappa}=20$, $T=240$]{\includegraphics[width=0.30\linewidth]{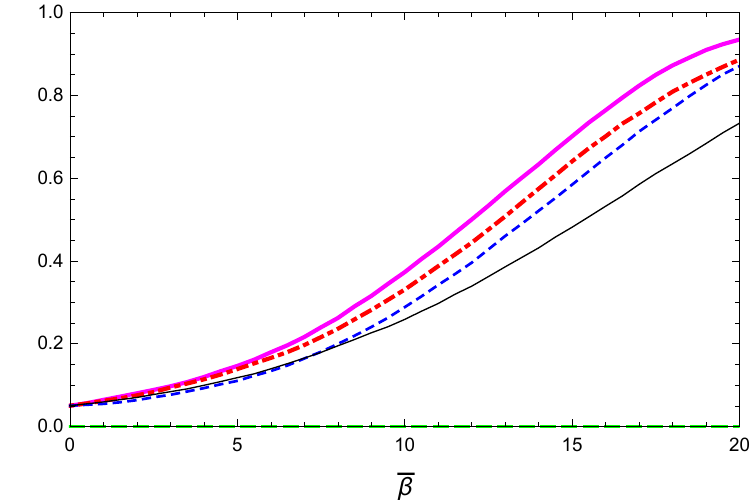}\label{fig:9:8}}
\subfigure[$\bar{\kappa}=20$, $T=600$]{\includegraphics[width=0.30\linewidth]{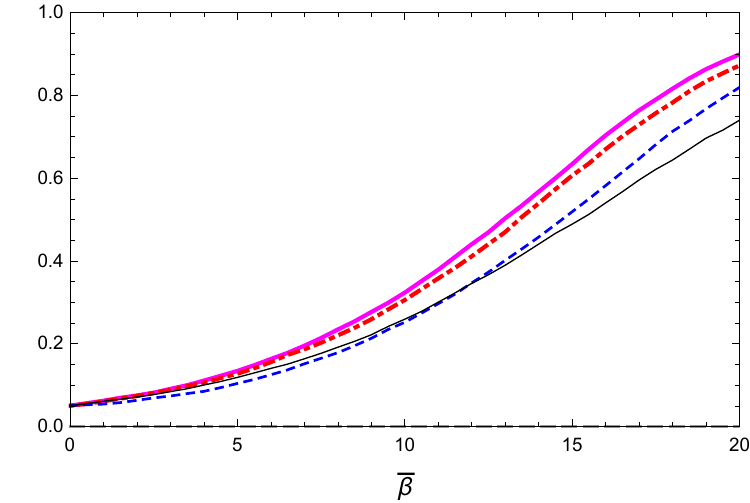}\label{fig:9:9}}
\end{center}%
\caption{Power for GARCH with $\alpha=0.9$ and $\theta=0.1$ (discrete time)}
\label{fig5}
\centering
\footnotesize{OLS:$\textcolor{magenta}{\rule[0.25em]{2em}{1.6pt}\ }$,
Bonf. Q:$\textcolor{red}{\rule[0.25em]{0.6em}{1.7pt} \ \mathbf{\cdot} \ \rule[0.25em]{0.6em}{1.7pt} \ }$
RLRT:$\textcolor{blue}{\rule[0.25em]{0.4em}{1.6pt} \ \rule[0.25em]{0.4em}{1.6pt}\ }$, 
NP:$\textcolor{black}{\rule[0.25em]{1.9em}{0.5pt}}$}
\end{figure}
\end{landscape}

\begin{landscape}
\begin{figure}[h]%
\begin{center}%
\subfigure[$\bar{\kappa}=0$, $T=60$]{\includegraphics[width=0.30\linewidth]{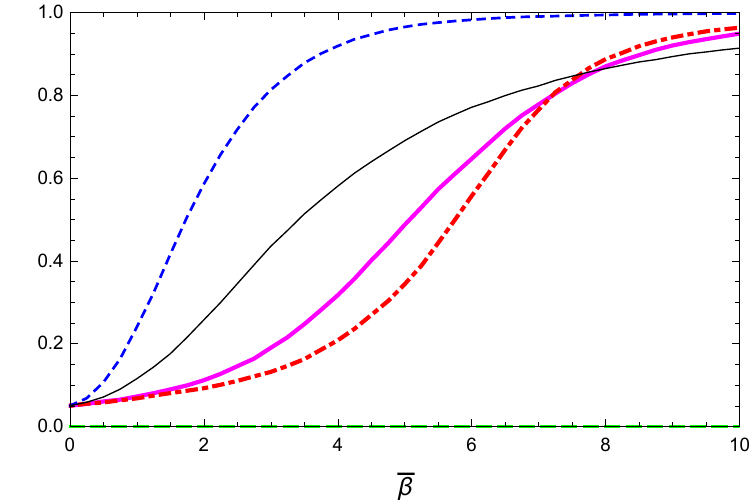}\label{fig:10:1}}
\subfigure[$\bar{\kappa}=0$, $T=240$]{\includegraphics[width=0.30\linewidth]{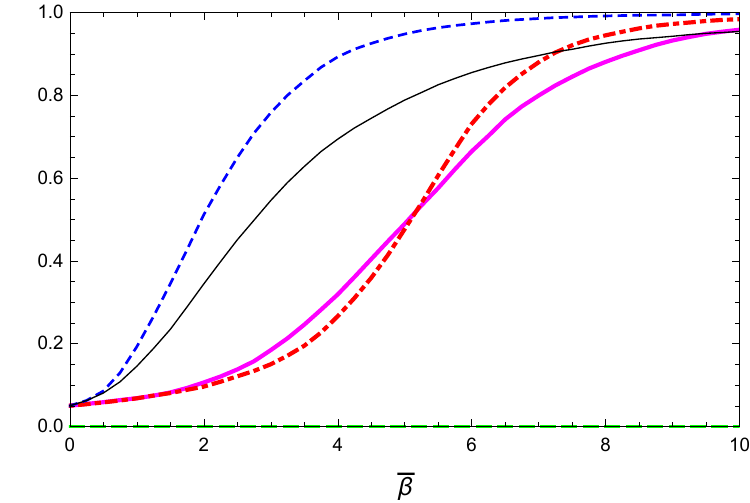}\label{fig:10:2}}
\subfigure[$\bar{\kappa}=0$, $T=600$]{\includegraphics[width=0.30\linewidth]{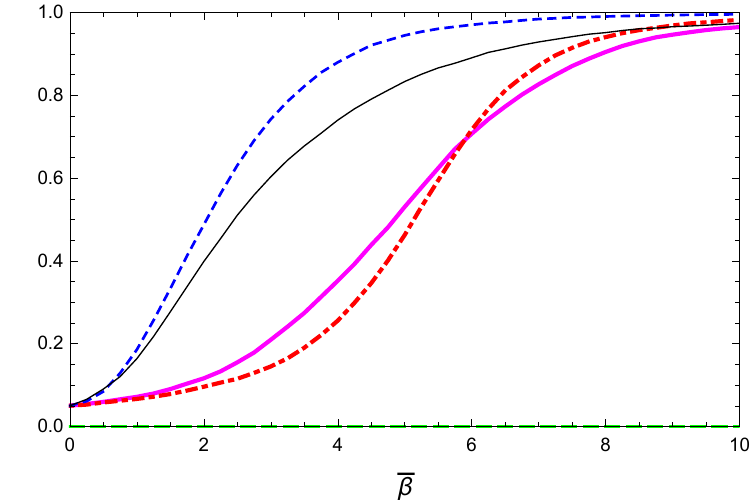}\label{fig:10:3}}\\
\subfigure[$\bar{\kappa}=5$, $T=60$]{\includegraphics[width=0.30\linewidth]{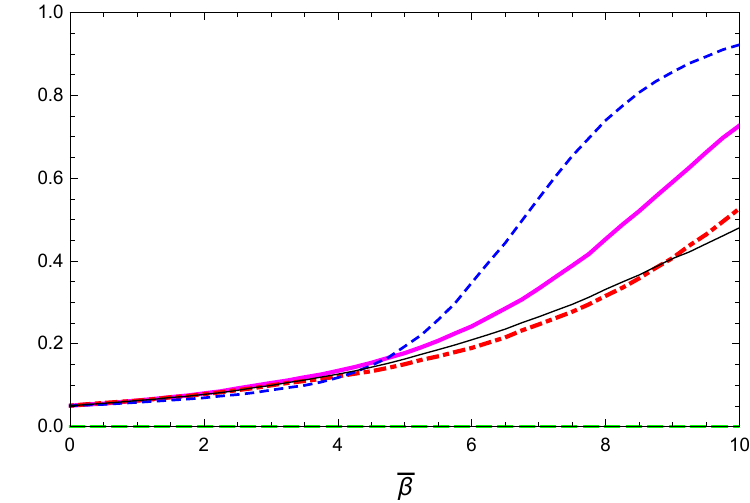}\label{fig:10:4}}
\subfigure[$\bar{\kappa}=5$, $T=240$]{\includegraphics[width=0.30\linewidth]{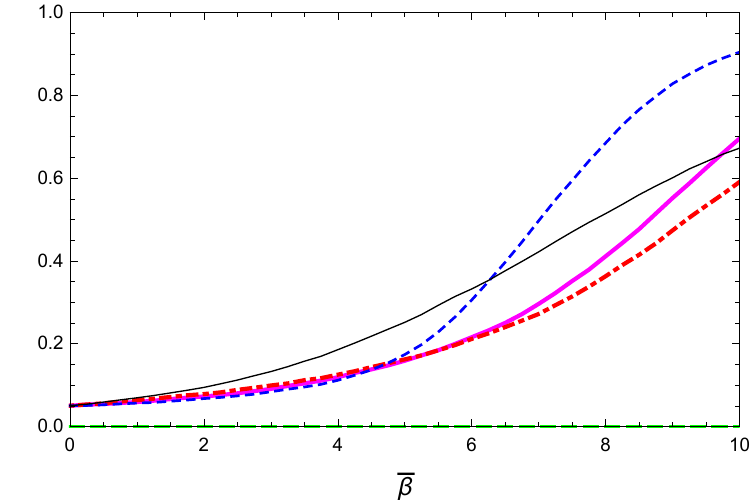}\label{fig:10:5}}
\subfigure[$\bar{\kappa}=5$, $T=600$]{\includegraphics[width=0.30\linewidth]{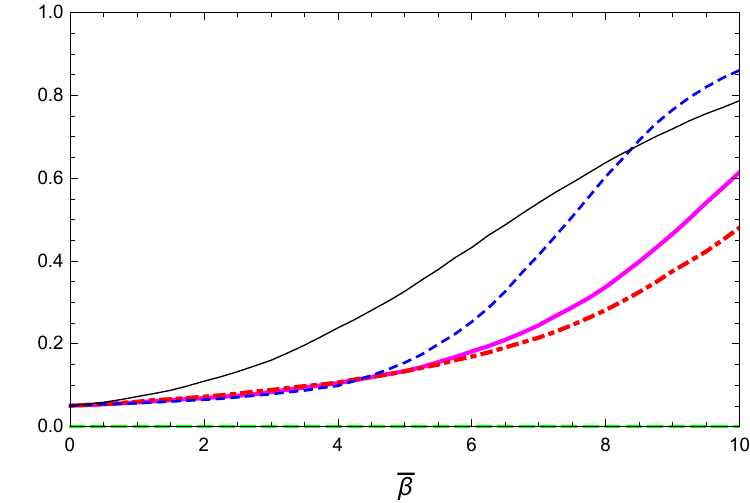}\label{fig:10:6}}\\
\subfigure[$\bar{\kappa}=20$, $T=60$]{\includegraphics[width=0.30\linewidth]{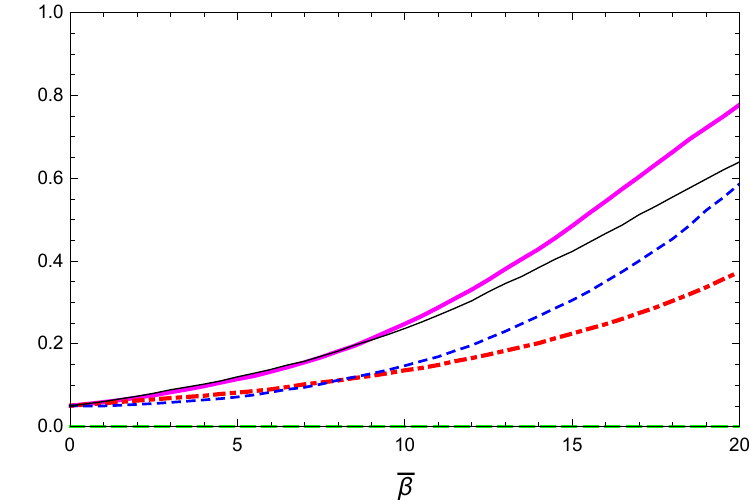}\label{fig:10:7}}
\subfigure[$\bar{\kappa}=20$, $T=240$]{\includegraphics[width=0.30\linewidth]{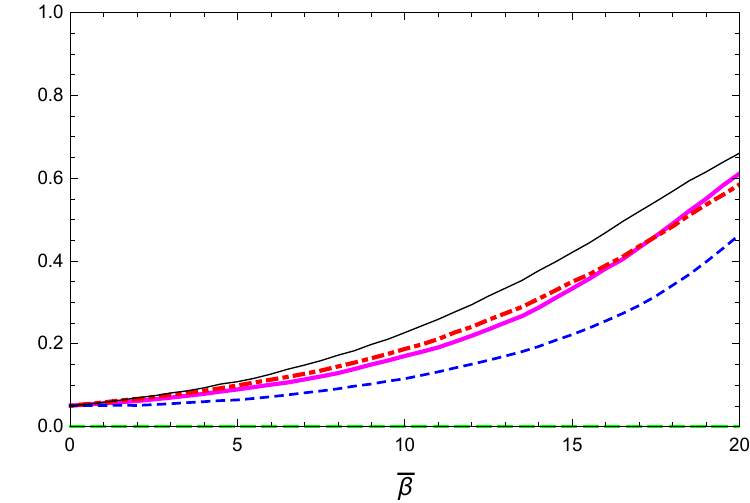}\label{fig:10:8}}
\subfigure[$\bar{\kappa}=20$, $T=600$]{\includegraphics[width=0.30\linewidth]{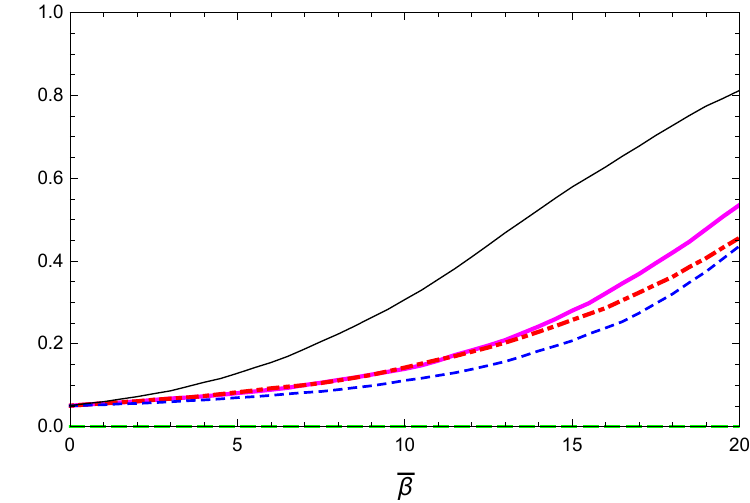}\label{fig:10:9}}
\end{center}%
\caption{Power for GARCH with $\alpha=0.1$ and $\theta=0.9$ (discrete time)}
\label{fig6}
\centering
\footnotesize{OLS:$\textcolor{magenta}{\rule[0.25em]{2em}{1.6pt}\ }$,
Bonf. Q:$\textcolor{red}{\rule[0.25em]{0.6em}{1.7pt} \ \mathbf{\cdot} \ \rule[0.25em]{0.6em}{1.7pt} \ }$
RLRT:$\textcolor{blue}{\rule[0.25em]{0.4em}{1.6pt} \ \rule[0.25em]{0.4em}{1.6pt}\ }$, 
NP:$\textcolor{black}{\rule[0.25em]{1.9em}{0.5pt}}$}
\end{figure}
\end{landscape}

\pagebreak

\newpage

\clearpage

\appendix

\setcounter{equation}{0}
\setcounter{figure}{0}
\renewcommand{\theequation}{S.\arabic{equation}}%
\renewcommand{\thelemma}{S.\arabic{lemma}}

\renewcommand{\theproposition}{S.\arabic{proposition}}

\renewcommand{\thefigure}{S.\arabic{figure}}%

\begin{center}
\vspace*{2in}{\Huge Online supplementary appendix to}

\bigskip

{\Huge \bigskip}

{\Large New robust inference for predictive regressions}

\bigskip

\author{
Rustam Ibragimov$^{a,d}$, Jihyun Kim$^b$, Anton Skrobotov$^{c,d}$ \\
{\small {$^{a}$ Imperial College Business School, Imperial College London}}\\
{\small {$^{b}$ Sungkyunkwan University and Toulouse School of Economics}}\\
{\small {$^{c}$ Russian Presidential Academy of National Economy and Public Administration}}\\
{\small {$^{d}$ Center for Econometrics and Business Analytics, St. Petersburg University}}
}
\end{center}

\bigskip

This supplement includes four appendices, Appendices A, B, C, and D. In Appendix A, we consider a class of nonlinear IV estimators, and discuss why the Cauchy estimator, which is a special case of nonlinear IV estimators, is useful in inference problems. Appendix B presents some useful lemmas, Lemmas S.1-S.8, and their proofs; Appendix C provides the proofs of the main results in the paper; and Appendix D provides some additional simulation results on finite sample performance of inference approaches dealt with.

\setcounter{section}{0}
\def\thesection{\arabic{section}}%

\renewcommand{\thesection}{A.\arabic{section}}
\setcounter{section}{1}%

\renewcommand{\thetable}{S.\arabic{table}}
\setcounter{table}{0}%

\renewcommand{\thepage}{[S.\arabic{page}]}
\setcounter{page}{1}%

\setcounter{footnote}{0}%


\section*{Appendix A: Nonlinear IV Approaches}
As mentioned in the main paper, the Cauchy estimator $\check{\beta}$ is the special case of a class of nonlinear IV estimators $\tilde{\beta}(\gamma)$ with $\gamma(\cdot)=sign(\cdot)$, where
\[
\tilde{\beta}(\gamma) =  \left(\sum_{t=1}^T  \gamma(x_{t-1})x_{t-1} \right)^{-1}\sum_{t=1}^T \gamma(x_{t-1}) y_{t}
\]
for some function $\gamma:\mathbb{R}\to\mathbb{R}$. Clearly, $\tilde{\beta}(sign)=\check{\beta}$. Moreover, if we let $\iota$ be an identity function:  $\iota(x)=x$, then $\tilde{\beta}(\iota)$ becomes the OLS estimator $\hat{\beta}$. In this section, we discuss why the choice of $\gamma(\cdot)=sign(\cdot)$, as in the Cauchy estimator, is important and useful when considering the class of nonlinear IV estimators in an inference problem. To explain the idea, we focus on the issues about the predictor $x_t$, and assume that $v_t=\sigma$ for all $t\geq1$.

\subsection*{A.1 Nonstationary Predictor}
The asymptotics for (near) unit root processes under various transformations are well known (see \cite{park2001nonlinear} and \cite{park2003strong}). For instance, let $x_t$ be a unit root process and $(x_t , \sum_{s=1}^t \varepsilon_s)$ satisfy the functional CLT with a limiting bivariate Brownian motion $(X,W)$. It then follows from \cite{park2001nonlinear} that for regularly integrable functions $f$ and $g^2$
\begin{equation}\label{pp1}
\begin{aligned}
\frac{1}{T^{1/2}}\sum_{t=1}^T f(x_{t-1})
&\to_d L_X(1,0) \int_{-\infty}^{\infty} f(x)dx,\\
\frac{1}{T^{1/4}}\sum_{t=1}^T g(x_{t-1})\varepsilon_t
&\to_d \left(L_X(1,0) \int_{-\infty}^{\infty} g^2(x)dx\right)^{1/2}Z(1),
\end{aligned}
\end{equation}
where $L_X(1,0)$ is the local time at the origin of $X$ and $Z$ is a Brownian motion independent of $X$. In particular, $\left(L_X(1,0) \int_{-\infty}^{\infty} g^2(x)dx\right)^{1/2}Z(1) =_d \mathbb{MN}\left(0,L_X(1,0) \int_{-\infty}^{\infty} g^2(x)dx\right)$, where $\mathbb{MN}$ is a mixed normal distribution.

On the other hand, if $f$ and $g^2$ are asymptotically homogeneous functions, then
\begin{equation}\label{pp2}
\begin{aligned}
\frac{1}{T} f_\nu(T^{1/2})\sum_{t=1}^T f(x_{t-1})
&\to_d \int_0^1 f_H(X(r))dr,\\
\frac{1}{T^{1/2} g_\nu(T^{1/2})}\sum_{t=1}^T g(x_{t-1})\varepsilon_t
&\to_d \int_0^1 g_H(X(r))dW(r),  
\end{aligned}
\end{equation}
where $F_\nu$ and $F_H$ are, respectively, the asymptotic order and the limit homogeneous function of asymptotically homogeneous function $F=f,g$. The reader is referred to Section 3 of \cite{park2001nonlinear} for more detailed discussions about the asymptotics \eqref{pp1} and \eqref{pp2} as well as the precise definitions of the regularly integrable and asymptotically homogeneous functions.\footnote{Similar asymptotic results for a diffusion process can be found in Kim and Park (2017).} Note also that for a near unit root process the asymptotics \eqref{pp1} and \eqref{pp2} remain valid if $X$ is replaced by the limiting Ornstein-Uhlenbeck process of $(x_t)$ (see, e.g., Section 3 of \cite{park2003strong}).

Importantly, the limit distribution of $\sum_{t=1}^T g(x_{t-1})\varepsilon_t$ in \eqref{pp2} is not Gaussian for an asymptotically homogeneous $g^2$ except in some special cases including $g_H(x) = sign(x)$. In particular, the sign function $sign(\cdot)$ is asymptotically homogeneous with $sign_\nu(\lambda) = 1$ for all $\lambda$ and $sign_H(x) = sign(x)$. Since $\int_0^r sign(X(s))dW(s)$ is a Brownian motion by L\'{e}vy's characterization of Brownian motion, we have
\[
\frac{1}{T^{1/2}}\sum_{t=1}^T sign(x_{t-1})\varepsilon_t
\to_d W(1).  
\]

Using the asymptotics \eqref{pp1} and \eqref{pp2}, one may construct a nonlinear IV estimator $\tilde{\beta}(\gamma)$ being asymptotically Gaussian for a proper choice of $\gamma$, i.e., for $\gamma$ square integrable or the sign function. For such $\gamma$, one may use the following test statistic
\[
\tilde\tau(\gamma) = \frac{\sum_{t=1}^T \gamma(x_{t-1})x_{t-1}}{\left(\sum_{t=1}^T \gamma^2(x_{t-1})\right)^{1/2}}\times \tilde{\beta}(\gamma)
\]
to test the null hypothesis of $\beta=0$.\footnote{When $\gamma(\cdot)=sign(\cdot)$, $\tilde\tau(\gamma) = T^{-1/2} \sum_{t=1}^T sign(x_{t-1})y_t$, which is a special case of $\tau(v)$ in \eqref{test-infeasible} with $v_t=1$. }

\begin{proposition}\label{proposition-niv}
Let Assumption \ref{assumption-mds} hold with $v_t = \sigma$ for all $t\geq1$, and let $(x_t)$ be a unit root process. Further assume that the convergences in \eqref{pp1} and \eqref{pp2} hold.
 
(a) Let $\gamma(\cdot)=sign(\cdot)$ or $\gamma^2$ be regularly integrable. Under $\beta=0$,
\[
\tilde\tau(\gamma)\to_d \mathbb{N}(0,\sigma^2)
\]

(b) Let $\gamma(\cdot)=sign(\cdot)$. Under $\beta\neq0$, 
\[
\frac{1}{T}\tilde\tau(\gamma) \to_d \beta \int_0^1 |B(r)|dr.
\]

(c) Let $\gamma^2(x)$ be regularly integrable, and let $x\gamma(x)$ be either asymptotically homogeneous or regularly integrable such that $\int_{-\infty}^{\infty} x\gamma(x)dx\neq0$. Under $\beta\neq0$, we have $\tilde{\tau}(\gamma)\to_p\infty$ and $\tilde{\tau}(\gamma)=o_p(T)$.
\end{proposition}

\begin{proof}
Note that 
\[
\tilde{\tau}(\gamma) 
= \frac{\sum_{t=1}^T \gamma(x_{t-1})y_t}{\left(\sum_{t=1}^T \gamma^2(x_{t-1})\right)^{1/2}}
= \beta \frac{\sum_{t=1}^T \gamma(x_{t-1})x_{t-1}}{\left(\sum_{t=1}^T \gamma^2(x_{t-1})\right)^{1/2}}
+\frac{\sum_{t=1}^T \gamma(x_{t-1})u_t}{\left(\sum_{t=1}^T \gamma^2(x_{t-1})\right)^{1/2}}.
\]
The stated results in the parts (a) and (b) then follow immediately from the convergences \eqref{pp1} and \eqref{pp2} since, in particular for $\gamma(\cdot)=sign(\cdot)$, we have
\[
\sum_{t=1}^T sign^2(x_{t-1})=T(1+o_p(1)),\quad
\frac{1}{T^{3/2}}\sum_{t=1}^T |x_{t-1}|\to_d \int_0^1 |B_r|dr.
\]

We let $\iota(x)=x$. As for Part (c), we first let $\gamma^2$ and $\iota\gamma$ be regularly integrable. Then, under $\beta\neq0$,
\[
\tilde{\tau}(\gamma) = \beta \frac{\sum_{t=1}^T \gamma(x_{t-1})x_{t-1}}{\left(\sum_{t=1}^T \gamma^2(x_{t-1})\right)^{1/2}}(1+o_p(1))
\]
and
\[
\frac{1}{T^{1/4}}\frac{\sum_{t=1}^T \gamma(x_{t-1})x_{t-1}}{\left(\sum_{t=1}^T \gamma^2(x_{t-1})\right)^{1/2}}
\to_d
\frac{L_X(1,0)\int_{-\infty}^\infty \gamma(x)xdx}{\left(L_X(1,0)\int_{-\infty}^\infty \gamma^2(x)dx\right)^{1/2}}
\]
by \eqref{pp1}, from which $\tilde{\tau}(\gamma) = o_p(T)$.

On the other hand, if $\iota\gamma$ is asymptotically homogeneous and $\gamma^2$ is regularly integrable, then
\[
\frac{1}{T^{3/4}(\iota\gamma)_\nu(T^{1/2})}\frac{\sum_{t=1}^T \gamma(x_{t-1})x_{t-1}}{\left(\sum_{t=1}^T \gamma^2(x_{t-1})\right)^{1/2}}
\to_d
\frac{\int_0^1 (\iota\gamma)_H(X_r)dr}{\left(L_X(1,0)\int_{-\infty}^\infty \gamma^2(x)dx\right)^{1/2}}
\]
by \eqref{pp1} and \eqref{pp2}. Since $\gamma^2$ is regularly integrable, $\gamma(\lambda)\lambda^{1/2}=o(1)$, and hence, $(\iota\gamma)(\lambda)=o(\lambda^{1/2})$. It then follows from the construction of the asymptotically homogeneous function that $(\iota\gamma)_\nu(T^{1/2})=o(T^{1/4})$. Therefore, if $\gamma^2$ is regularly integrable and $\iota\gamma$ is either regularly integrable or asymptotically homogeneous, then $\tilde{\tau}(\gamma) =o_p(T)$ as required. 
\end{proof}

According to Proposition \ref{proposition-niv} (a) and (b), one may easily conduct a Gaussian inference using the nonlinear IV estimator with $\gamma(\cdot)=sign(\cdot)$. Moreover, Proposition \ref{proposition-niv} (a) and (c) imply that a similar Gaussian inference can be conducted using a square integrable $\gamma$. However, Proposition \ref{proposition-niv} (b) and (c) imply that under $\beta\neq0$ the divergence rate of $\tilde{\tau}(sign)$ is faster than that of $\tilde{\tau}(\gamma)$ with a square integrable $\gamma$, which implies the test with $\tilde{\tau}(sign)$ tends to have a better power property than the test with $\tilde{\tau}(\gamma)$ in finite samples.

\subsection*{A.2 Stationary Predictor}
For a stationary $x_t$, the asymptotic distribution of the nonlinear IV based test statistic $\tilde{\tau}(\gamma)$ can be obtained easily and is given by a Gaussian distribution under $\beta=0$ when $\mathbb{E}|\gamma^2(x_{t-1})|<\infty$. However, if $x_t$ has a heavy-tailed marginal distribution and $\mathbb{E}|\gamma^2(x_{t-1})|$ is unbounded for a given $\gamma$, then the limit distribution of $\tilde{\tau}(\gamma)$ is generally non-Gaussian. Therefore, the choice of $\gamma$, as in the unit root type predictor, is important in a Gaussian inference relying on the nonlinear IV. Importantly, the Cauchy based test statistic is always asymptotically Gaussian since $\mathbb{E}|\gamma^2(x_{t-1})|=1$, and hence, the result of Proposition \ref{proposition-niv} (a) remains valid. Moreover, it is shown in Theorem \ref{theorem-vol} that the test statistic based on the Cauchy estimator diverges under the alternative hypothesis and its divergence rate is no slower than the usual $\sqrt{T}$ rate for any nontrivial stationary process. As a conclusion of Section 2.3, the Cauchy estimator can be used to construct a robust inference having a Gaussian limit with no significant loss of testing power compared with other nonlinear IV based methods.

\section*{Appendix B: Useful Lemmas}

\begin{lemma}\label{lemma-b2}
Let Assumption \ref{assumption-kernel} hold, and let $f_h(s) = f(s/h)$ for $f=K,K^2$. We have 
\[
\sup_{h\leq r\leq 1}\left|\frac{1}{hT}\sum_{t=1}^T f_h(r-t/T) - \int_0^1 f(s)ds\right| = O(1/(h^2T)).
\]
\end{lemma}

\begin{proof}
We only prove the result for the case $f=K$ since the argument in the case $f=K^2$ is similar. For the proof, we define a function $I_{r,h}:[0,1]\to\{0,1\}$ for $r\in[0,1]$ and $h>0$ as $I_{r,h}(s)=1\{r-h\leq s\leq r\}$. We then write
\[
\sup_{h\leq r\leq 1-h}\sum_{t=1}^T \int_{t-1}^t |K_h(r-t/T) - K_h (r-s/T)| ds
=A_T(r) + B_T(r) + C_T(r),
\]
where
\begin{align*}
A_T(r) 
&=\sum_{t=1}^T \left(\int_{t-1}^t|K_h(r-t/T)-K_h(r-s/T)|ds\right)I_{r,h-1/T}(t/T),\\
B_T(r) 
&=\sum_{t=1}^T \left(\int_{t-1}^t|K_h(r-t/T)-K_h(r-s/T)|ds\right)(1-I_{r,h+1/T}(t/T)),\\
C_T(r) 
&=\sum_{t=1}^T \left(\int_{t-1}^t|K_h(r-t/T)-K_h(r-s/T)|ds\right)(I_{r,h+1/T}(t/T)-I_{r,h-1/T}(t/T)).
\end{align*}

By Assumption \ref{assumption-kernel} (a) and (c), we have
\[
\sup_{h\leq r\leq 1}A_T(r)\leq C \sum_{t=1}^T \int_{t-1}^t \frac{|t-s|}{hT} ds 
\leq C  \frac{1}{h}.
\]
Moreover, $\int_{t-1}^t|K_h(r-t/T)-K_h(r-s/T)=0$ for all $t$ satisfying $I_{r,h+1/T}(t/T)=0$ by Assumption \ref{assumption-kernel} (a), and hence, $\sup_{h\leq r\leq 1}B_T(r)=0$. Moreover, it may be deduced from Assumption \ref{assumption-kernel} (a) that
\[
\sup_{h\leq r\leq 1} \sum_{t=1}^T |I_{r,h+1/T}(t/T)-I_{r,h-1/T}(t/T)|\leq 2,
\]
from which, jointly with Assumption \ref{assumption-kernel} (b), we have $\sup_{h\leq r\leq 1}C_T(r)=O(1)$. The stated result for $f=K$ then follows immediately since
\begin{align}\label{lemma-b2-3}
\frac{1}{hT} \int_0^T K_h(r-s/T)ds = \frac{1}{h}\int_0^1 K_h(r-s)ds = \int_0^1 K(s)ds
\end{align}
for $r\in[h,1]$, by Assumption \ref{assumption-kernel} (a) and the change of variable in integrals. 
\end{proof}

\begin{lemma}\label{lemma-b3}
Let Assumption \ref{assumption-limit} hold.  As $h\to0$,
\[
\sup_{r\in \mathcal{C}_h, 0<h'<h} |\sigma_T^2(r) - \sigma_T^2(r-h')|=o_{a.s.}(1).
\]
\end{lemma}

\begin{proof}[Proof]
We have
\begin{align*}
|\sigma_T^2(r) - \sigma_T^2(r-h')|
\leq |\sigma_T^2(r) - \sigma^2(r)|  + |\sigma_T^2(r-h') - \sigma^2(r-h')| + |\sigma^2(r) - \sigma^2(r-h')|.
\end{align*} 
Note that, under our conventions, $\sup_{0\leq r\leq 1}|\sigma_T^2(r) - \sigma^2(r)|=o_{a.s.}(1)$ since $\sigma_T\to_{d}\sigma$ in Assumption \ref{assumption-limit} holds almost surely on $\mathbf{D}_{\mathbb{R}^+}[0,1]$ endowed with the uniform topology. It follows that the first two terms are of $o_{a.s.}(1)$.

As for the last term, we note that a c\`{a}dl\`{a}g function is uniformly right-continuous on finite closed intervals (see, e.g., Applebaum (2009), pp. 140). It follows that an associated c\`{a}dl\`{a}g function $\sigma_{-}(r)= \sigma(r-)$ is uniformly left continuous, and hence,
\[
\sup_{r\in [h,1], 0<h'<h}|\sigma^2(r-) - \sigma^2(r-h')| =o_{a.s.}(1).
\]
However, $\sigma^2(r-) = \sigma^2(r)$ for all $r\in \mathcal{C}_h$, which completes the proof. 
\end{proof}

\begin{lemma}\label{lemma-b4}
Let Assumptions \ref{assumption-limit} and \ref{assumption-kernel} hold. If $h\to0$ and $h^2T\to\infty$, then for any fixed real number $c\geq0$
\begin{align*}
\sup_{r\in \mathcal{C}_h} |\hat{\sigma}_1^2(r-c/T)-\sigma_T^2(r)|
=o_p(1).
\end{align*}
\end{lemma}

\begin{proof}
We have
\[
\hat{\sigma}_1^2(r-c/T)-\sigma_T^2(r)
=
\frac{\sum_{t=1}^T \left(\sigma_T^2(t/T)-\sigma_T^2(r)\right) K_h(r-c/T-t/T)}{\sum_{t=1}^T K_h(r-c/T-t/T)}
\]
and $\sup_{h\leq r\leq 1}hT/(\sum_{t=1}^T K_h(r-c/T-t/T))=O_p(1)$ by Lemma \ref{lemma-b2} and  Assumption \ref{assumption-kernel} (a). 

To complete the proof, we write
\begin{align*}
\frac{1}{hT}\sum_{t=1}^T \left(\sigma_T^2(t/T)-\sigma_T^2(r)\right) K_h(r-c/T-t/T)
=A_T(r) + B_T(r),
\end{align*}
where 
\begin{align*}
A_T(r)
&=\frac{1}{hT}\sum_{t=1}^{T}  
\left(\sigma_T^2(t/T) 1\{t/T\in(r-h,r]\} - \sigma_T^2 (r)\right)  K_h(r-c/T - t/T)ds,\\
B_T(r)
&=\frac{1}{hT}\sum_{t=1}^{T}  
\sigma_T^2(t/T) \left(1\{(t+c)/T\in[r-h, r]\} - 1\{t/T\in(r-h,r]\}\right)K_h(r-c/T - t/T)ds.
\end{align*}

We can deduce from Lemmas \ref{lemma-b2} and \ref{lemma-b3} that
\begin{align*}
\sup_{r\in \mathcal{C}_h}|A_T(r)|
=\sup_{r\in \mathcal{C}_h, 0<h'<h} |\sigma_T^2(r) - \sigma_T^2(r-h')|\left(\frac{1}{hT}\sum_{t=1}^T K_h(r - t/T)\right)
=o_p(1).
\end{align*}

We note that $\sigma_T$ and $K$ are bounded due to Assumptions \ref{assumption-mds} (a) and \ref{assumption-kernel} (b). Also, we have for a fixed $c\geq0$ and large $T$
\begin{align*}
&\left| 1\left\{\frac{t+c}{T}\in[r-h, r]\right\} - 1\left\{\frac{t}{T}\in(r-h,r]\right\}\right|\\
&\leq 
1\left\{\frac{t}{T}\in \left[r-\frac{c}{T}-h, r-h\right] \right\} + 1\left\{\frac{t}{T}\in\left[r-\frac{c}{T}, r\right]\right\},
\end{align*}
and hence,
\[
\sup_{h\leq r\leq 1}\sum_{t=1}^T\left| 1\{(t+c)/T\in[r-h, r]\} - 1\{t/T\in(r-h,r]\}\right|
\leq 2c+2.
\]
It follows that
\begin{align*}
\sup_{h\leq r\leq 1}|B_T(r)| 
&\leq \frac{1}{hT} \sup_{h\leq r\leq 1}|\sigma_T^2(r) K(r)| 
\sup_{h\leq r\leq 1}\sum_{t=1}^T\left| 1\{(t+c)/T\in[r-h, r]\} - 1\{t/T\in(r-h,r]\}\right|\\
&=O_p(1/(hT)).
\end{align*}
This completes the proof.
\end{proof}

\begin{lemma}\label{lemma-b5}
Let Assumptions \ref{assumption-mds}, \ref{assumption-kernel} and \ref{assumption-error} hold.  If $h\to0$, $h^2T\to\infty$ and $h^\kappa T = O(1)$  for some $\kappa>2$, then
\begin{align*}
\sup_{h\leq r\leq 1} |\hat{\sigma}_2^2(r)|
&=O_p\left(\left( \log (hT)/(hT) \right)^{1/2}T^{2q} \right).
\end{align*}
\end{lemma}

\begin{proof}
Since $\sup_{h\leq r\leq 1}hT/(\sum_{t=1}^T K_h(r-t/T))=O_p(1)$, it suffices to show that
\[
\sup_{h\leq r\leq 1}\left|\sum_{t=1}^T v_t^2 (\varepsilon_t^2 -1) K_h(r-t/T)\right|
=O_p((hT\log (hT))^{1/2} T^{2q}).
\]

To complete the proof, we split the interval $[h,1]$ into $\bar{k}$ intervals of the form $I_k = 1\{ r | r_k \leq r\leq r_{k+1}\}$, where $r_k = h + k h^{\bar{\kappa}}$ for $k=0,\cdots,[(1-h)/h^{\bar{\kappa}}]$ for some $\bar{\kappa} \geq (\kappa + 2)/2$. Then we can write
\begin{align*}
\sup_{h\leq r\leq 1}\left|\sum_{t=1}^T v_t^2 (\varepsilon_t^2 -1) K_h(r-t/T)\right|
&\leq \max_{0\leq k\leq \bar{k}}|S_T(r_k)| + R_T,
\end{align*}
where $S_t(r_k) = \sum_{s=1}^t v_s^2 (\varepsilon_s^2 -1) K_h(r_k-s/T)$ for $t=1,\cdots, T$, and
\begin{align*}
R_T = \max_{0\leq k\leq \bar{k}} \sup_{r\in I_k} \sum_{t=1}^T |v_t^2  (\varepsilon_t^2-1)| \left|K_h(r-t/T) - K_h(r_k-t/T)\right|.
\end{align*}

For $R_T$, we note that for each $k$, $r_{k+1} = r_k + h^{\bar{\kappa}}$ and
\[
\left|K_h(r-t/T) - K_h(r_k-t/T)\right| = \left|K_h(r-t/T) - K_h(r_k-t/T)\right| 1\{r_k - h\leq t/T\leq r_k + h^{\bar{\kappa}}\}
\]
for $r\in I_k$. It follows that
\begin{align*}
R_T
&= \max_{0\leq k\leq \bar{k}} \sup_{r\in I_k} \sum_{t=1}^T |v_t^2  (\varepsilon_t^2-1)| \left|K_h(r-t/T) - K_h(r_k-t/T)\right| 1\{r_k - h\leq t/T\leq r_k + h^{\bar{\kappa}} \}\\
&\leq C h^{\bar{\kappa}-1} \sum_{t=1}^T |v_t^2 (\varepsilon_t^2-1)|1\{r_k - h\leq t/T\leq r_k + h^{\bar{\kappa}} \}\\ 
&\leq C h^{\bar{\kappa}-1} \left(\max_{0\leq k\leq \bar{k}}\sum_{t=1}^T v_t^41\{r_k - h\leq t/T\leq r_k + h^{\bar{\kappa}} \}\right)^{1/2} \left(\sum_{t=1}^T (\varepsilon_t^2-1)^2\right)^{1/2}\\
&\leq O_p(h^{\bar{\kappa}-1} T^{1/2} (hT)^{1/2}),
\end{align*}
where the second line follows from Assumption \ref{assumption-kernel}, the third line holds due to the Cauchy-Schwarz inequality, and the last line follows from the fact that for any $\bar{\kappa}\geq1$
\[
\max_{0\leq k\leq \bar{k}}\sum_{t=1}^T v_t^41\{r_k - h\leq t/T\leq r_k + h^{\bar{\kappa}} \}
\leq 2\sup_{h\leq r\leq 1}\sum_{t=1}^T v_t^4 1\{r - h\leq t/T\leq r\}
=O_p(hT).
\]
Consequently, we have $R_T = O_p((hT)^{1/2})$ since $\bar{\kappa}\geq2$ and $h^{2\bar{\kappa} - 2}T\leq h^{\kappa}T = O(1)$. 

As for $S_T(r_k)$, we note that for each $k=1,\cdots,\bar{k}$, $(S_t(r_k),\mathcal{F}_t)$ is a square integrable martingale since $(\varepsilon_t^2 -1,\mathcal{F}_t)$ is an MDS  and $v_t^2$ is $\mathcal{F}_t$-adapted such that $\varepsilon_t^2 -1$ has a finite second moment and $\sup_{r\in[h,1]}\sum_{t=1}^T v_t^2 1\{r-h\leq t/T\leq r\} =O_p(hT)$ by Assumptions \ref{assumption-mds}. The predictive quadratic variation $\langle S(r_k)\rangle$ and the total quadratic variation $[S(r_k)]$ of $S(r_k)$ are respectively given by
\begin{align*}
\langle S(r_k)\rangle_T 
&= \sum_{t=1}^T v_t^4E[(\varepsilon_t^2 - 1)^2|\mathcal{F}_{t-1}] K_h^2(r_k - t/T),\\
[S(r_k)]_T 
&=\sum_{t=1}^T v_t^4 (\varepsilon_t^2 - 1)^2 K_h^2(r_k - t/T).
\end{align*}
Clearly, $\max_{0\leq k\leq \bar{k}}\langle S(r_k)\rangle_T = O_p(hT)$ and $\max_{0\leq k\leq \bar{k}} [S(r_k)]_T =O_p(hT T^{4q})$. It follows that $\max_{0\leq k\leq \bar{k}}\left|\langle S(r_k)\rangle_T  + [S(r_k)]_T\right| = O_p(hTT^{4q})$ which implies that for any $\delta$, we can find a constant $M>0$ such that
\[
P\left(\max_{0\leq k\leq \bar{k}}\left| \langle S(r_k) \rangle_T + [S(r_k)]_T\right| < M(hTT^{4q})\right)\leq \delta.
\]
Consequently, we have
\begin{align*}
&P\left(\max_{0\leq k\leq \bar{k}}|S_T(r_k)|>M(hTT^{4q}\log(hT))^{1/2}\right)\\
&\leq P\left(\max_{0\leq k\leq \bar{k}}|S_T(r_k)|>M(hTT^{4q}\log(hT))^{1/2}, \max_{0\leq k\leq \bar{k}}\left|\langle S(r_k) \rangle_T + [S(r_k)]_T\right| < M(hTT^{4q})\right) + \delta\\
&\leq \sum_{k=0}^{\bar{k}}P\left(|S_T(r_k)|>M(hTT^{4q}\log(hT))^{1/2}, \langle S(r_k) \rangle_T + [S(r_k)]_T < M(hTT^{4q})\right) + \delta\\
&\leq \frac{2}{h^{\bar{\kappa}}}\exp\left(- M\log(hT)/2\right) + \delta,
\end{align*}
where the last line follows from the two-sided martingale exponential inequality (see, e.g., Theorem 2.1 of Bercu and Touati (2008)). Moreover, $h^{-\bar{\kappa}}\exp\left(- M\log(hT)/2\right)
=(h^{\bar{\kappa}} (hT)^{M/2})^{-1} \to 0$ for all $M\geq 2\bar{\kappa}$ as $h^2T\to\infty$. This completes the proof.
\end{proof}

\begin{lemma}\label{lemma-b6}
Let Assumptions \ref{assumption-mds}, \ref{assumption-kernel} and \ref{assumption-ols} hold. If $h\to0$ and $h^2T\to\infty$, then
\begin{align*}
\sup_{h\leq r\leq 1} |\hat{\sigma}_3^2(r)|
=O_p\left(T^{2p}/(hT)\right).
\end{align*}
\end{lemma}

\begin{proof}[Proof]
By Lemma \ref{lemma-ols}, we have $(\hat{\beta}-\beta)^2 = O_p\left(T^{2p} \left(\sum_{t=1}^T x_{t-1}^2\right)^{-1}\right)$. Moreover, we may  deduce from Lemma \ref{lemma-b2} with Assumptions  \ref{assumption-kernel}  that for some $0<M<\infty$
\[
\sup_{h\leq r\leq 1}\left|\frac{\sum_{t=1}^T x_{t-1}^2 K_h(r-t/T)}{\sum_{t=1}^T K_h(r-t/T)}\right| 
\leq  \frac{M}{hT}\left(\sum_{t=1}^T x_{t-1}^2\right)(1+o_p(1)),
\]
from which we have the stated result.
\end{proof}

\begin{lemma}\label{lemma-b7}
Let Assumptions \ref{assumption-mds}, \ref{assumption-kernel} and \ref{assumption-ols} hold. If $h\to0$ and $h^2T\to\infty$, then
\begin{align*}
\sup_{h\leq r\leq 1} |\hat{\sigma}_4^2(r)|
=O_p\left(T^{2p}/(hT)\right).
\end{align*}
\end{lemma}

\begin{proof}[Proof]
As in the proof of  Lemma \ref{lemma-b6}, we have $\sup_{h\leq r\leq 1}|\sum_{t=1}^T K_h(r-t/T)|^{-1} = O_p((hT)^{-1})$ and $\hat{\beta}-\beta = O_p\left(T^p \left(\sum_{t=1}^T x_{t-1}^2\right)^{-1/2}\right)$. Moreover, we have
\[
\sup_{h\leq r\leq 1}\left|\sum_{t=1}^T x_{t-1} u_t K_h(r-t/T)\right|
=O_p\left(T^p \left(\sum_{t=1}^T x_{t-1}^2\right)^{1/2}\right)
\]
by Assumption \ref{assumption-ols}. This completes the proof.

\end{proof}

\begin{lemma}\label{lemma-b8}
Let Assumptions \ref{assumption-mds}, \ref{assumption-kernel} and \ref{assumption-ols} hold. If $h\to0$, $h^2T\to\infty$ and $h^\kappa T=O(1)$ for some $\kappa>2$, then
\[
\sup_{h\leq r\le 1}\frac{1}{\hat{\sigma}^2(r)},\quad \sup_{h\leq r\le 1}\frac{1}{\hat{\sigma}_1^2(r) + \hat{\sigma}_2^2(r)}= O_p(1).
\]
\end{lemma}

\begin{proof}
It follows from Lemmas \ref{lemma-b5}-\ref{lemma-b7} that
\[
\frac{1}{\hat{\sigma}(r)} = \frac{\sum_{t=1}^T K_h(r-t/T)}{\sum_{t=1}^T \sigma_T^2(t/T)K_h(r-t/T)} +o_p(1)
\]
uniformly in $r\in[h,1]$. However, $\sigma_T^2(r)\geq \underline{v}>0$ by Assumption \ref{assumption-mds}, and hence, $\sup_{h\leq r\le 1} 1/\hat{\sigma}^2(r) \geq \underline{v} + o_p(1) = O_p(1)$. Similarly, we can show that $\sup_{h\leq r\le 1} 1/(\hat{\sigma}_1^2(r) + \hat{\sigma}_2^2(r)) = O_p(1)$.
\end{proof}

\begin{lemma}\label{lemma-limit}
Under Assumptions \ref{assumption-mds} and \ref{assumption-limit}, 
\[
\left( W_T, \sigma_T, \int \sigma_T(r) dW_T(r)\right)
\to_d
\left( W, \sigma, \int \sigma(r) dW(r)\right)
\]
in $\mathbf{D}_{\mathbb{R}\times \mathbb{R}^+\times  \mathbb{R}}[0,1]$.
\end{lemma}

\begin{proof}[Proof]
The lemma follows from Theorem 4.6 of \cite{KP} (see also Theorem 2.1 of \cite{hansen-1992}.
\end{proof}

\section*{Appendix C: Proofs of the Main Results}

\begin{proof}[Proof of Lemmas \ref{lemma-cauchy}]
The stated result follows immediately from Lemma \ref{lemma-limit}. 
\end{proof}

\begin{proof}[Proof of Lemma \ref{lemma-ols}]
The stated result follows immediately from Assumptions \ref{assumption-mds} and \ref{assumption-ols}.
\end{proof}

\begin{proof}[Proof of Proposition \ref{proposition-vol}]
Assumption \ref{assumption-volest} implies that $h^2T\to\infty$ and $h^\kappa T=O(1)$ for some $\kappa>2$. The stated results then follow immediately from Lemmas \ref{lemma-b4}-\ref{lemma-b7}.
\end{proof}

\begin{proof}[Proof of Theorem \ref{theorem-vol}]
We define $\tilde{\sigma}$ by $\tilde{\sigma}(r) = \hat{\sigma}_1(r-1/T)+\hat{\sigma}_2(r-1/T)$ for $r\in[h,1]$ and $\tilde{\sigma}(r)=\sigma_T(r)$ for $r\in[0,h)$. We also define $\bar{\sigma}$ as 
\[
\bar{\sigma}(r) = \tilde{\sigma}(r)1\{\sigma(s)=\sigma(s-), s\in (r-h,r]\} + \sigma_T(r)1\{\sigma(s)\neq \sigma(s-), s\in (r-h,r]\}
\]
for $r\in[h,1]$, and $\bar{\sigma}(r) = \sigma_T(r)$ for $r\in[0,h)$.

Clearly, $\bar{\sigma}(r)$ is a $\mathcal{F}_{Tr}$-adapted c\`{a}dl\`{a}g process such that $\bar\sigma\to_d \sigma$ since 
\[
\sup_{0\leq r\leq 1}|\bar\sigma^2(r) - \sigma_T^2(r)|
\leq \sup_{r\in \mathcal{C}_h}|\bar\sigma^2(r) - \sigma_T^2(r)|
\leq \sup_{r\in \mathcal{C}_h}|\hat\sigma_1^2(r-1/T) - \sigma_T^2(r)|
+ \sup_{h\leq r\leq 1}|\hat{\sigma}_2(r)|
\]
which is of $o_{p}(1)$ by Lemmas \ref{lemma-b4} and \ref{lemma-b5}.  It then follows from Lemma \ref{lemma-limit} that under $\beta=0$
\[
\tau(\bar{\sigma}) = \frac{1}{\sqrt{T}}\sum_{t=1}^T \frac{sign(x_{t-1}) y_t}{\bar{\sigma}(t/T)}
=\frac{1}{\sqrt{T}}\sum_{t=1}^T sign(x_{t-1})\varepsilon_t\frac{\sigma_T(t/T)}{\bar{\sigma}(t/T)}
\to_d \mathbb{N}(0,1).
\]

To complete the proof, we show that
\begin{align*}
\tau(\hat{\sigma}) - \tau(\tilde{\sigma}), \,\, \tau(\tilde{\sigma}) - \tau(\bar{\sigma}) = o_p(1).
\end{align*}
We write
\begin{align*}
\tau(\hat{\sigma}) - \tau(\tilde{\sigma})
=A_T-B_T+C_T,
\end{align*}
where
\begin{align*}
A_T
&=
\frac{1}{\sqrt{T}}\sum_{t=1}^{hT} sign(x_{t-1})\varepsilon_t\frac{\sigma_T(t/T)}{\hat{\sigma}((t-1)/T)},\quad
B_T
=
\frac{1}{\sqrt{T}}\sum_{t=1}^{hT} sign(x_{t-1})\varepsilon_t\frac{\sigma_T(t/T)}{\tilde{\sigma}(t/T)},\\
C_T
&=\frac{1}{\sqrt{T}}\sum_{t=hT+1}^T sign(x_{t-1})\varepsilon_t\left(\frac{\sigma_T(t/T)}{\hat{\sigma}((t-1)/T)} - \frac{\sigma_T(t/T)}{\tilde{\sigma}(t/T)}\right).
\end{align*}

To show $A_T, B_T=o_p(1)$, we note that
\begin{align*}
\frac{1}{\sqrt{hT}}\sum_{t=1}^{hT} sign(x_{t-1})\varepsilon_t \sigma_T(t/T),\quad
\frac{1}{\sqrt{hT}}\sum_{t=1}^{hT} sign(x_{t-1})\varepsilon_t 
=O_p(1)
\end{align*}
by Lemma \ref{lemma-limit}. It follows that
\begin{align*}
A_T
&=
\sqrt{h}\frac{1}{\hat{\sigma}(h)}
\left(\frac{1}{\sqrt{hT}}\sum_{t=1}^{hT} sign(x_{t-1})\varepsilon_t\sigma_T(t/T)\right)
=O_p(\sqrt{h})
\end{align*}
since $1/\hat{\sigma}(h) = O_p(1)$ by Lemma \ref{lemma-b8}, and
\begin{align*}
B_T
&=
\sqrt{h}
\left(\frac{1}{\sqrt{hT}}\sum_{t=1}^{hT} sign(x_{t-1})\varepsilon_t\right)
=O_p(\sqrt{h}).
\end{align*}
For $C_T$, we have
\begin{align*}
|C_T|
&\leq \frac{1}{\sqrt{T}}\sum_{t=hT+1}^T |\sigma_T(t/T)\varepsilon_t| \left|\frac{\tilde{\sigma}^2(t/T) - \hat\sigma^2((t-1)/T)}{\hat{\sigma}((t-1)/T)\tilde{\sigma}(t/T)(\tilde{\sigma}(t/T) + \sigma_T(t/T))}\right|\\
&\leq O_p(\sqrt{T})\times \sup_{h\leq r\leq 1}|\tilde{\sigma}^2(r) - \hat{\sigma}^2(r-1/T)| 
\end{align*}
since $\sum_{t=hT+1}^T |\sigma_T(t/T)\varepsilon_t| = O_p(T)$ and $1/\hat{\sigma}(r), 1/\tilde{\sigma}(r) = O_p(1)$ uniformly in $h\leq r\leq 1$ by Lemma \ref{lemma-b8}. However, it follows from Lemmas \ref{lemma-b6} and \ref{lemma-b7} that
\[
\sup_{h\leq r\leq 1}|\tilde{\sigma}^2(r) - \hat{\sigma}^2(r-1/T)| 
\leq \sup_{h\leq r\leq 1}|\hat{\sigma}_3^2(r)| +  \sup_{h\leq r\leq 1}|\hat{\sigma}_4^2(r)| = O_p(T^{2p}/(hT)),
\]
from which we have 
\[
C_T = O_p(T^{2p}/(hT^{1/2})) = o_p(1)
\]
due to Assumption \ref{assumption-volest} (a). Thus,  $\tau(\hat{\sigma}) - \tau(\tilde{\sigma}) = o_p(1)$.

We write $\tau(\tilde{\sigma}) - \tau(\bar{\sigma}) = T^{-1/2}D_T$, where
\begin{align*}
D_s
&=\sum_{t=1}^s \varepsilon_t \left(sign(x_{t-1})\left(\frac{\sigma_T(t/T)}{\tilde{\sigma}(t/T)} - \frac{\sigma_T(t/T)}{\bar{\sigma}(t/T)}\right)1\{\sigma(r)\neq\sigma(r-), r\in (t/T-h, t/T]\}\right)\\
&\equiv \sum_{t=1}^s \varepsilon_t z_t.
\end{align*}
We note that $(D_t,\mathcal{F}_t)$ is a square integrable martingale due, in particular, to Assumption \ref{assumption-mds} and our constructions of $\tilde{\sigma}$ and $\bar{\sigma}$. Moreover, the predictive quadratic variation $\langle D\rangle$ and the total quadratic variation $[D]$ of $D$ satisfy
\[
\langle D\rangle_T = \sum_{t=1}^T z_t^2 = O_p(hT),\quad
[D]_T = \sum_{t=1}^T z_t^2\varepsilon_t^2 = O_p(hTT^{2q})
\]
by Lemma \ref{lemma-b8} and Assumption \ref{assumption-error}. It then follows from the two-sided martingale exponential inequality as in the proof of Lemma \ref{lemma-b5} that $D_T = O_p(T^q\sqrt{hT})$, and hence, $\tau(\tilde{\sigma}) - \tau(\bar{\sigma}) = O_p(T^q \sqrt{h})=o_p(1)$ by Assumption \ref{assumption-volest} (b). 

Now we let $\beta\neq0$. We have
\begin{align*}
\left|\tau(\hat{\sigma})\right|
&=\left|\frac{\beta}{\sqrt{T}}\sum_{t=1}^{T-1}\frac{|x_t|}{\hat\sigma((t-1)/T)} + \frac{\beta}{\sqrt{T}}\sum_{t=1}^{T-1}sign(x_{t-1})\frac{\sigma_T(t)}{\hat\sigma((t-1)/T)}\right|\\
&\geq\left|\frac{\beta}{\sqrt{T}}\sum_{t=1}^{T-1}\frac{|x_t|}{\hat\sigma((t-1)/T)}\right| - \left|\frac{\beta}{\sqrt{T}}\sum_{t=1}^{T-1}sign(x_{t-1})\frac{\sigma_T(t)}{\hat\sigma((t-1)/T)}\right|\\
&=\frac{|\beta|}{\sqrt{T}}\sum_{t=1}^{T-1}\frac{|x_t|}{\hat\sigma((t-1)/T)} + O_p(1)\\
&\geq \frac{1}{\sqrt{T}}\frac{|\beta|}{\underline{v}}\sum_{t=1}^{T-1}|x_t| + O_p(1),
\end{align*}
where the second line follows from the reverse triangle inequality, the third line holds due to Theorem \ref{theorem-vol} (a), and the last line can be deduced from the proof of Lemma \ref{lemma-b8}. The stated result then follows immediately from the condition $\sum_{t=1}^{T-1}|x_t|/\sqrt{T}\to\infty$ in the theorem.
\end{proof}

\begin{proof}[Proof of Corollary \ref{cor-vol}]
We note that Lemma \ref{lemma-ols} holds under Assumption \ref{two factor} as long as Assumption \ref{assumption-ols} holds. Also, we can show the stated results in Proposition \ref{proposition-vol} (b)-(d) under Assumption \ref{two factor} similar to the proofs of Proposition \ref{proposition-vol} (b)-(d). 

To show that Proposition \ref{proposition-vol} (a) holds under Assumption \ref{two factor}, we write for any fixed $c\geq0$
\begin{align*}
\hat{\sigma}_1^2(r-c/T) - \sigma_T^2(r)
= \frac{\sum_{t=1}^T (\sigma_T(t/T)  w_t^2 - \sigma_T(r) )K_h(r-c/T-t/T)}{\sum_{t=1}^T K_h(r-c/T-t/T)}
=A_T(r) + B_T(r),
\end{align*}
where
\begin{align*}
A_T(r)
&=\frac{\sum_{t=1}^T (\sigma_T(t/T)  - \sigma_T(r) )w_t^2 K_h(r-c/T-t/T)}{\sum_{t=1}^T K_h(r-c/T-t/T)},\\
B_T(r)
&=\sigma_T(r) \frac{\sum_{t=1}^T (w_t^2 -1) K_h(r-c/T-t/T)}{\sum_{t=1}^T K_h(r-c/T-t/T)}.
\end{align*}
Similar to the proof of Lemma \ref{lemma-b4}, we can show that $\sup_{r\in \mathcal{C}_h}|A_T(r)| = o_p(1)$. Moreover, by applying an exponential inequality for a strong mixing process (see, e.g., \cite{vogt2012nonparametric}, Theorem 4.1 with $d=0$), we may show under Assumption \ref{two factor} that
\[
\sup_{h\leq r\leq 1}|B_T(r)| = O_p\left((\log T / (hT))^{1/2}\right) = o_p(1).
\]
This shows that Proposition \ref{proposition-vol} (a) holds under Assumption \ref{two factor}. 

As for the validity of Theorem \ref{theorem-vol} under Assumption \ref{two factor}, we consider $\bar{\sigma}$ as defined in the proof of Theorem \ref{theorem-vol}.  It is easy to see that $\bar{\sigma}(r)$ is a $\mathcal{F}_{Tr}$-adapted c\`{a}dl\`{a}g process such that $\bar{\sigma}\to_d \sigma$. Therefore, the validity of Theorem \ref{theorem-vol} under Assumption \ref{two factor} can be shown similarly to the proof of Theorem \ref{theorem-vol}. 
\end{proof}

\section*{Appendix D: Additional Figures}

In this section, we present the finite sample power properties for Models CNST (continuous time)  and Models CNST and ARCH (discrete time). The simulation settings are the same as those in Section 5.2 of the main paper.

Figure \ref{figs1} presents the results on finite sample power properties of the tests for the constant volatility case in continuous time. One can see that the power curves of the Cauchy RT test and our test are very close to each other. The other tests have higher size-adjusted power for the local-to-unit root regressor cases and comparable size-adjusted power for purely non-stationary regressors.

Figures \ref{figs2}-\ref{figs4} present the numerical results on power properties under discrete time settings for all the tests considered except Cauchy RT which is inapplicable in  discrete time settings. Results in the figures are provided for the cases of constant volatility (Figure \ref{figs2}); the ARCH cases with $\alpha=0.5773$ (Figure \ref{figs3}) and $\alpha=0.7325$ (Figure \ref{figs4}).

\begin{landscape}
\begin{figure}[h]%
\begin{center}%
\subfigure[$\bar{\kappa}=0$, $T=5$]{\includegraphics[width=0.30\linewidth]{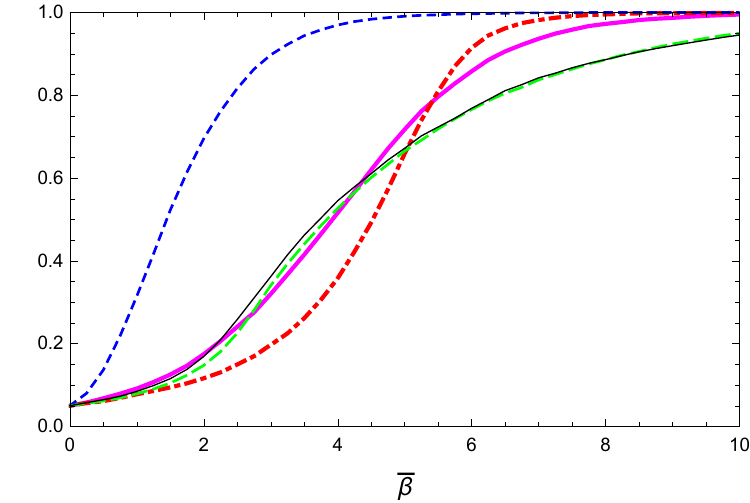}\label{fig:1:1}}
\subfigure[$\bar{\kappa}=0$, $T=20$]{\includegraphics[width=0.30\linewidth]{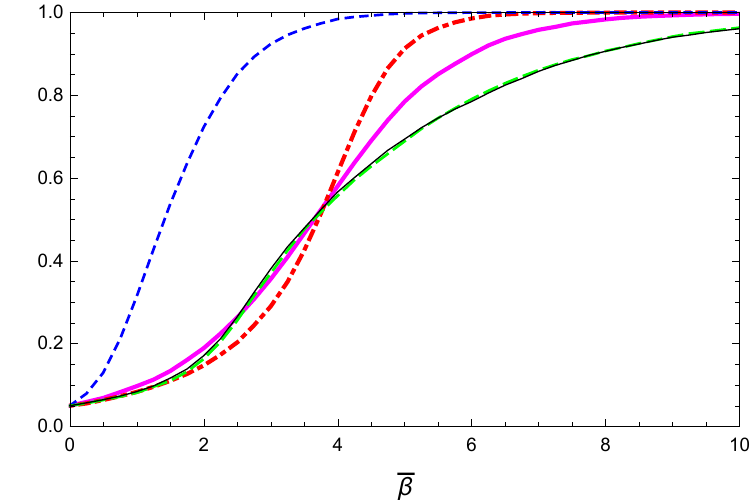}\label{fig:1:2}}
\subfigure[$\bar{\kappa}=0$, $T=50$]{\includegraphics[width=0.30\linewidth]{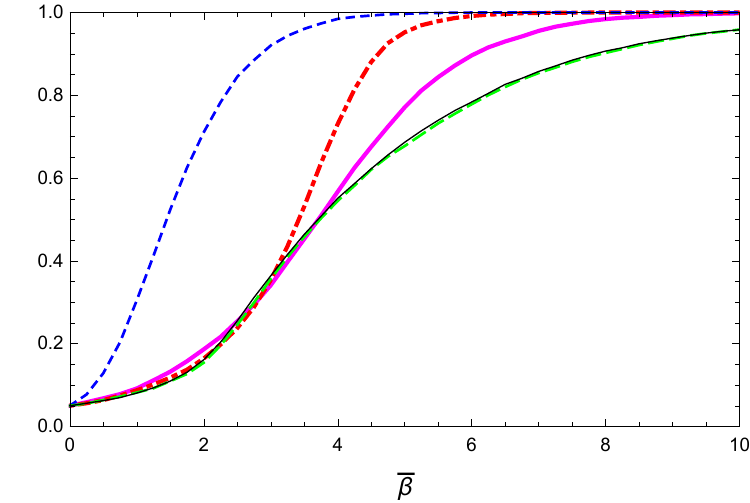}\label{fig:1:3}}\\
\subfigure[$\bar{\kappa}=5$, $T=5$]{\includegraphics[width=0.30\linewidth]{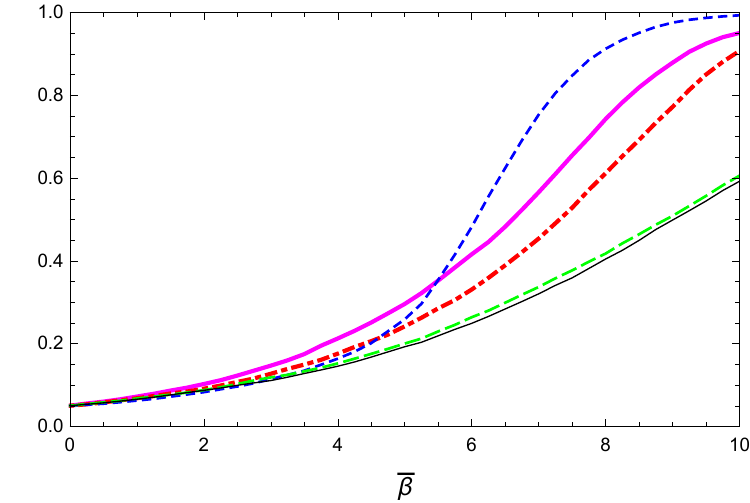}\label{fig:1:4}}
\subfigure[$\bar{\kappa}=5$, $T=20$]{\includegraphics[width=0.30\linewidth]{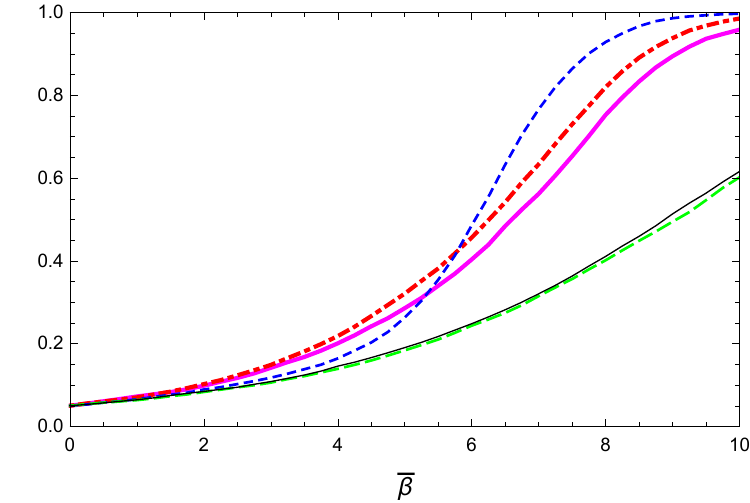}\label{fig:1:5}}
\subfigure[$\bar{\kappa}=5$, $T=50$]{\includegraphics[width=0.30\linewidth]{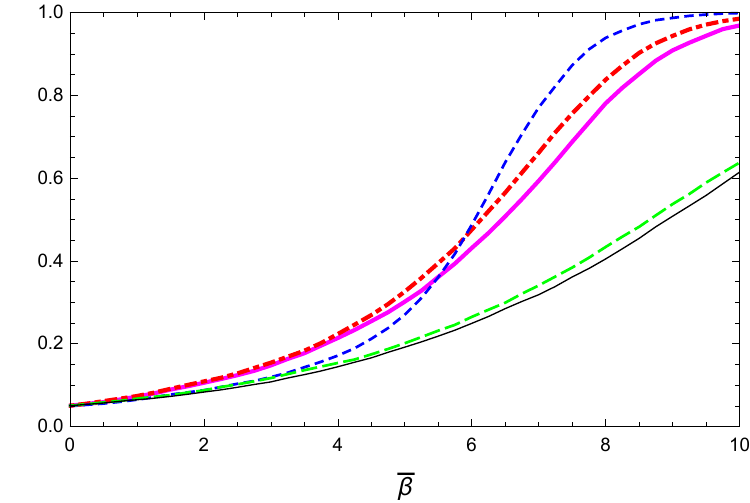}\label{fig:1:6}}\\
\subfigure[$\bar{\kappa}=20$, $T=5$]{\includegraphics[width=0.30\linewidth]{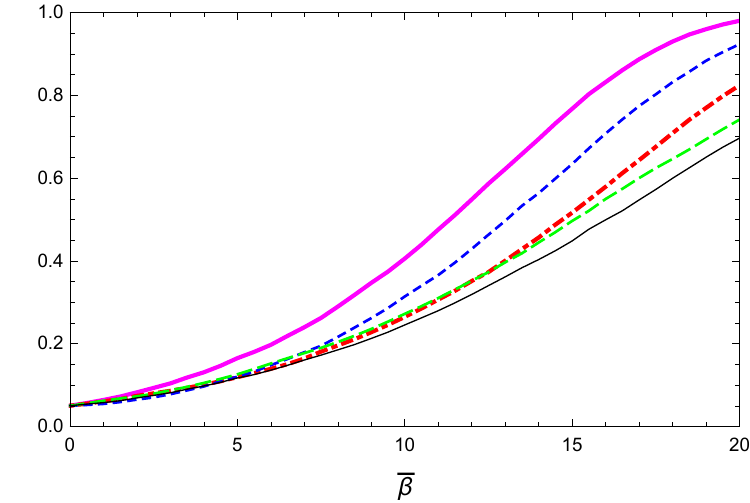}\label{fig:1:7}}
\subfigure[$\bar{\kappa}=20$, $T=20$]{\includegraphics[width=0.30\linewidth]{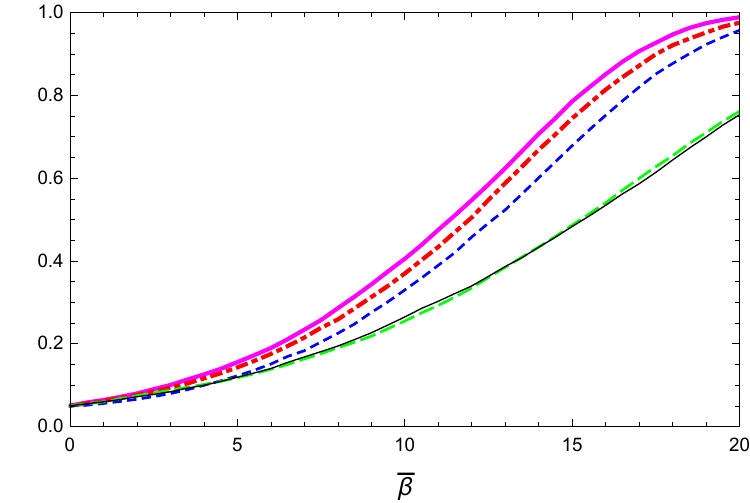}\label{fig:1:8}}
\subfigure[$\bar{\kappa}=20$, $T=50$]{\includegraphics[width=0.30\linewidth]{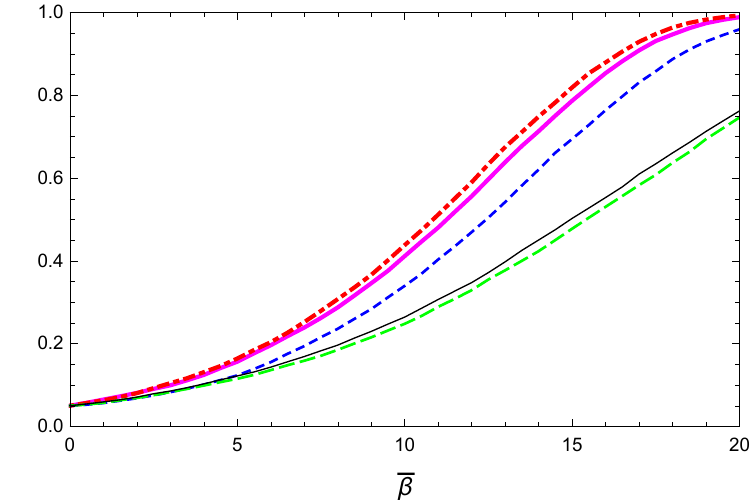}\label{fig:1:9}}
\end{center}%
\caption{Power for CNST (continuous time)}
\label{figs1}
\centering
\footnotesize{
OLS:$\textcolor{magenta}{\rule[0.25em]{2em}{1.6pt}\ }$,
Bonf. Q:$\textcolor{red}{\rule[0.25em]{0.6em}{1.7pt} \ \mathbf{\cdot} \ \rule[0.25em]{0.6em}{1.7pt} \ }$
RLRT:$\textcolor{blue}{\rule[0.25em]{0.4em}{1.6pt} \ \rule[0.25em]{0.4em}{1.6pt}\ }$, 
Cauchy RT:$\textcolor{green}{\rule[0.25em]{0.8em}{1.6pt} \ \rule[0.25em]{0.8em}{1.6pt}\ }$, 
NP:$\textcolor{black}{\rule[0.25em]{1.9em}{0.5pt}}$}
\end{figure}
\end{landscape}

\begin{landscape}
\begin{figure}[h]%
\begin{center}%
\subfigure[$\bar{\kappa}=0$, $T=60$]{\includegraphics[width=0.30\linewidth]{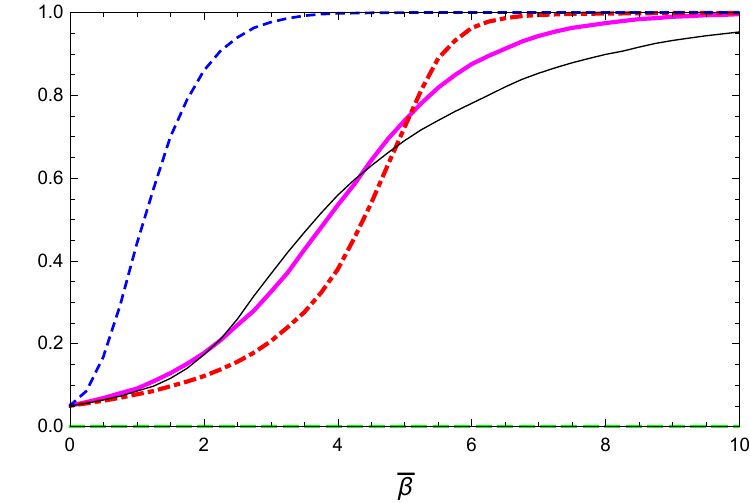}\label{fig:5:1}}
\subfigure[$\bar{\kappa}=0$, $T=240$]{\includegraphics[width=0.30\linewidth]{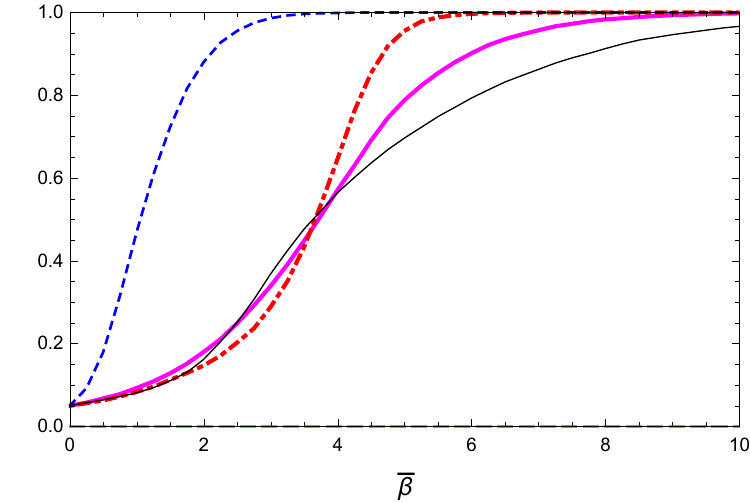}\label{fig:5:2}}
\subfigure[$\bar{\kappa}=0$, $T=600$]{\includegraphics[width=0.30\linewidth]{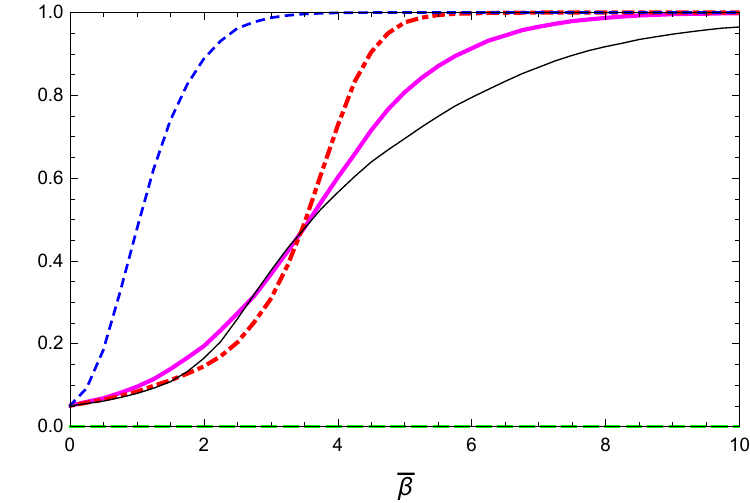}\label{fig:5:3}}\\
\subfigure[$\bar{\kappa}=5$, $T=60$]{\includegraphics[width=0.30\linewidth]{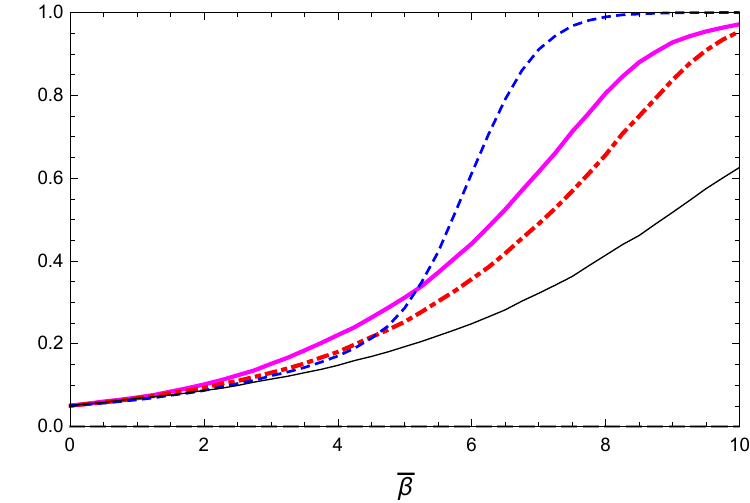}\label{fig:5:4}}
\subfigure[$\bar{\kappa}=5$, $T=240$]{\includegraphics[width=0.30\linewidth]{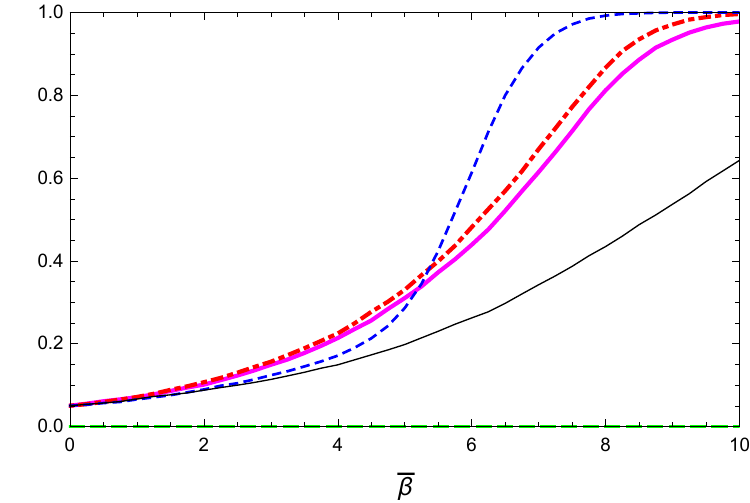}\label{fig:5:5}}
\subfigure[$\bar{\kappa}=5$, $T=600$]{\includegraphics[width=0.30\linewidth]{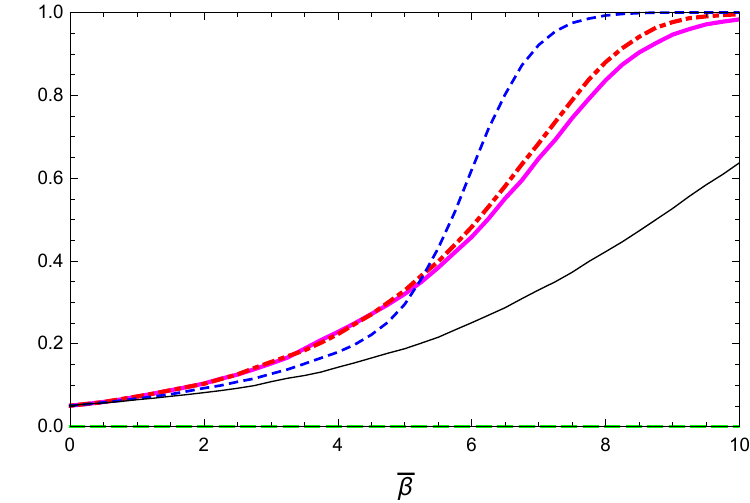}\label{fig:5:6}}\\
\subfigure[$\bar{\kappa}=20$, $T=60$]{\includegraphics[width=0.30\linewidth]{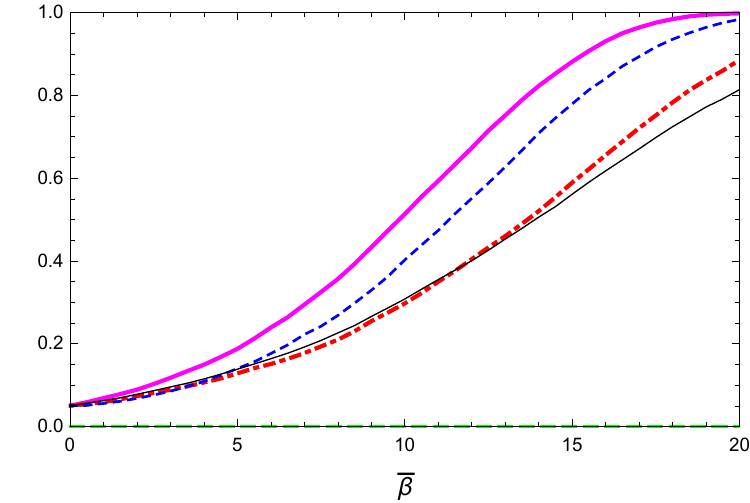}\label{fig:5:7}}
\subfigure[$\bar{\kappa}=20$, $T=240$]{\includegraphics[width=0.30\linewidth]{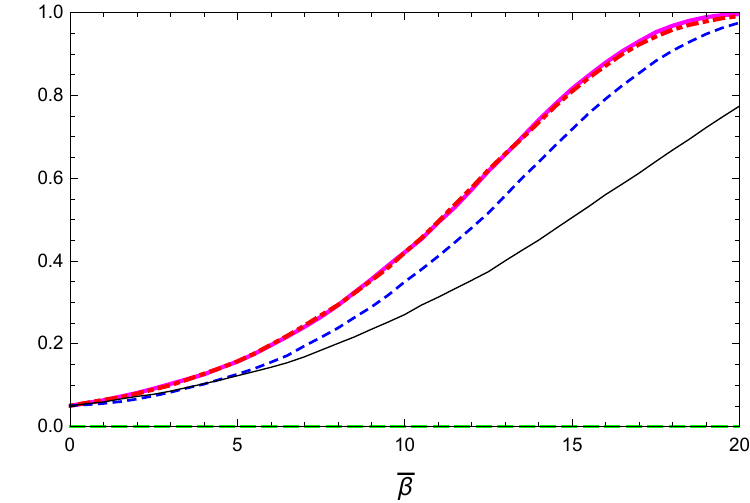}\label{fig:5:8}}
\subfigure[$\bar{\kappa}=20$, $T=600$]{\includegraphics[width=0.30\linewidth]{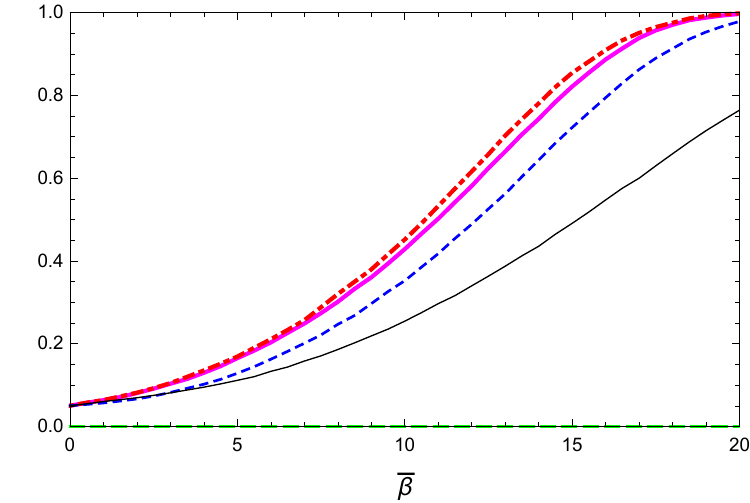}\label{fig:5:9}}
\end{center}%
\caption{Power for CNST (discrete time)}
\label{figs2}
\centering
\footnotesize{OLS:$\textcolor{magenta}{\rule[0.25em]{2em}{1.6pt}\ }$,
Bonf. Q:$\textcolor{red}{\rule[0.25em]{0.6em}{1.7pt} \ \mathbf{\cdot} \ \rule[0.25em]{0.6em}{1.7pt} \ }$
RLRT:$\textcolor{blue}{\rule[0.25em]{0.4em}{1.6pt} \ \rule[0.25em]{0.4em}{1.6pt}\ }$, 
NP:$\textcolor{black}{\rule[0.25em]{1.9em}{0.5pt}}$}
\end{figure}
\end{landscape}

\begin{landscape}
\begin{figure}[h]%
\begin{center}%
\subfigure[$\bar{\kappa}=0$, $T=60$]{\includegraphics[width=0.30\linewidth]{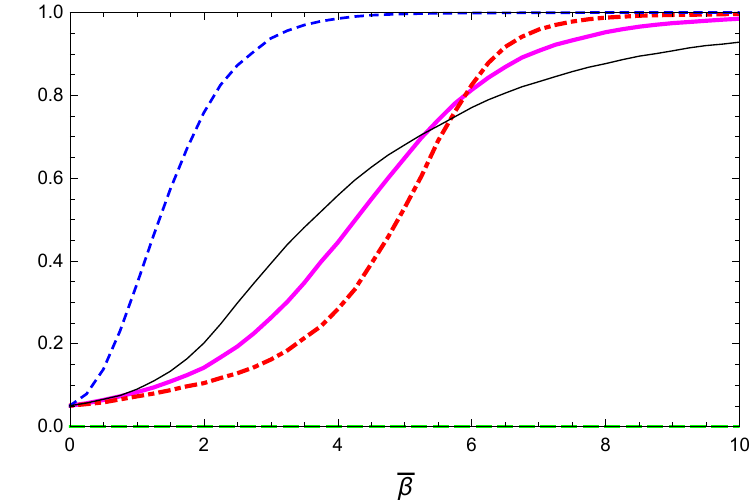}\label{fig:7:1}}
\subfigure[$\bar{\kappa}=0$, $T=240$]{\includegraphics[width=0.30\linewidth]{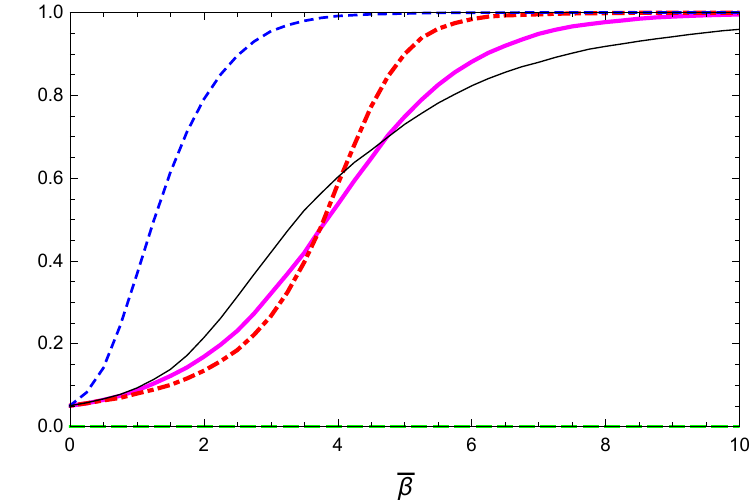}\label{fig:7:2}}
\subfigure[$\bar{\kappa}=0$, $T=600$]{\includegraphics[width=0.30\linewidth]{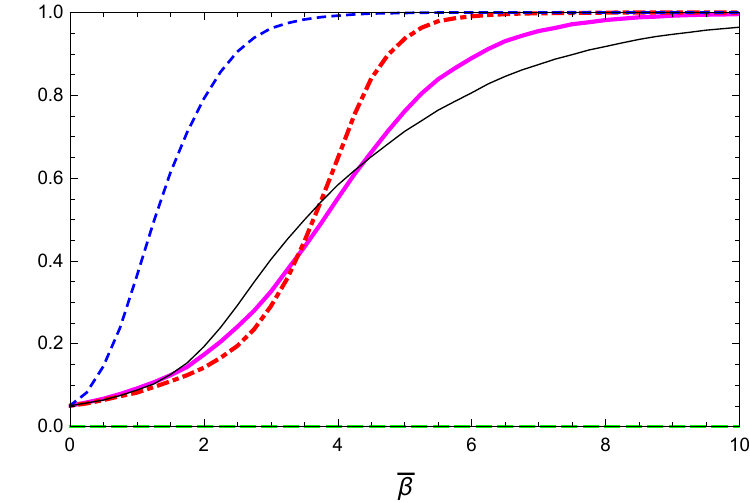}\label{fig:7:3}}\\
\subfigure[$\bar{\kappa}=5$, $T=60$]{\includegraphics[width=0.30\linewidth]{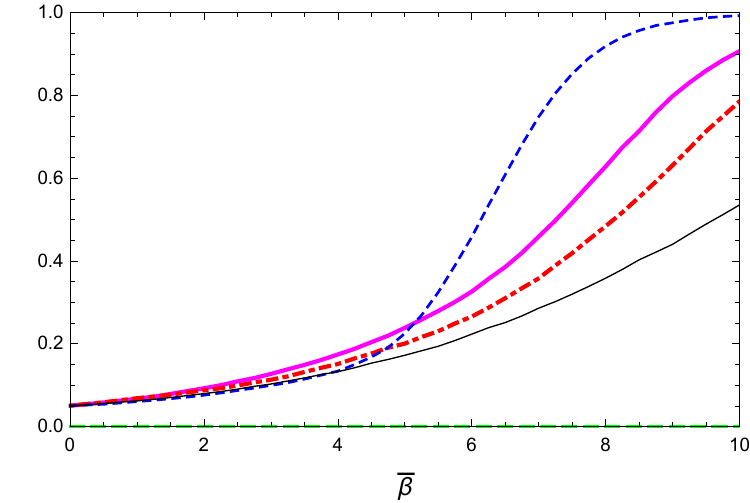}\label{fig:7:4}}
\subfigure[$\bar{\kappa}=5$, $T=240$]{\includegraphics[width=0.30\linewidth]{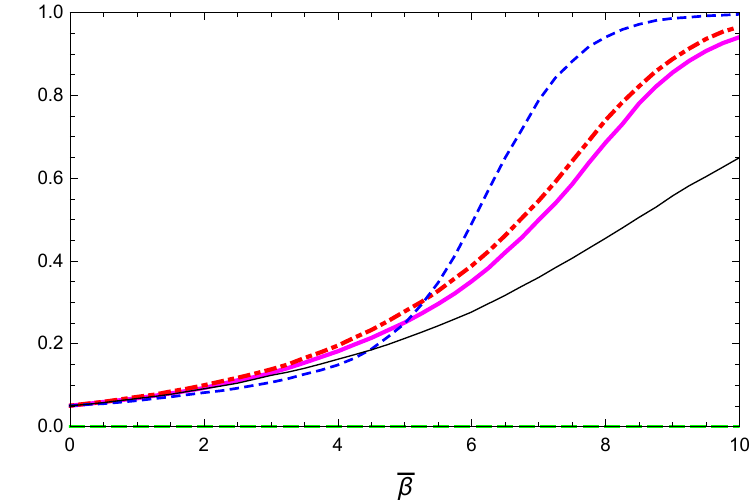}\label{fig:7:5}}
\subfigure[$\bar{\kappa}=5$, $T=600$]{\includegraphics[width=0.30\linewidth]{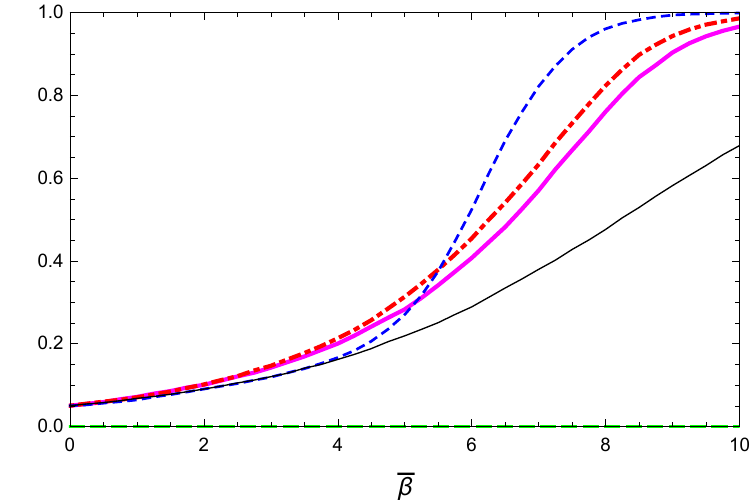}\label{fig:7:6}}\\
\subfigure[$\bar{\kappa}=20$, $T=60$]{\includegraphics[width=0.30\linewidth]{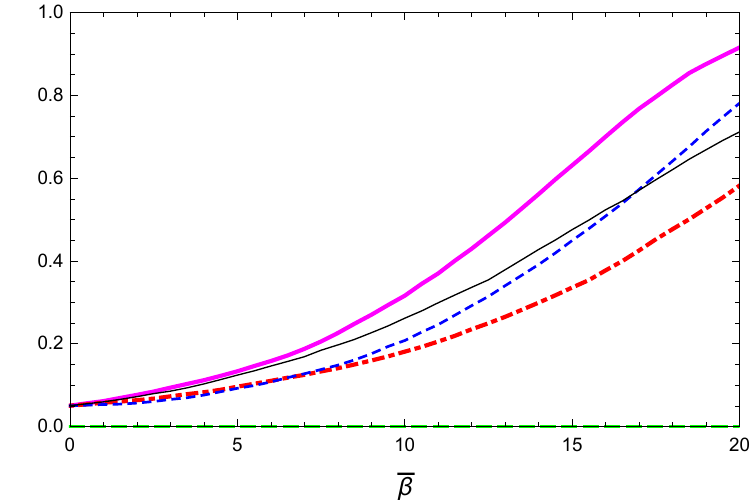}\label{fig:7:7}}
\subfigure[$\bar{\kappa}=20$, $T=240$]{\includegraphics[width=0.30\linewidth]{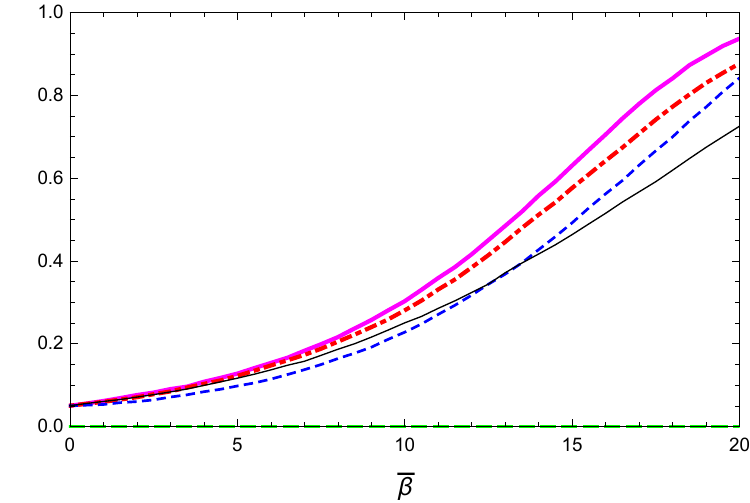}\label{fig:7:8}}
\subfigure[$\bar{\kappa}=20$, $T=600$]{\includegraphics[width=0.30\linewidth]{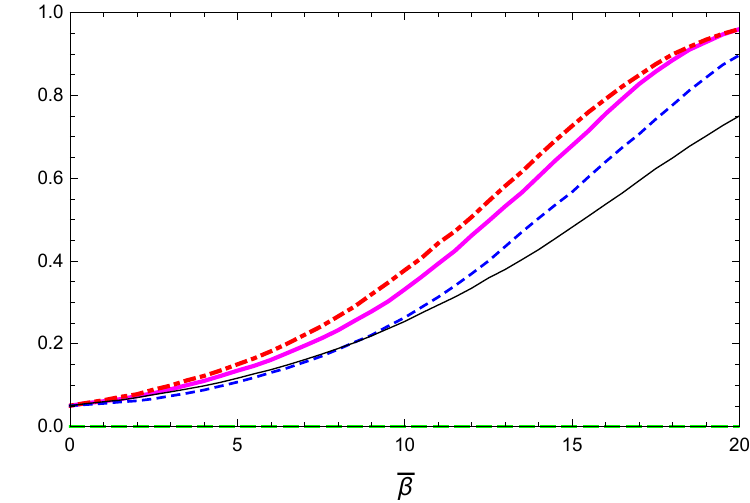}\label{fig:7:9}}
\end{center}%
\caption{Power for ARCH with $\alpha=0.5773$ (discrete time)}
\label{figs3}
\centering
\footnotesize{OLS:$\textcolor{magenta}{\rule[0.25em]{2em}{1.6pt}\ }$,
Bonf. Q:$\textcolor{red}{\rule[0.25em]{0.6em}{1.7pt} \ \mathbf{\cdot} \ \rule[0.25em]{0.6em}{1.7pt} \ }$
RLRT:$\textcolor{blue}{\rule[0.25em]{0.4em}{1.6pt} \ \rule[0.25em]{0.4em}{1.6pt}\ }$, 
NP:$\textcolor{black}{\rule[0.25em]{1.9em}{0.5pt}}$}
\end{figure}
\end{landscape}

\begin{landscape}
\begin{figure}[h]%
\begin{center}%
\subfigure[$\bar{\kappa}=0$, $T=60$]{\includegraphics[width=0.30\linewidth]{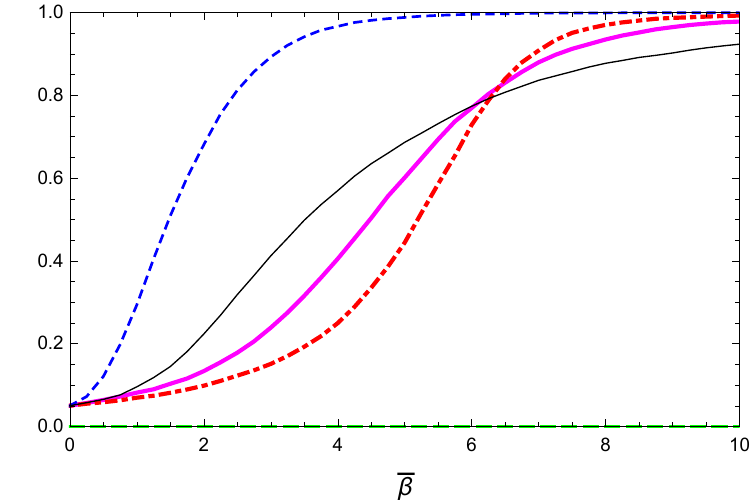}\label{fig:8:1}}
\subfigure[$\bar{\kappa}=0$, $T=240$]{\includegraphics[width=0.30\linewidth]{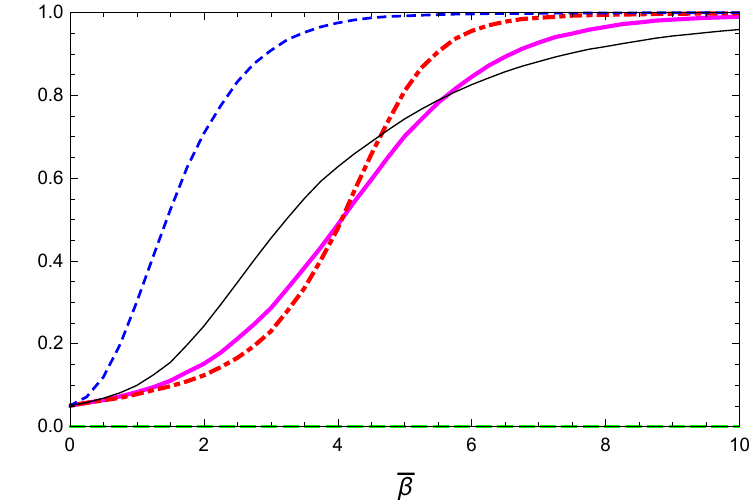}\label{fig:8:2}}
\subfigure[$\bar{\kappa}=0$, $T=600$]{\includegraphics[width=0.30\linewidth]{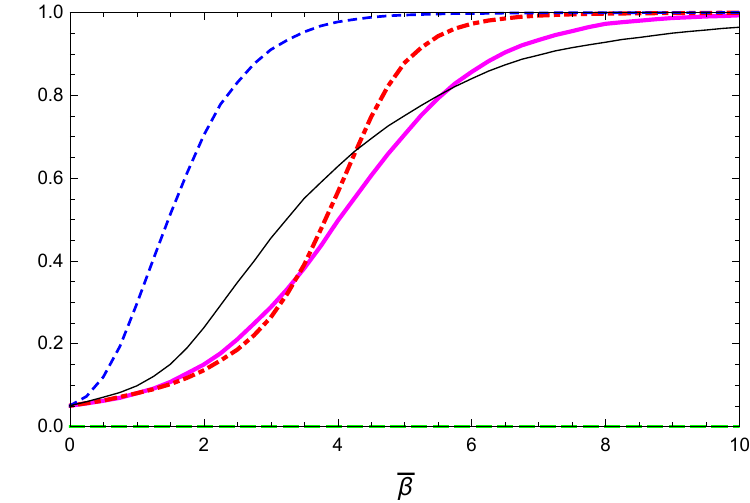}\label{fig:8:3}}\\
\subfigure[$\bar{\kappa}=5$, $T=60$]{\includegraphics[width=0.30\linewidth]{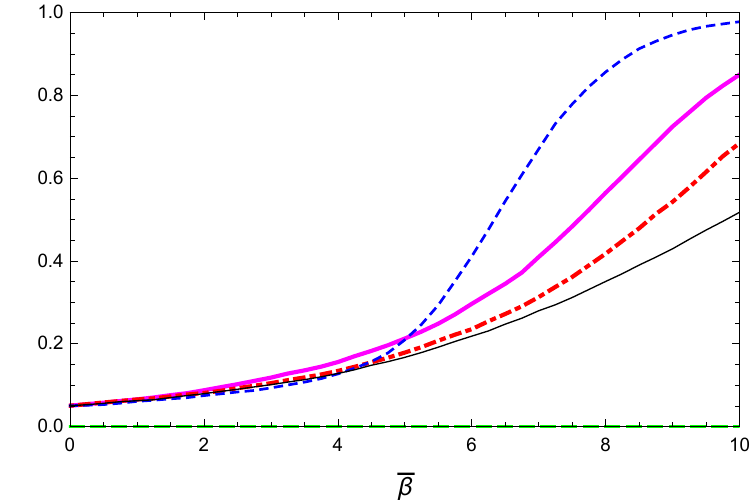}\label{fig:8:4}}
\subfigure[$\bar{\kappa}=5$, $T=240$]{\includegraphics[width=0.30\linewidth]{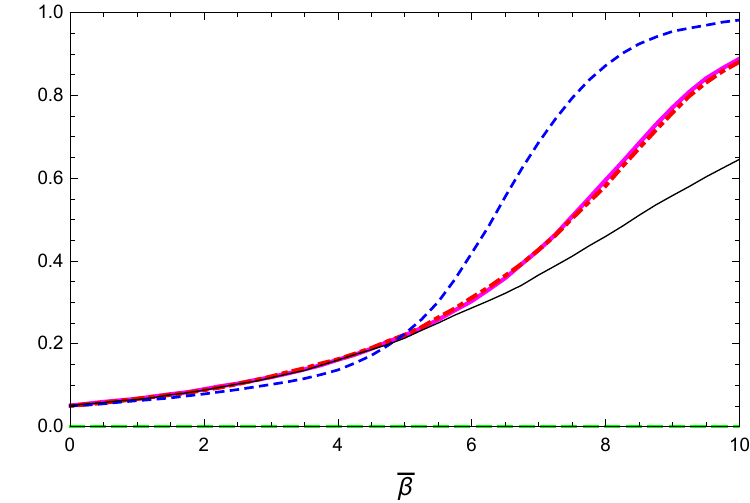}\label{fig:8:5}}
\subfigure[$\bar{\kappa}=5$, $T=600$]{\includegraphics[width=0.30\linewidth]{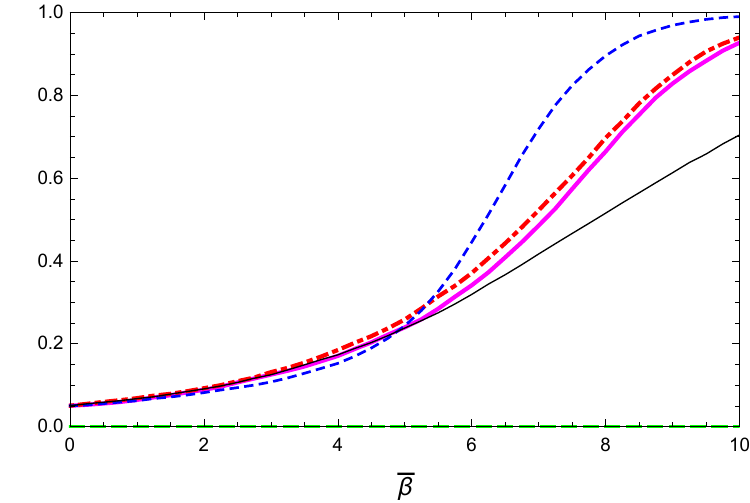}\label{fig:8:6}}\\
\subfigure[$\bar{\kappa}=20$, $T=60$]{\includegraphics[width=0.30\linewidth]{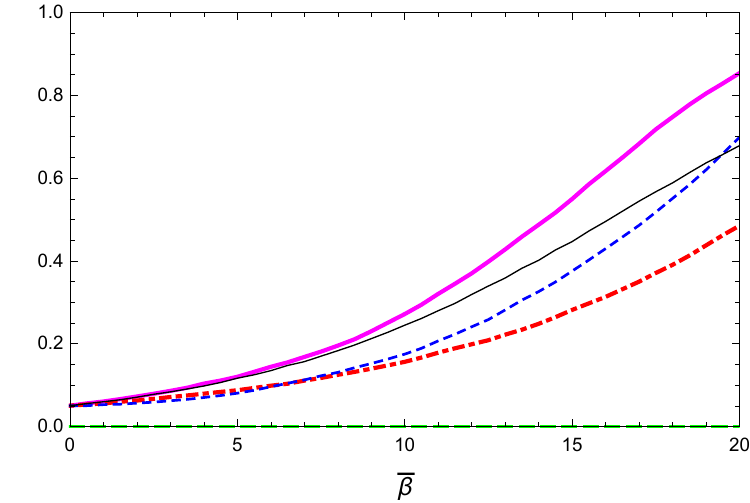}\label{fig:8:7}}
\subfigure[$\bar{\kappa}=20$, $T=240$]{\includegraphics[width=0.30\linewidth]{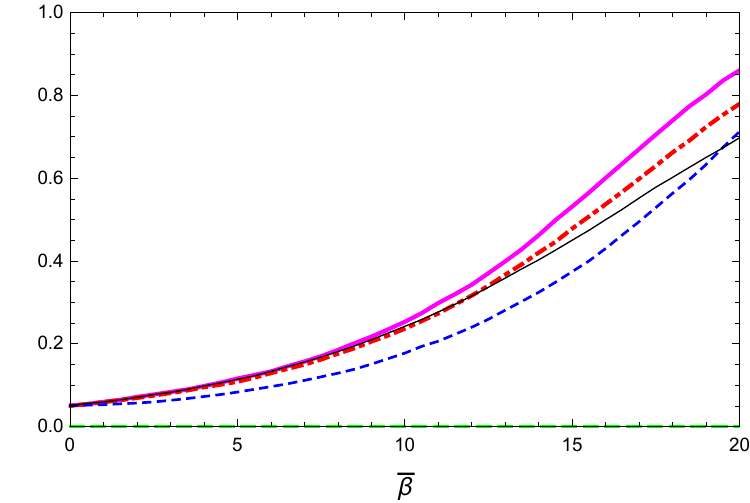}\label{fig:8:8}}
\subfigure[$\bar{\kappa}=20$, $T=600$]{\includegraphics[width=0.30\linewidth]{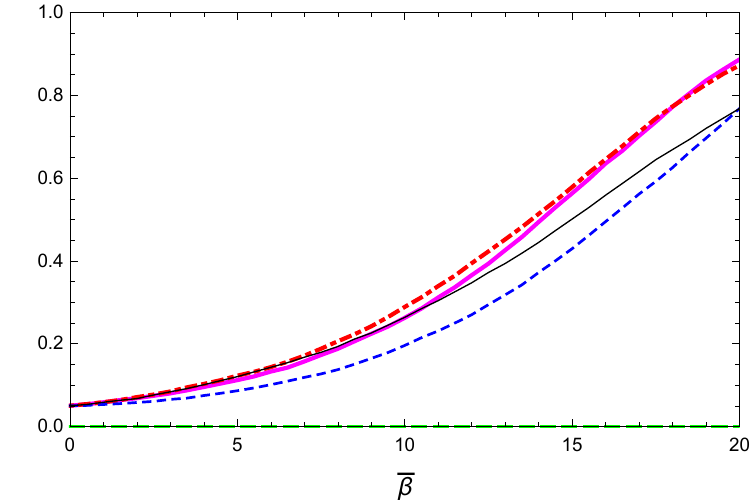}\label{fig:8:9}}
\end{center}%
\caption{Power for ARCH with $\alpha=0.7325$ (discrete time)}
\label{figs4}
\centering
\footnotesize{OLS:$\textcolor{magenta}{\rule[0.25em]{2em}{1.6pt}\ }$,
Bonf. Q:$\textcolor{red}{\rule[0.25em]{0.6em}{1.7pt} \ \mathbf{\cdot} \ \rule[0.25em]{0.6em}{1.7pt} \ }$
RLRT:$\textcolor{blue}{\rule[0.25em]{0.4em}{1.6pt} \ \rule[0.25em]{0.4em}{1.6pt}\ }$, 
NP:$\textcolor{black}{\rule[0.25em]{1.9em}{0.5pt}}$}
\end{figure}
\end{landscape}

\section*{References}

\begin{spacing}{1.5}
\noindent Applebaum, D. (2009), `L\'{e}vy Processes and Stochastic Calculus'. Cambridge University Press.\smallskip

\noindent Bercu, B., and Touati, A. (2008), `Exponential inequalities for self-normalized martingales with applications', \emph{Annals of Applied Probability} \textbf{18}, 1848--1869.\smallskip

\noindent Kurtz, T. G. and Protter, P. (1991), `Weak limit theorems for stochastic integrals and stochastic differential equations', \emph{Annals of Probability} \textbf{19}, 1035--1070.\smallskip

\noindent Hansen, B. E. (1992), `Convergence to stochastic integrals for dependent heterogeneous processes', \emph{Econometric Theory} \textbf{8}, 489--500.\smallskip

\noindent Vogt, M. (2012), `Nonparametric regression for locally stationary time series', \emph{Annals of Statistics} \textbf{40}, 2601--2633.
\end{spacing}

\end{document}